\documentclass[sigconf, nonacm]{acmart}

\newcommand\vldbdoi{XX.XX/XXX.XX}
\newcommand\vldbpages{XXX-XXX}

\newcommand\vldbvolume{15}
\newcommand\vldbissue{13}
\newcommand\vldbyear{2022}

\newcommand\vldbauthors{\authors}
\newcommand\vldbtitle{\shorttitle} 

\newcommand\vldbavailabilityurl{https://github.com/dominatorX/open}

\newcommand\vldbpagestyle{empty} 

\usepackage{bm}
\usepackage{balance}
\usepackage{epstopdf}
\usepackage[linesnumbered,ruled,vlined]{algorithm2e}
\usepackage{subfigure}
\usepackage{enumitem}
\usepackage{soul}
\usepackage{tcolorbox}
\usepackage[normalem]{ulem} 
\newcommand\abs[1]{\left|#1\right|}

\begin{document}

\title{Online Ridesharing with Meeting Points}

\author{Jiachuan Wang}
\affiliation{
	\institution{The Hong Kong University of Science and Technology}
	\city{Hong Kong}
	\country{China}
}
\email{jwangey@cse.ust.hk}
\author{Peng Cheng}
\affiliation{
	\institution{East China Normal University}
	\city{Shanghai}
	\country{China}
}
\email{pcheng@sei.ecnu.edu.cn}
\author{Libin Zheng}
\affiliation{
	\institution{Sun Yat-sen University}
	\city{Guangzhou}
	\country{China}
}
\email{zhenglb6@mail.sysu.edu.cn}
\author{Lei Chen}
\affiliation{
	\institution{The Hong Kong University of Science and Technology}
	\city{Hong Kong}
	\country{China}
}
\email{leichen@cse.ust.hk}
\author{Wenjie Zhang}
\affiliation{
	\institution{The University of New South Wales}
	\country{Australia}
}
\email{wenjie.zhang@unsw.edu.au}

\begin{abstract}
	Nowadays, ridesharing becomes a popular commuting mode. Dynamically arriving riders post their origins and destinations, then the platform assigns drivers to serve them. In ridesharing, different groups of riders can be served by one driver if their trips can share common routes. Recently, many ridesharing companies (e.g., Didi and Uber) further 
	propose a new mode, namely ``ridesharing with meeting points''. Specifically, with a short walking distance but less payment, riders can be picked up and dropped off \emph{around} their origins and destinations, respectively. 
	In addition, \textit{meeting points} enables more flexible routing for drivers, which can potentially improve the global profit of the system. In this paper, we first formally define the Meeting-Point-based Online Ridesharing Problem (MORP). We prove that MORP is NP-hard and there is no polynomial-time deterministic algorithm with a constant competitive ratio for it. We notice that a structure of vertex set, $k$-skip cover, fits well to the MORP. $k$-skip cover tends to find the vertices (meeting points) that are convenient for riders and drivers to come and go. With meeting points, MORP tends to serve more riders with these convenient vertices. Based on the idea, we introduce a convenience-based meeting point candidates selection algorithm.
	We further propose a hierarchical meeting-point oriented graph (HMPO graph), which  ranks vertices for assignment effectiveness and constructs $k$-skip cover to accelerate the whole assignment process. Finally, we utilize the merits of $k$-skip cover points for ridesharing and propose a novel algorithm, namely SMDB, to solve MORP. Extensive experiments on real and synthetic datasets validate the effectiveness and efficiency of our algorithms. 
\end{abstract}

\maketitle
\pagestyle{\vldbpagestyle}
\begingroup\small\noindent\raggedright\textbf{PVLDB Reference Format:}\\
\vldbauthors. \vldbtitle. PVLDB, \vldbvolume(\vldbissue): \vldbpages, \vldbyear.\\
\href{https://doi.org/\vldbdoi}{doi:\vldbdoi}
\endgroup
\begingroup
\renewcommand\thefootnote{}\footnote{\noindent
	This work is licensed under the Creative Commons BY-NC-ND 4.0 International License. Visit \url{https://creativecommons.org/licenses/by-nc-nd/4.0/} to view a copy of this license. For any use beyond those covered by this license, obtain permission by emailing \href{mailto:info@vldb.org}{info@vldb.org}. Copyright is held by the owner/author(s). Publication rights licensed to the VLDB Endowment. \\
	\raggedright Proceedings of the VLDB Endowment, Vol. \vldbvolume, No. \vldbissue\ %
	ISSN 2150-8097. \\
	\href{https://doi.org/\vldbdoi}{doi:\vldbdoi} \\
}\addtocounter{footnote}{-1}\endgroup

\ifdefempty{\vldbavailabilityurl}{}{
	\vspace{.3cm}
	\begingroup\small\noindent\raggedright\textbf{PVLDB Artifact Availability:}\\
	The source code, data, and/or other artifacts have been made available at \url{https://github.com/dominatorX/open}.
	\endgroup
}

\section{Introduction}

Nowadays, on-demand ridesharing becomes important in civil commuting services. Together with online platforms (e.g., DiDi \cite{didi}), ridesharing surpasses traditional taxi services with more saved energy, less air pollution, and lower cost \cite{sharedmobilityweb}.

In \textit{online ridesharing}, riders arrive dynamically. Platforms need to
deal with them immediately for different objectives, including maxmizing the
number of served riders
\cite{cici2015designing,d2012empirical,kleiner2011mechanism,santos2013dynamic,yeung2016flexible},
minimizing the total travel distance \cite{alonso2017demand,
	gupta2010dial,herbawi2012genetic,huang2014large,ma2013t,ota2017stars,rubinstein2012incremental,santos2013dynamic,thangaraj2017xhare},
or maximizing the unified revenue
\cite{asghari2016price,asghari2017line,tong2018unified}.

Ridesharing allows one driver to serve more than one group of riders simultaneously. The route of a driver is a sequence of pick-up/drop-off points. Given a set of drivers and riders, \textit{route planning} is to design and update routes every time a rider arrives.
A key operation, called \emph{insertion}, shows great effectiveness and efficiency for solving online ridesharing problem \cite{ma2013t,huang2014large,thangaraj2017xhare,ota2017stars,santos2013dynamic,yeung2016flexible,cheng2017utility,cici2015designing,tong2018unified}.
It tries to insert a newly coming rider's origin and destination into a driver's route without changing the order of his/her current sequence of pick-up/drop-off points.

\begin{figure}[t!]\centering
	\scalebox{0.25}[0.25]{\includegraphics{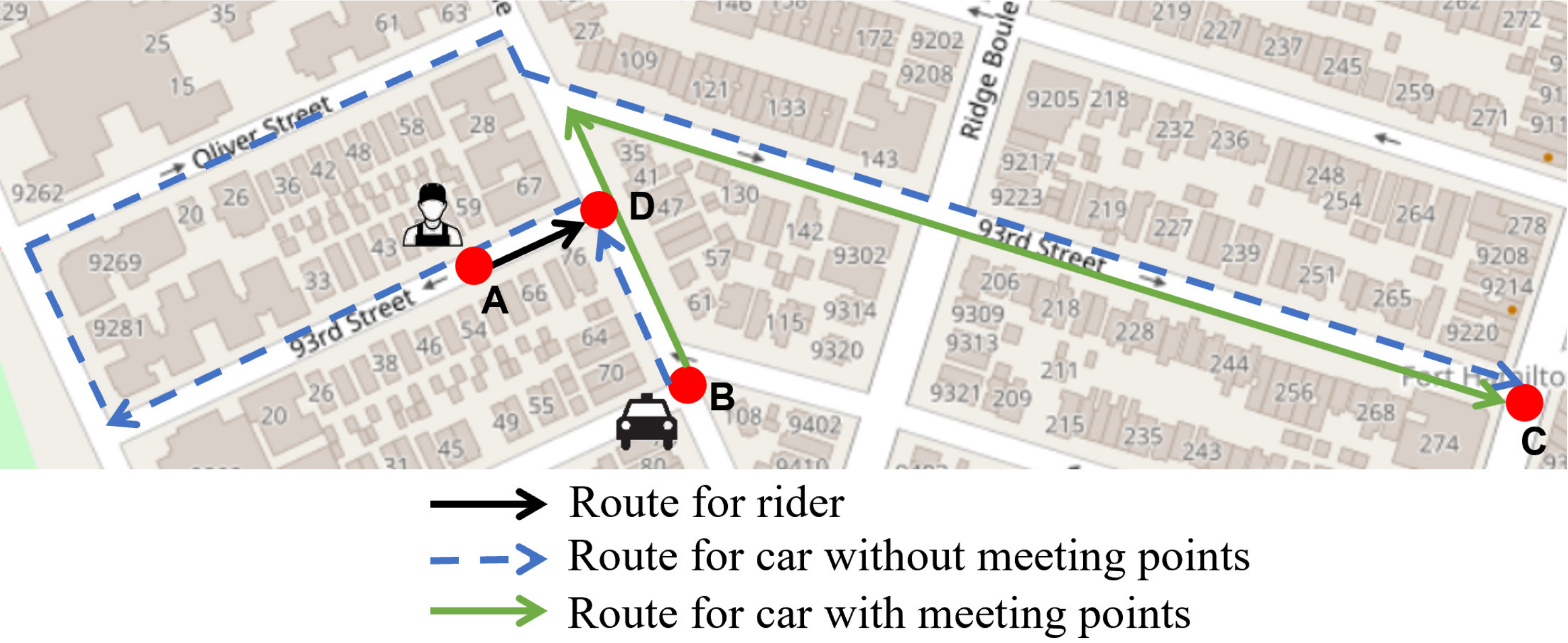}}
	\vspace{-2ex}
	\caption{\small An Example of Meeting Points}\vspace{-5ex}
	\label{fig:example}
\end{figure}

However, due to the complex topology of the city road network, some locations are spatially close to each other but hard
to access for vehicles. Especially, if two locations are only connected by a
Pedestrian Street, where vehicles cannot go through, a short walking could
greatly reduce the travel cost of the assigned vehicle. To deal with the case,
\textit{meeting points (MP for short)} are introduced as alternative locations
for pick-up/drop-off locations of riders~\cite{stiglic_benefits_2015}. As shown
in Figure \ref{fig:example}, a rider $r$ at location \textbf{A} wants to go to
location \textbf{C}. The nearby roads are directed roads. A driver $w$ at
location \textbf{B} is assigned to serve $r$. Then, the shortest route for $w$
is represented in blue dashed arrow lines. If $r$ can move a short distance,
for example, to location \textbf{D} (i.e., an MP), $w$ can serve $r$ through a
much shorter route displayed in the green line.

In a recent work~\cite{zhao_ridesharing_2018}, the authors utilize meeting
points to improve the results of offline ridesharing problems, whose method
however is slow and only can handle up to 40 riders/vehicles in real time. Thus, it is not practical for
online applications (e.g., Uber and DiDi) with hundreds of riders/vehicles
every several seconds. Existing studies also investigate the strategies to
properly select MPs with online
surveys~\cite{czioska2017gis,eser2018tracking}. In industry, Uber recently
offers Express POOL to encourage riders to walk to Express spots (meeting
points) for efficient routing~\cite{ExpressPool}. Nevertheless, Uber Express
Pool only schedules the route for each vehicle when there are shareable
ride-requests to group with, otherwise, the rider needs to wait until other
shareable riders come. In addition, the MPs in Uber Express Pool are
similar to the stops of buses for nearby riders to come together and thus not
flexible \cite{UberExpress}. For the example in Figure \ref{fig:example}, Uber
Express POOL will not assign driver $w$ to pick up rider $r$ until another
rider $r'$ appears close to point {\scriptsize$D$} (i.e., the selected pick-up
stop). In summary, to the best of our knowledge, in the existing research
works, there is no solution for the \textit{online} ride-sharing services
boosted with flexible MPs.

With MPs, online ridesharing is more flexible but challenging. To solve it, we first define \textit{Meeting-Point-based Online Ridesharing Problem}
(MORP) mathematically. Based on
existing studies \cite{asghari2016price,tong2018unified}, we prove that the
MORP problem is NP-hard and has no deterministic algorithms with a constant
competitive ratio, thus intractable.

In the traffic network, some vertices are more convenient to come and go and thus ``popular'' during assignments, such as those close to highways. Flexible MPs makes it possible to serve more riders at or near those vertices, which makes them even more frequently used. This motivates us to take the advantage of $k$-skip cover $V^*$ \cite{DBLP:conf/sigmod/TaoSP11}, which is a subset of vertices to be the skeleton of a graph $G$. %

To solve the MORP problem, we prepare \textit{Meeting point candidates} for each vertex offline. On the other hand, we propose a hierarchical meeting-point
oriented (HMPO) graph, which further filters MPs for effectiveness and accelerate shortest path queries during insertion. Based on the $k$-skip cover in HMPO graph, we propose a meeting-point-based insertion operator, named SMDB, which can solve MORP effectively and efficiently.

Here we summarize our main contributions:
\begin{itemize}[leftmargin=*]
	\item We formulate the online route planning problem with MPs
	mathematically, namely MORP. We prove that it is NP-hard and has no
	algorithm with constant competitive ratio in Section~\ref{sec:problem}.
	\item With observations and analyses, we propose a heuristic algorithm to
	select MP candidates for riders in Section~\ref{sec:MeetingPoints}, which
	is based on a unified cost function considering the travel cost from
	additional walking.  We propose a novel hierarchical structure of the road network, namely hierarchical meeting-point oriented (HMPO) graph,
	to fasten the solution for MORP in Section~\ref{sec:HSGraph}.
	\item With HMPO graph, we propose an effective and efficient insertor, namely SMDB, to handle the requests in MORP in Section
	\ref{sec:hmpo_framework}.
	\item Extensive experiments on synthetic and real data sets show the efficiency and effectiveness of SMDB
	in Section~\ref{sec:experiment}.
\end{itemize}\vspace{-2ex}

\section{Background and Related Works} 
\label{sec:background}
 
\subsection{Online Ridesharing}
Route planning for ridesharing, which has been widely studied in recent years, is a variant of the dial-a-ride problem (DARP) proposed in 1975 \cite{wilson1975advanced,wilson1976advanced}. Traditional DARP problems usually have additional restrictions, such as limiting the drivers to start from/return to depot(s) and serve all the requests \cite{ho2018survey, DBLP:journals/transci/GschwindI15}. These settings lead to small scale datasets with near-optimal solutions. In comparison, route planning for ridesharing is more applicable in the real world, which applies to hundreds of thousands of requests and tens of thousands of drivers with locations distributed over large scale road network \cite{asghari2016price,asghari2017line,huang2014large,tong2018unified}. Realistic revenue and serving cost can be designed as objectives to meet the requirement of ridesharing platforms \cite{asghari2016price,asghari2017line, tong2018unified,zheng2018order}. A common setting for the serving cost is a unified score based on distance/time cost of driving and penalty of rejecting riders. One can further extend the unified cost to an application-specific one, such as maximizing the score combined with complicated social utilities from both workers and requests \cite{cheng2017utility,feuerstein2001line}. Using meeting point results in additional costs such as walking, which is handled with a unified cost function.

\subsection{Insertion}
Real-world ridesharing services require solutions for online instead of off-line mode. Efficient heuristic methods are developed for route planning without information of future workers and requests in advance \cite{asghari2016price,asghari2017line,cici2015designing,huang2014large,ma2013t,ota2017stars,thangaraj2017xhare,yeung2016flexible}. 
With large scale dataset and requirement for real-time response, a commonly 
used operator called insertion shows good performance for route planning 
\cite{cheng2017utility,cici2015designing,huang2014large,ma2013t,jaw1984solving,ma2015real,ota2017stars,rubinstein2012incremental,santos2013dynamic,thangaraj2017xhare,tong2018unified}.
Insertion greatly reduces the search space of possible new routes to serve 
each rider from {\scriptsize$O(N!)$} to {\scriptsize$O(N^2)$}. Tong \emph{et 
	al.} further reduce its time complexity to linear time using dynamic 
programming \cite{tong2018unified, tong2017flexible}. 
We adapt the linear insertion \cite{tong2018unified} for our MORP problem as baseline and further propose a more effective insertor based on a new graph structure. 

\subsection{Ridesharing with meeting points}
As an effective way to improve ridesharing experience, meeting points (MPs) are 
used in online hailing companies, such as Didi and Uber. 
Stiglic \emph{et al.} \cite{stiglic_benefits_2015} first introduce the concept 
of ``meeting points'' to give alternatives to pick up and drop off riders. They 
devise a heuristic algorithm for meeting-point-based offline ridesharing 
problem. Zhao \emph{et al.} \cite{zhao_ridesharing_2018} develop the 
mathematical model for the offline ridesharing problem with flexible pickup and delivery locations and propose an integer 
linear programming model to solve it. In recent years, Uber had proposed 
Express Pool as an online ridesharing service, in which riders need to walk a 
little and may have a longer waiting time, but get a discount. On the other hand, Uber prefers to group passengers together with 
the same MPs, then pick up and drop off them like a bus with selectable 
stations, which has less flexibility \cite{UberExpress}. In this paper, we 
focus on the online ridesharing problem with MPs, which needs to respond to 
requests within a very short time (e.g., within 5 seconds).

\subsection{$k$-skip  cover}\label{subsec:kskip}
Tao et al. \cite{DBLP:conf/sigmod/TaoSP11} first propose $k$-skip cover: given a graph $G(V,E)$, we call a set $V^*\subseteq V$ a $k$-skip cover if for any shortest path $SP$ on $G$ with exactly $k$ vertices, there is at least one vertex $u\in SP$ satisfying $u\in V^*$. In general, for any shortest path $SP$ in $G$, vertices of $SP\cap V^*$ succinctly describes $SP$ by sampling the vertices in $SP$ with a rate of at least $\frac{1}{k}$. Such a sub-path out of the whole path is called a $k$-skip
shortest path. In many applications, such as electronic map presentation, given all vertices are unnecessary and the $k$-skip shortest path gives a good skeleton of it. The study further shows that answering k-skip queries is significantly faster than finding the original shortest paths. 
Funke et al. \cite{DBLP:journals/pvldb/FunkeNS14} further generalize the work of 
\cite{DBLP:conf/sigmod/TaoSP11} by constructing $k$-skip path cover for all 
paths instead of only shortest paths. Besides, they devise a new way of 
constructing a smaller size of $k$-skip path cover. { To construct $k$-skip cover, $k$-\emph{skip neighbor} is defined \cite{DBLP:conf/sigmod/TaoSP11}. If vertices $u$ and $v$ are both in a $k$-skip cover and $u$ is a $k$-\emph{skip neighbor} of $v$, the shortest path from $u$ to $v$ does not pass any other vertices in the $k$-skip cover.}

\vspace{-1ex}
\section{Problem Definition}
\label{sec:problem}


\subsection{Basic Notations}
We use graph $G_c=\langle V_c,E_c\rangle$ to represent a road network for cars,
where  $V_c$ and $E_c$ indicate a set of vertices and a set of edges,
respectively. Each edge, $(u,v)\in E_c(u,v\in V_c)$, is associated with a
weight $t_c(u,v)$ indicating travel time for driving from vertex $u$ to $v$
through it. Similarly, graph $G_p=\langle V_p,E_p\rangle$ is used to represent
a road network for passengers. Each edge $(u,v)\in E_p(u,v\in V_p)$ is weighed
by $t_p(u,v)$ as its travel time for walking. For the two graphs, we denote the
union of vertices as  $V = V_c\cup V_p$. In the city network, passengers are
more flexible. We set all the edges in $V_p$ undirected according to the
network of OSM \cite{OSM}. In addition, usually for any edge $(u,v)$, walking
is slower than driving (e.g., $t_c(u,v) < t_p(u,v)$). We denote \emph{path} as
a sequence of vertices $\{v_1, v_2\cdots, v_k\}$ with travel time
$\sum_{i=1}^{k-1}{t(v_i,v_{i+1})}$. For each pair of vertices $(u,v)$, we
represent the time cost of its shortest path for cars and passengers as
$SP_c(u, v)$ and $SP_p(u, v)$, respectively.

\begin{definition}[Drivers]\vspace{-1ex}
	Let $W=\{w_1,w_2, \cdots,w_n\}$ be a set of $n$ drivers that can provide transportation services. Each driver $w_i$ is defined as a tuple $w_i=\langle l_{i}, a_{i}\rangle$ with a current location $l_{i}$ and a capacity limitation  $a_{i}$.
\end{definition}\vspace{-1ex}
At any time, the number of riders in a taxi of driver $w_i$ must not exceed its capacity $a_{i}$.

\begin{definition} [Requests]
	\label{def:Request}
	Let $R=\left\{r_1,r_2, \cdots, r_m\right\}$ be a set of $m$ requests. 
	Each request 
	$r_j=\langle s_{j},e_{j},tr_{j},tp_j, td_{j},p_j, a_{j}, pi_j, de_j, wp_j, wd_j\rangle$
	is denoted with its source location $s_{j}$, destination location $e_{j}$, release
	time $tr_j$, latest pick-up time $tp_j$, deadline $td_j$, rejection penalty
	$p_j$, and a capacity $a_j$. Once it is assigned, two vertices as pick-up
	point $pi_{j}$ and drop-off point $de_j$ will be recorded. The shortest
	time for a request to walk from source to pick-up point is represented as
	$wp_j=SP_p(s_j, pi_j)$ and from drop-off point to destination is denoted as
	$wd_j=SP_p(de_j, e_j)$.
\end{definition}\vspace{-1ex}

In practice, we do not ask riders to set all the parameters in Definition \ref{def:Request}. Excluding $s_j, e_j,$ and $a_j$, which are given by the rider, other parameters can be auto-filled by the platform to improve the user's experience, such as deadline $td_j$ for reasonable serving time~\cite{huang2014large}. 
A request $r_j$ can be served by driver $w_i$ only if: (a) $w_i$ can arrive at $pi_{j}$ after $tr_j$; (b) the remaining capacity of $w_i$ is at least $a_j$ when he/she arrives at $pi_j$; and (c) $w_i$ can pick $r_j$ at $pi_j$ no later than $tp_j$ and deliver $r_j$ at $de_j$ no later than $td_j-wd_j$.

Note that in real-application, rejections are unavoidable for the ``urgent''
requests on a platform, especially at rush hours. The loss from rejecting $r_j$
is denoted by penalty $p_j$. The penalty can be application-specific.
Furthermore, we denote all the requests that are served by driver $w_i$ as
$R_{w_i}$. { \label{rw:r1-3}Then, $\hat{R}=\cup_{w_i\in W}R_{w_i}$ and $\bar{R}=R \backslash
	\hat{R}$ refer to the total served and unserved requests, respectively.} To
simplify, we will use $r_j$ to indicate a request or a rider of a request
without differentiation.


\begin{definition} [Meeting Points]
	\label{def:Meeting Point}
	For a request $r_j\in R_{w_i}$, the pick-up point $pi_j = u\in V_c\cap V_p$
	denotes that driver $w_i$ will pick up rider $r_j$ at vertex $u$. The
	drop-off point $de_j = v\in V_c\cap V_p$ denotes  that $w_i$ will drop off
	$r_j$ at vertex $v$. Pick-up and drop-off points are \emph{meeting points
		(MP for short)}.
\end{definition}
With MPs, we allow the drivers to flexibly pick up and drop off passengers.
Traditional online ridesharing solutions only assign drivers to pick up a rider
$r_j$ from its source $s_j$ and drop off $r_j$ to its destination $e_j$. With
MPs, the rider $r_j$ can move a short distance to location $pi_{j}$ and be
picked up there by a driver. After being dropped off at location $de_j$, $r_j$
walks to his/her destination $e_j$.

\begin{definition}[Route]
	\label{def:Route}
	The route of a driver $w_i$ located at $l_i$ is a sequence, $S_{w_i}=[l_i, l_{x_1},l_{x_2}, \cdots, l_{x_k}]$, where each { \label{rw:r1-5} $l_{x_n} (n\in[1, k])$} is a pick-up or drop-off point of a request $r_j\in R_{w_i}$ and the driver will reach these locations in the order from $1$ to $k$.
\end{definition}

We call the vertices of a route as \textit{stations}. Drivers move on shortest paths between stations.
A feasible route satisfies: (a) $\forall r_j\in R_{w_i}$, its drop-off time of $de_j$ is no later than $td_j-wd_j$; (b) $\forall r_j\in R_{w_i}$, its pick-up point $pi_j$ appears earlier than its drop-off point $de_j$ in $S_{w_i}$; (c) The total capacity of undropped riders is no larger than the driver's capacity $a_i$ at any time. We denote $S=\{S_{w_i}| w_i\in W\}$ as all the route plans.

Here we define $D({S_{w_i}})$ as the shortest time to finish $S_{w_i}$:\vspace{-2ex}
{\scriptsize\begin{equation}
		D(S_{w_i})=SP_c(l_i,l_{x_1})+\sum_{k=1}^{|{S_{w_i}}|-2}{SP_c(l_{x_k}, l_{x_{k+1}})}\notag
\end{equation}}\vspace{-3ex}

\subsection{Meeting-Point-based Online Ridesharing}\vspace{0.5ex}
\begin{definition}[Meeting-Point-based Online Ridesharing Problem, MORP]
	Given transportation networks {\scriptsize$G_c$} for cars and {\scriptsize$G_p$} for passengers, a set of drivers {\scriptsize$W$}, a set of dynamically arriving requests {\scriptsize$R$}, a driving distance cost coefficient {\scriptsize$\alpha$}, a walking distance cost coefficient {\scriptsize$\beta$}, MORP problem is to find a set of routes {\scriptsize$S=\{S_{w_i}| w_i\in W\}$} for all the drivers with the minimal unified cost:\vspace{-0.5ex}
	{\scriptsize\begin{equation}
			UC(W,R)=\alpha\sum_{w_i\in W}{D(S_{w_i})}+\beta\sum_{r_j\in \hat{R}}{\left(wp_j+wd_j\right)}+\sum_{r_j\in \bar{R}}{p_j}
			\label{eq:obj}
	\end{equation}}\vspace{-2.5ex}
	
	\noindent    which satisfies the following constraints: $(i)$ Feasibility constraint: each driver is assigned with a feasible route; $(ii)$ Non-undo constraint: if a request is assigned in a route, it cannot be canceled or assigned to another route; if it is rejected, it cannot be revoked.
\end{definition}

\subsection{Hardness Analysis}
\begin{figure*}[t!]\centering\vspace{-4ex}
	\scalebox{0.37}[0.37]{\includegraphics{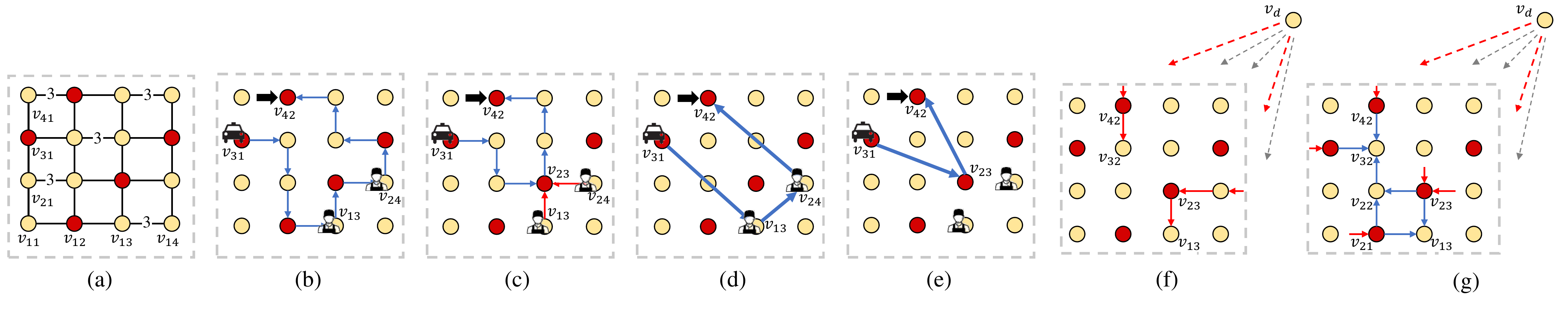}}\vspace{-4ex}
	\caption{\small An Example for the fitness of meeting points and $k$-skip cover. (a) undirected graph $G$ has edges with weight of 1 except for those marked as 3. Red vertices form a $2$-skip cover, which are also good MP candidates. (b) Two requests at $v_{13},v_{24}$ are heading to $v_{42}$. Without meeting points, driver at $v_{31}$ need to serve them with route $v_{31}\rightarrow v_{13}\rightarrow v_{24}\rightarrow v_{42}$. (b) Serving them with MP $v_{23}$ is much more effective, traversing along the popular vertices $v_{31}\rightarrow v_{23}\rightarrow v_{42}$. (d) Traditional ridesharing computes 3 shortest paths $SP(v_{31}, v_{13}),SP(v_{13}, v_{24})$, and $SP(v_{24}, v_{42})$. (e) With MPs, it only computes 2 shortest paths $SP(v_{31}, v_{23})$ and $SP(v_{23}, v_{42})$ beginning and ending with vertices in $V^*$, which can be computed efficiently with $V^*$. (f) $v_d$ is far away from the subgraph, and we want to compare 2 shortest path queries from $v_d$ to $v_{32}$ and to $v_{13}$. Before reaching them, the two paths must reach their surrounding MPs first (e.g., $v_{42}$ and $v_{23}$). (g) $k$-skip cover ``cut off'' shortest paths, so that any shortest paths towards $v_{32}$ and $v_{13}$ must reach $v_{32},v_{42},v_{21}$, and $v_{23}$ first. The distance relationships inside the ``cut'' can help us to bound the distance difference between expensive queries from $v_d$.}
	\label{fig:k-skip}
\end{figure*}
\begin{lemma}\vspace{-1ex}
	\label{proof1}
	The MORP problem is NP-hard.
\end{lemma}\vspace{-3ex}

\begin{proof}
	Please refer to Appendix 
	\ref{A:proof1}.
\end{proof}\vspace{-2ex}

The Competitive Ratio (CR) is commonly used to analyze the online problem. CR
is defined as the ratio between the result achieved by a given algorithm and
the optimal result for the corresponding offline scenario. The existing work
proves no constant CR to maximize the total revenue for basic route planning
for shareable mobility problems with neither deterministic nor randomized
algorithm \cite{asghari2016price,tong2018unified}. Here we have the following
lemma for MORP.

\begin{lemma}\vspace{-1ex}
	\label{proof2}
	There is no randomized or deterministic algorithm guaranteeing constant CP for the MORP problem.
\end{lemma}\vspace{-3ex}
\begin{proof}
	Please refer to Appendix 
	\ref{A:proof2}.
\end{proof}\vspace{-2ex}

\section{Overview of the Framework}
\label{sec:framework}

In this section, we first introduce the $k$-skip cover and how it coincides with the demand of MORP \cite{DBLP:conf/sigmod/TaoSP11}. Then we show the detail of our framework, which makes full use of the $k$-skip cover.

A $k$-skip cover $V^*$ is vertex set on graph $G$. By definition, any shortest path of length $k$ has at least 1 vertex that is $\in V^*$ \cite{DBLP:conf/sigmod/TaoSP11}. A good $V^*$ has small size, such that each of its vertices is frequently passed for transportation.

We claim that $k$-skip cover suits MORP problem well for 2 reasons. Figure \ref{fig:k-skip} is shown as an example: 
\begin{itemize}[leftmargin=*]
	\item For a road network, vertices that are convenient to come and go are good candidates for both a $k$-skip cover $V^*$ and MPs, as they are usually components of many short paths and thus vital for transportation.  
	These convenient vertices are fast for comuting and rider-concentrated, thus ``popular'' during assignments. MPs enable drivers to serve more requests through them, which make them more popular. 
	After {\label{rw:r1-6} constructing} the cover $V^*$, shortest path queries on points of $V^*$ can be computed quickly. 
	This motivates us to build a hierarchical meeting-point oriented graph with $k$-skip cover in Section~\ref{sec:HSGraph}, which encourages more requests to be served effectively through $V^*$ and boost the overall query time cost. 
	
	\item During assignment, we try to insert nearby MP candidates into each worker, where many queries are from same source to different MPs. We claim that $k$-skip cover has underexplored merits to bound the differences between these queries, shown in Section \ref{sec:hmpo_framework}. If one of them is infeasible to insert, the bound makes it possible to prune other MPs, which greatly improve the efficiency. We explore this attribute and devise a new insertion 
	algorithm SMDB in Section \ref{subs:SMDB}.
\end{itemize}

{\label{rw:r1-7} Taking} the advantage of $k$-skip cover, we construct our framework to solve MORP problem, shown in Figure \ref{fig:frame}.

\begin{figure}[h!]\centering\vspace{-2ex}
	\scalebox{0.24}[0.24]{\includegraphics{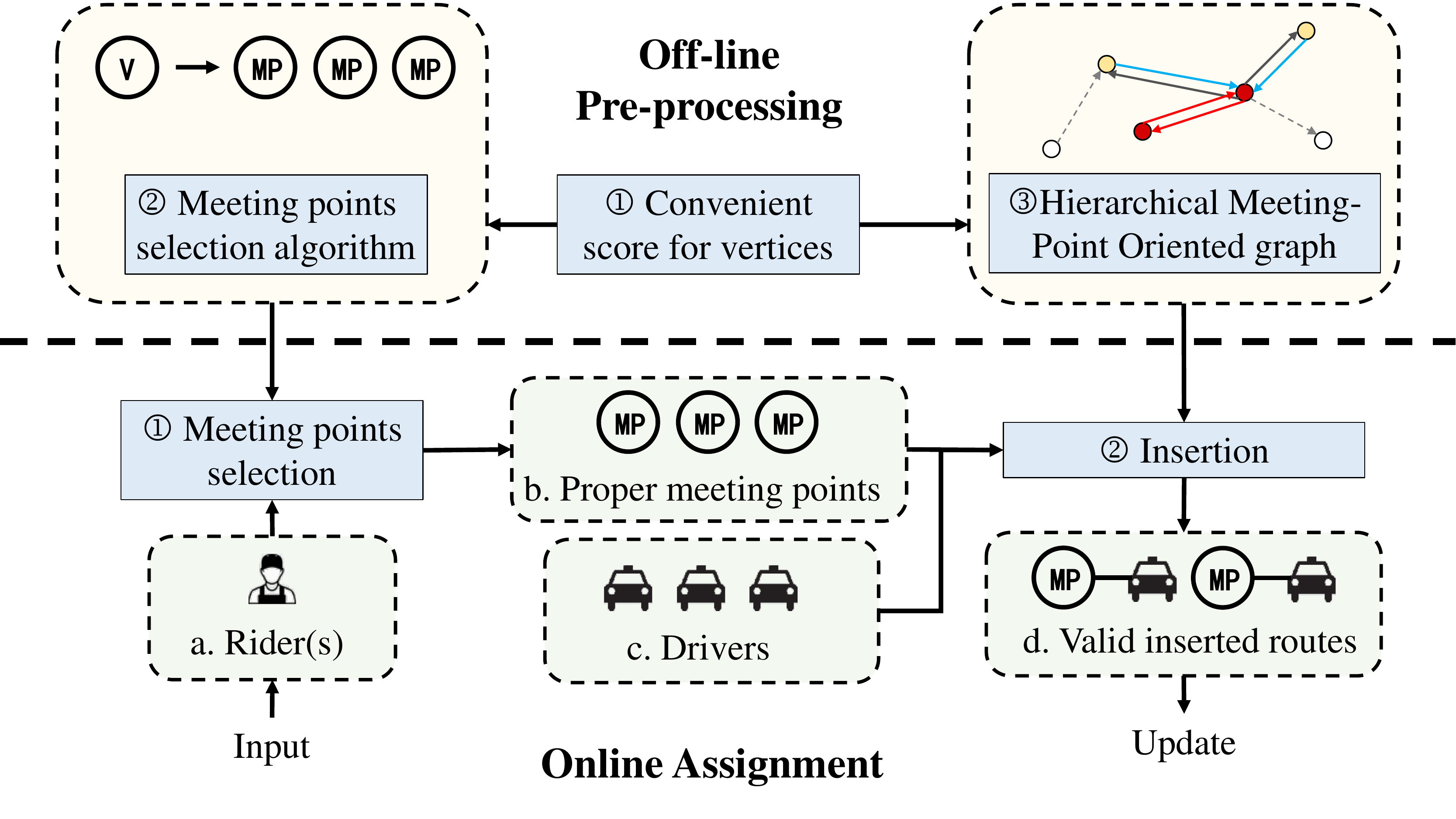}}\vspace{-4ex}
	\caption{\small Assignment framework overview.}
	\vspace{-3ex}
	\label{fig:frame}
\end{figure}

During the online assignment, requests {\label{rw:r1-8}arrive and are assigned} one-by-one. Given a new request $r_j$, we first select meeting points (MP) according to its source and destination locations. Then we iteratively insert each pair of MPs into each driver $w_i$. In our work, we adapt the insertion algorithm with time complexity $O(n)$ for MP insertion \cite{tong2018unified}. If there exists valid insertion(s), we choose the one with the minimal unified cost; otherwise, we reject $r_j$.

Compared with traditional ridesharing problem, we need to select MPs and accelerate related computations. Our work conducts off-line pre-processing to improve the efficiency and effectiveness of online assignment. we first propose a method to evaluate the convenience of each vertex, which bases on statistics of shortest path queries to meet the demand of MPs and $k$-skip cover. 
MP candidates are selected for each vertex in Section~\ref{sec:MeetingPoints}, which greatly shrink search space with $O(1)$ time complexity during online assignment. In addition, we design a structure, namely hierarchical meeting-point oriented graph (HMPO graph), to rank vertices for effective assignments in Section~\ref{sec:HSGraph}. $k$-skip cover is embedded for efficiency. Based on it, we further devise a new insertion algorithm in Section~\ref{sec:hmpo_framework}, namely SMDB, which prunes candidate MPs and drivers during the insertion phase.

\section{Select Meeting Point Candidates}
\label{sec:MeetingPoints}
After Stiglic et al.~\cite{stiglic_benefits_2015} introduced the concept of 
``meeting points'' (MP) to provide flexible pick-up and drop-off points for 
riders, many researchers aim to find an effective solution for ridesharing with 
MPs \cite{stiglic_benefits_2015, czioska2017gis}.
In this paper, we pre-select a set of vertices as candidates to serve their 
nearby vertices.

In this section, we first introduce the motivation of selecting MP candidates. 
Then, we propose a heuristic algorithm, Local-Flexibility-Filter, to select them. 

\subsection{Meeting Point Candidates}
To insert one rider into a route, traditional ridesharing only inserts 1 pair of pick-up and drop-off points. Assume that on average, one vertex has $K$ nearby vertices, which are within the acceptable distance for rider to walk to. Enumerating $K$ pick-up points and $K$ drop-off points as MPs increases the time cost by a factor of $K^2$, which is unacceptable.
Here, we pre-select \textit{Meeting Point Candidates} for each vertex. To insert a pair of origin and destination, 
we can directly get their MPs, instead of searching among all their neighboors.

\begin{definition}(Meeting Point Candidates)
	Given the road network for cars $G_c=\langle V_c,E_c\rangle$ and passengers $G_p=\langle V_p,E_p\rangle$, MP Candidates $MC$ is a dictionary, which maps each 
	vertex $u\in V_p$ to a vertex set 
	$MC(u)=\left\{v_1, v_2, \cdots\right\}\subseteq V_c\cap V_p$. For the MORP 
	problem, we only select MPs for a vertex from its MP candidates, that is, $\forall r_j\in \hat{R}, pi_j\in MC(s_j)$ and $de_j\in 
	MC(e_j)$.
\end{definition}

\subsection{Meeting Point Candidate Selection}
\label{subsec:LFF}
MP candidates should easily get to and conveniently reach other vertices. 
We introduce our Local-Flexibility-Filter algorithm to find the candidate sets $MC(\cdot)$ in two phases.

\textbf{Vertices convenient for drivers}. The first phase aims to find the vertices which are convenient for drivers, thus boosting transportation efficiency.
As the example shown in Figure \ref{fig:k-skip}, MPs and $k$-skip cover have similar preferences. We quantify the convenience from the statistic of shortest path queries, named \textit{equivalent in/out cost} 
{\scriptsize$ECI/ECO$} for each vertex. If a vertex $u$ is inserted into a 
route between $v_1$ and $v_2$, the {\scriptsize$ECI(u)$} and {\scriptsize$ECO(u)$} indicates the average 
cost from $v_1$ to $u$ and from $u$ to $v_2$, respectively. 

As riders are usually assigned to nearby drivers, shortest path queries between a vertex and its surrounding vertices {\label{rw:r1-9}will} indicate its convenience. For each source vertex $u$, we directly select its $n_r$ nearest vertices on the car graph $G_c$ as \emph{reference vertices} $n_o(u)$. Intuitively, we define the equivalent out cost of $u$ as the average distance towards its reference vertices:
$$\scriptsize ECO(u)=\frac{\sum_{v\in 
		n_o(u)}SP_c(u,v)}{n_r}$$

Similarly, for $ECI(u)$, we reverse the graph and select its $n_r$ nearest inward neighboors $n_i(u)$. $\scriptsize ECI(u)=\frac{\sum_{v\in 
		n_i(u)}SP_c(v,u)}{n_r}$
indicates the  average cost of reaching $u$ from other vertices. 
Here, $n_r$ depends on the density of a road network and the speed of drivers. 

\textbf{Vertices convenient for riders}. The second phase takes the walking convenience into account. 
For each vertex $u$, we select vertices $\{v_1,v_2,\cdots\}$ no farther than a maximum walking distance $d_m$. For each reachable vertex $v_i$, we calculate a serving-cost score {\scriptsize$SCS(u,v_i)$} combining both walking distance and the equivalent in/out costs as {\scriptsize$SCS(u,v_i)=\beta\cdot SP_p(u,v_i)+\alpha\left(ECI(v_i)+ECO(v_i)\right)$},
where $\alpha$ and $\beta$ are the weight factors for driving and walking cost defined in MORP.
Especially, each vertex $u$ has the \textit{SCS} score for itself: 
{\scriptsize$SCS(u,u)=\beta\cdot 
	SP_p(u,u)+\alpha\left(ECI(u)+ECO(u)\right)=\alpha\left(ECI(u)+ECO(u)\right)$}.

\textbf{Filtering of MP candidates.} We want a small candidate set for pruning 
effectiveness, while only good MPs are retained. A vertex with a low average cost to serve a rider (low $SCS$) does 
not need alternatives. But a vertex with high $SCS$ needs more choices to find 
a good MP.

Here, for a vertex $u$, we prune all its candidates with a score higher than 
$SCS(u,u)+thr_{CS}$. $thr_{CS}$ is a user-specified threshold. 
We also set an upper bound $nc_m$ for the number of candidates of each vertex.

Now, we  have the MP candidate set for each vertex. Once a request arrives, we 
find the candidate MPs only in the MP candidate sets of its source and 
destination.

We show an $O(\abs{V})$ selection algorithm with detailed analysis in Appendix 
\ref{A:select}.

\section{Hierarchical Meeting-Point Oriented Graph}
\label{sec:HSGraph}
With MPs, assigned routes can be concentrated on the convenient vertices, where inconvenient ones can be replaced by nearby MPs. In this section, we first validate this assumption by analyzing the real-world data. Then we rank the vertices and design a Hierarchical Meeting-Point 
Oriented graph (\emph{HMPO graph}). To be more specific, we 1) find defective vertices and guide us to assign riders through convenient vertices for effectiveness; and 2) further define core vertices, {which form a $k$-skip cover and reduce additional computing costs.}

\subsection{Graph Analysis for Meeting Point}
Take the road network of New York City on Open Street Map (OSM) \cite{OSM} as an example, which contains 58189 vertices
and 122337 edges. In this real-world road network, some vertices are more ``convenient'' 
than the others, which have larger traffic flows and lower transition costs. 
Based on this intuition, we evaluate the original road graph and 
vertex-flexibility with methods in Section~\ref{sec:MeetingPoints}. Based on 
the result in Section~\ref{sec:MeetingPoints}, vertices can be ranked by 
$ECI(\cdot)+ECO(\cdot)$ as an indication of convenience. We find that for more than 70\% vertices, each of their MP candidates has at least one vertex among the  
20\% most convenience vertices. According to this observation, convenient vertices would be more frequently used for ridesharing with MPs compared without MPs. On 
the other hand, some vertices are inconvenient to drive in and out for drivers, while MPs can be their convenient alternatives. This motivates us to build a 
hierarchical graph, which gives an indication for effective assignment and boosts the queries on convenient vertices. 


Formally, we introduce Hierarchical Meeting-Point Oriented Graph (\textsf{HMPO Graph}), which gives the hierarchical order over the vertex set $V$ and has 3 levels of vertices: \vspace{-1ex}
\begin{itemize}[leftmargin=*]
	\item \textbf{Core vertices $V_{co}$}. They are used as MPs  frequently. 
	\item \textbf{Defective vertices $V_{de}$}. They are inconvenient to access. 
	\item \textbf{Sub-level vertices $V_{su}$}. The remaining vertices are classified as sub-level vertices.
\end{itemize}\vspace{-1ex}

The three sets of vertices form a partition of all vertices in $G_c$ and $G_p$, 
that is, $V_{co}\cup V_{su}\cup V_{de}=V, V_{co}\cap V_{su}=\emptyset, 
V_{co}\cap V_{de}=\emptyset,$ and $V_{su}\cap V_{de}=\emptyset$. We introduce 
an example below, which will be used to show the main steps of our algorithms 
to build a HMPO graph in this section.

\begin{figure}[t!]\centering\vspace{-3ex}
	\centering
	\subfigure[{\scriptsize Graph for car $G_c$}]{
		\includegraphics[width=0.35\linewidth]{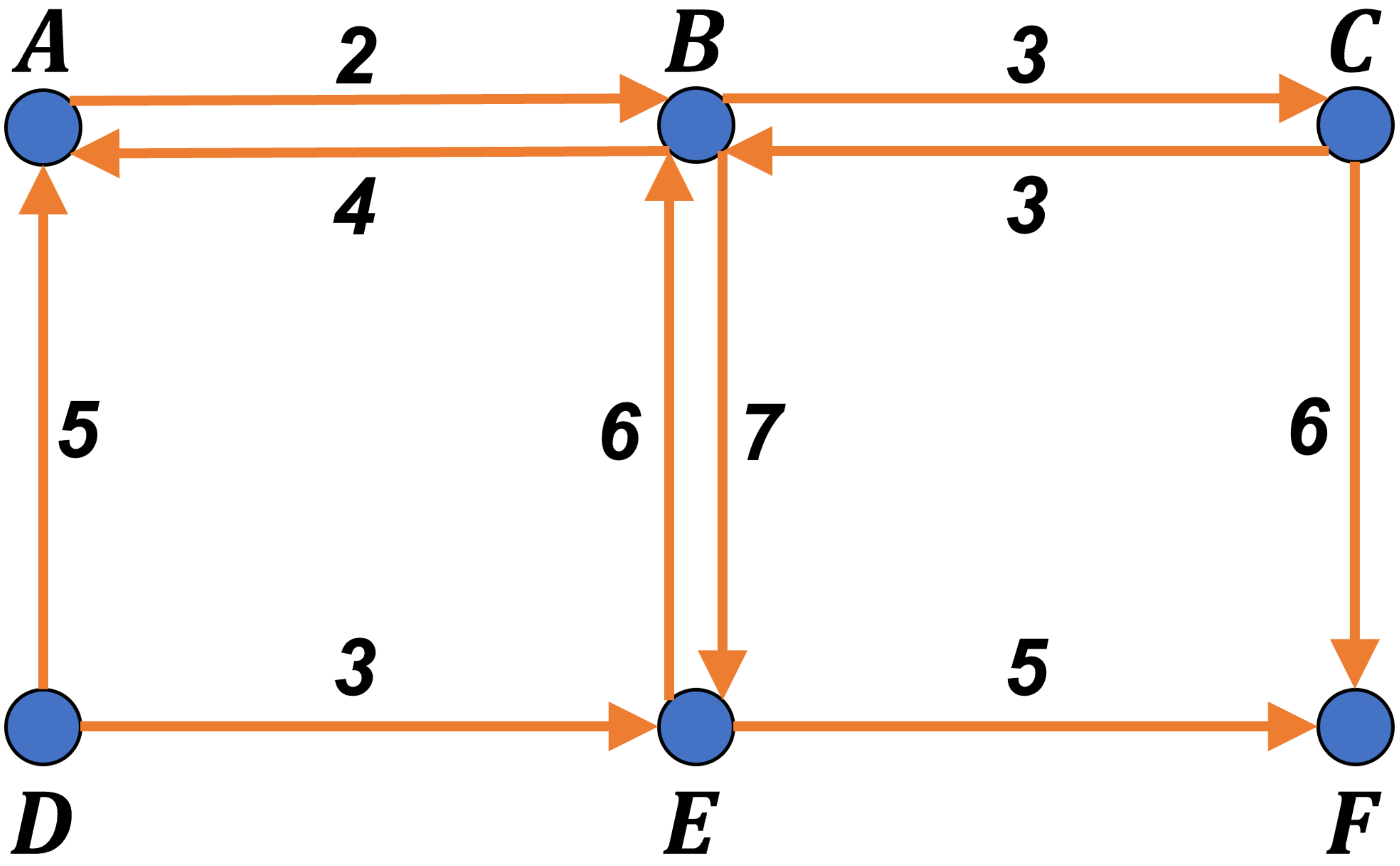}\label{fig:Gc_origin}
	}
	\subfigure[{\scriptsize Graph for passenger $G_p$}]{
		\includegraphics[width=0.35\linewidth]{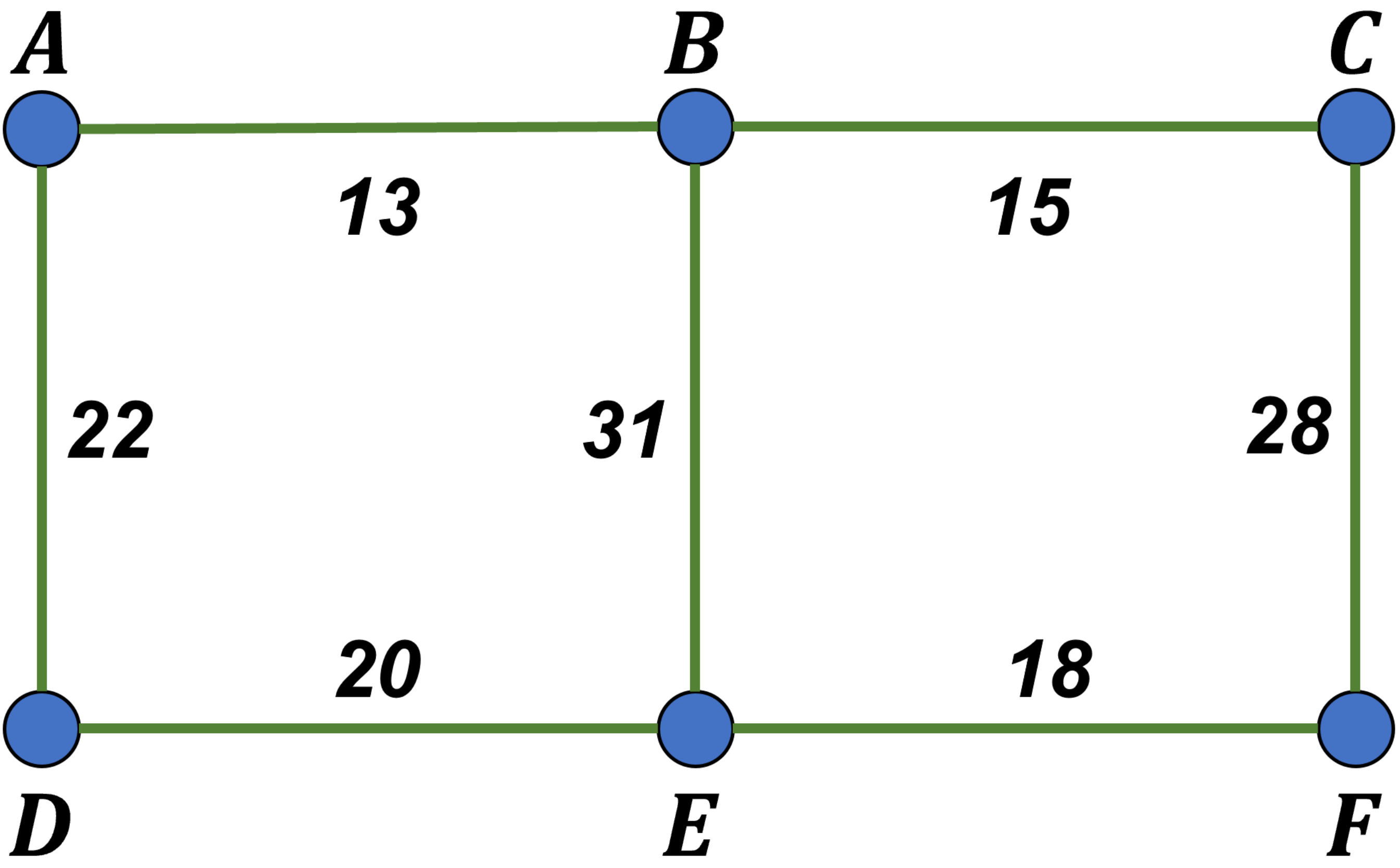}\label{fig:Gp_origin}
	}\vspace{-3ex}
	\caption{\small Original Graph for car and passenger}\label{fig:origin}
	\vspace{-2ex}
\end{figure}

\begin{example}
	\label{ex:graph}
	There are 6 vertices $A$ to $F$ in Figure \ref{fig:origin}. The graph for 
	car is directed in Figure \ref{fig:Gc_origin}. The graph for passenger is 
	undirected shown in Figure \ref{fig:Gp_origin}. Considering $n_r=3$ nearby 
	vertices, we derive $ECO(\cdot)$ and $ECI(\cdot)$ according to Section~\ref{sec:MeetingPoints} in Table~\ref{table:EC}. With $\alpha=\beta=1$, 
	$nc_m=3$, $d_m=30$, and $thr_{CS}=15$, we derive the MP candidates 
	$MC(\cdot)$ shown in Table~\ref{table:EC}.
\end{example}


\subsection{Defective Vertices}
Defective vertices are inconvenient vertices for vehicles to access and will be eliminated.
These defective vertices are not welcomed in traditional ridesharing either, but only with meeting points, removing them is feasible as we can serve riders with MPs. However, vertex removal also induces additional travel cost. We use a heuristic method to select and remove these vertices and improve the assignment effectiveness.



\textbf{Vertex removing cost.} Removing a vertex $u$ from a traditional graph leads to 2 kinds of additional costs for transportation: (i) \emph{the detour cost}. If a query finds a path containing $u$, removing $u$ means that we need to find a new path without passing $u$. The new path is usually longer than the original one and results in a detour; (ii) \emph{the inaccessibility cost}, which is from queries with $u$ as origins or destinations. No path will exist after removing $u$. 


With prepared MP candidates for each vertex in $G_p$, the walking 
cost is checked directly without query on $G_p$. Thus, for assigment, shortest path queries are only on the graph $G_c$, which is affected by vertex removing.

With MPs, the inaccessibility 
cost of removing vertices can decrease from $\infty$ to a limited mixture of 
walking and driving costs. 
We propose a heuristic algorithm called Defective Vertices Selection Algorithm (\emph{DVS 
	algorithm} for short) to remove the defective vertices without error from the 
detour cost and maintain its accessibility in the meantime. As the vertices with larger equivalent in/out cost $ECI$ and $ECO$ are harder to reach and leave, we try to remove defective vertices in decreasing order of $ECI+ECO$. For each vertex, we first check its MP candidates to ensure that at least one candidate is not defective and the vertex is accessible. Then, we search the shortest path queries passing it with a bounded search space. If there is no detour cost, we mark it as defective vertex.

We show the detail of DVS algorithm with time complexity $O(NlogN)$ in Appendix  
\ref{A:dvs} 
and an example in Appendix 
\ref{A:exp}. 
In addition, we propose two lemmas about the removing cost of defective vertices.

\begin{lemma}
	\label{proof3}
	Removing all vertices selected by the DVS algorithm from $G_c$ with 
	their edges leads to no detour cost.
\end{lemma}

\begin{proof}
	For details, please refer to Appendix 
	\ref{A:proof3}.
\end{proof}\vspace{-2ex}

\begin{lemma}
	\label{proof4}
	$\forall u\in V$ is accessible after removing vertices selected by the DVS 
	algorithm from $G_c$ with MPs.
\end{lemma}
\begin{proof}
	For details, please refer to Appendix 
	\ref{A:proof4}.
\end{proof}

\begin{table}[t!]
	\centering
	{\small \scriptsize
		\caption{\small $ECO$ and $ECI$ with \emph{3} nearby vertices, and $MC(\cdot)$}\vspace{-2ex}
		\label{table:EC}
		\begin{tabular}{c|cccccc}
			\hline
			{ID}&$A$&$B$&$C$&$D$&$E$&$F$\\
			\hline
			$ECO(\cdot)$ &5.33&4.67&5.33&5&6.67&$\infty$ \\
			
			$ECI(\cdot)$ &5.33&3.67&5.67&$\infty$&6.33&6.33 \\
			
			$MC(\cdot)$ &$\{A,B\}$&$\{B\}$&$\{B,C\}$&$\{A,D,E\}$&$\{E\}$&$\{C,E,F\}$ \\
			\hline
		\end{tabular}
	}\vspace{-3ex}
\end{table}

\subsection{Core Vertices}
Some convenient vertices take over the major traffic flow of a city network, such as the vertices along the highways. People in nearby places usually drive to the highway and go along it, then finally turn to small road for the destination. 
By assigning more route segments along highways, drivers can finish requests fast along highways in most time. 
With MPs, we can efficiently organize more routes along highways. Thus, these fast and request-concentrated vertices are more frequently used in solving the 
MORP problem. Finding such a group of convenient vertices and optimizing their related operations are pretty beneficial.

In this subsection, we select \emph{core vertices} $V_{co}$ as a backbone of the whole graph. As shown in Figure \ref{fig:k-skip}, $k$-skip cover \cite{DBLP:conf/sigmod/TaoSP11} is a skeleton vertex set, which coincides with our demand on finding convenient vertices and accelerating their related computations. Thus, we define the Core Vertices Selection Problem to find core vertices, which have a good coverage of MP with embedded $k$-skip cover. We first formulate the coverage of MP as an integer linear program, which can be solved with constant approximation ratio. Then, using an evaluational cost from the first step, we further complement its output and make it a $k$-skip cover as well.

First, we define the Core Vertices Selection Problem:
\begin{definition}
	Given the union of vertices $V=V_c\cup V_p$ and car graph $G_c$ with its selected defective vertices $V_{de}$, the \emph{Core Vertices Selection Problem} is to find a set of core vertices $V_{co}\subseteq V-V_{de}$ with the minimum size $\abs{V_{co}}$, such that
	
	\textbf{(i)} $V_{co}$ is a $k$-skip cover of the updated graph with $V_{de}$ and adjacent edges removed from $G_c$.
	
	\textbf{(ii)} Proportion factor $\epsilon$ of vertices in passenger vertex set $V_p$ has at least one 
	vertex $u\in V_{co}$ as its MP candidate.
\end{definition}

For attribute \textbf{(i)}, the concept of the \textbf{\bm{$k$}-skip cover} is introduced in Section~\ref{sec:background}. 
Both of \cite{DBLP:conf/sigmod/TaoSP11,DBLP:journals/pvldb/FunkeNS14} tend to 
\label{rw:r1-11}{ keep the 
	vertices that are components of more shortest paths} { and remove the ones that belong to fewer shortest paths}. 
They simply rank their vertices according to the degree, which is pretty coarse
as road network has very low vertex degree. Instead, we use a cost which is evaluated from the internal connections between vertices for ranking, aiming at comparing vertex importances more quantitatively. 
Attributes \textbf{(ii)} assures that most of the requests can be assigned with 
core vertices as MPs. So the optimization for core vertices can benefit more assignments. 
If a set of core vertices is valid to represent the graph skeleton, the smaller its size is, the more efficient its inner operations are after further preprocessing. So, we want to find a $V_{co}$ which satisfies attribute  \textbf{(i) (ii)} with minimal size.

Now, to satisfy the attribute 
\textbf{(ii)}, we can formulate an equivalent interger linear program. The detailed mapping steps include: (i) constructing candidate serving set $MS$ to 
indicate the set of servable vertices for each vertex {\scriptsize$u\in 
	V-V_{de}$}, that is, for any pair of vertex {\scriptsize$u,v\in V$, $u\in 
	MC(v)$} \emph{iff} {\scriptsize$v\in MS(u)$}. 
As an inverse-function relationship of $MC$, we can convert the MP candidate set by 
scanning each vertex {\scriptsize$u\in V$} and for each {\scriptsize$v\in 
	MC(u)$}, add $u$ to {\scriptsize$MS(v)$}. This process is linear; 
(ii) the universe {\scriptsize$\mathcal{U}$} is defined to be the set of all the 
passenger-accessible-vertices {\scriptsize$V_p$}; (iii) find the smallest sub-collection of 
vertices {\scriptsize$\{u_1,u_2,\cdots\}\subseteq V-V_{de}$} that the union of their serving 
sets covers at least $\epsilon$ of the universe, that is, 
{\scriptsize$\abs{\cup_{k}MS(u_k)}\geq \epsilon\cdot\abs{V_p}$}. 

Mathematically, we have the following integer linear program (ILP):

{\small 
	\begin{equation}
		\begin{aligned}
			{\min\text{\qquad}} & {\sum_{u \in V-V_{de}} \delta_u}&{} \\ 
			{\text { s.t.\qquad}} & {\phi_v+\sum_{u:v\in MS(u)}{\delta_u}\geq1}  & 
			{\forall v \in V_p} \\ 
			{}& {\sum_{v\in V_p}{\phi_v\leq(1-\epsilon)\abs{V_p}}}  & {} \\ 
			{}&{\delta_u\in\left\{0,1\right\}} &{\forall u\in V-V_{de},}\\
			{}&{\phi_v\in\left\{0,1\right\}} &{\forall v\in V_p,}
		\end{aligned}
	\end{equation}
}

\noindent where $\delta_u=1$ iff vertex $u$ is chosen to serve its candidate serving set $MS(u)$. The binary variable $\phi_v=1$ iff vertex $v$ is not served by any selected vertex.

The partial set cover problem is an NPC problem. Luckily, in our setting, each 
element (vertex) is in at most $nc_m$ candidate serving sets as the times of 
adding a vertex $u$ to some $MS(\cdot)$ equals to the number of $u$'s MP 
candidates {\scriptsize$\abs{MC(u)}\leq nc_m$}. Such a case is called a 
low-frequency system. There exists a solution to approximate the optimum within 
a factor $nc_m$ using LP relaxation in polynomial time 
\cite{DBLP:journals/jal/GandhiKS04}. In detail,
the corresponding LP relaxation can be derived by substituting the constraints 
$\delta_u\in\left\{0,1\right\}$ by $\delta_u\geq0, \forall u\in V-V_{de}$ and 
$\phi_v\in\left\{0,1\right\}$ by $\phi_v\geq0, \forall v\in V_p$. Then the 
problem is transferred to a linear Program $LP$. After deriving the dual LP of 
it, \cite{DBLP:journals/jal/GandhiKS04} iteratively chooses each set to be the 
highest cost set and obtains a feasible cover with the primal-dual stage. This step generates a cost $cost_u$ for each set. Intuitively, in our setting, this cost $cost_u$ can be treated as how bad the ``quantity'' each vertex $u$ with its serving set $MS(u)$ to be included in the final output. The higher its value, the worse a vertex is.
Finally, it chooses the solution with minimum cost.

However, to guarantee attribute (\textbf{i}) simultaneously, we cannot apply this algorithm directly to our problem. Both of the two attributes require the set of chosen vertices to have a good ``coverage'' of other vertices, that is, each of them can serve many nearby vertices. Motivated by this observation, we first use our equivalent in/out cost $ECI/ECO$ as a coarse estimation in the solution of \cite{DBLP:journals/jal/GandhiKS04}. Then we use the $cost_u$ of \cite{DBLP:journals/jal/GandhiKS04} as the ranking for $k$-skip cover completion. For the detail of our \emph{CVS} algorithm and complexity analysis, please refer to Appendix 
\ref{A:CVS}.

We show that the size of $V_{de}$ is bounded {with the following lemma} and present an example below.
\begin{lemma}\vspace{-1ex}
	\label{proof5}
	Assume that we have $N$ vertices in total, with $M$ set as optimal solution for the attribute \emph{(ii)}, the upper bound of the size of core vertex set is $\sigma(k)=\max(\frac{N}{k}log\frac{N}{k}, nc_m\cdot M)$.
\end{lemma}

\begin{proof}
	For details, please refer to Appendix
	\ref{A:proof5}.
\end{proof}

\begin{example}
	Let us continue using the setting in Example \ref{ex:graph}. After selecting $F$ and $D$ to $V_{de}$, we select the core vertices among the rest of vertices. We set $\epsilon=80\%$ to guarantee that no fewer than $\epsilon\abs{V}=4.8$ vertices are covered. $MS(\cdot)$ and the result of partial set cover are shown in Table~\ref{table:MS}. The final cover is {\scriptsize$V'_{co}=\{B,E\}$} with cost 2. To maintain a $2$-skip cover, we initialize {\scriptsize$V_{co}=\{A,B,C,E\}$} and iteratively check each vertex {\scriptsize$u\in V-V_{de}-V'_{co}=\{A,C\}$}. As none of these removals violates the 2-skip cover, both of them are removed from $V_{co}$. Removing $B$ or $E$ violates attribute (ii) and we finally output {\scriptsize$V_{co}=\{B,E\}$}.
\end{example}

\begin{table}[t!] \vspace{-2ex}
	\centering
	{\small \scriptsize
		\caption{\small Candidate serving set and partial set covers.}\vspace{-2ex}
		\label{table:MS}
		\begin{tabular}{c|cccccc}
			\hline
			{ID}&$A$&$B$&$C$&$E$\\
			\hline
			$MS(\cdot)$ &$\{A,D\}$&$\{A,B,C\}$&$\{C,F\}$&$\{D,E,F\}$ \\
			Partial set cover &$None$&$None$&$\{C,A,B\}$&$\{E,B\}$ \\
			Cost &$\infty$&$\infty$&$3$&$2$ \\
			\hline
		\end{tabular}
	}\vspace{-2ex}
\end{table}

Finally, as $V_{co}$, $V_{su}$, and $V_{de}$ is a partition of $V$,  $V_{su}=V-V_{co}-V_{de}$. \vspace{1ex}

\subsection{Construction of HMPO {color{red} \label{rw:r1-10}Graph} and Fast Query}
\label{subs:cons}
After obtaining the three levels of vertices, we use them to construct the 
hierarchical graph $G_h(V, E_h)$. In this subsection, we first concentrate on the 
formulation of edge $E_h$. Then, we show that how to compute shortest path queries fast on $G_h$.

\textbf{HMPO graph construction.} First, for the defective vertices, we ``discard'' them so there is no edge for 
drivers to come nor leave them. We get an updated car graph $G_{c'}$ by removing $V_{de}$ and its adjacent vertices from $G_c$ 

Recall that $V_{co}$ is a $k$-skip cover of $G_{c'}$. Based on that, all the shortest distance queries can be answered by $G^*$ 
efficiently. 

Core vertices serve as the skeleton of the original road network. Borrowing the definition in the previous work \cite{DBLP:conf/sigmod/TaoSP11}, for each $u\in V_{co}$, one can find its { \label{rw:r1-12}\emph{$k$-skip neighbors}} and build \emph{super-edges}, leading to a new graph for fast query. To be more specific, one can construct a graph $G^*=(V_{co}, 
E_{cc})$, which has 2 vertex sets $V_{co}$ and $V_{su}$ and 3 edge sets $E_{cs}$, $E_{sc}$, and $E_{ss}$. Any query result between these core vertices on the new graph is the same as that on the original graph. Query start from or aim at sub-level vertices can be answered by temporally extending the graph with super-edges between sub-level vertex and its nearby core vertices. 

Here we summarize the construction steps. (1) Instead of the $k$-skip neighbors of $u$, $N_k(u)$, we can find a superset $M_k(u)$ of $N_k(u)$ efficiently according to \cite{DBLP:conf/sigmod/TaoSP11}. (2) We build super-edge, which is weighted by the shortest path distance between two vertices, from $u\in V_{co}$ to each $v\in M_k(u)$. Add these super-edges between core vertices to $E_{cc}$. (3) For $u\in V_{su}$, we find its $M_k(u)$ on $G-V_{de}$ and $M_k^r(u)$ on the reversed graph. The super-edges built to $M_k(u)$ are stored into edge set $E_{sc}$, each represents a path from a sub-level vertex to a core vertex. Similarly, the super-edges built to $M_k^r(u)$ are stored into edge set $E_{cs}$ for core-to-sub-level vertex pairs. (4) In the middle of finding $M_k(u)$ for each $u\in V_{su}$, every $v\in V_{su}$ sharing a shortest path with $u$ without core vertex generates a sub-to-sub level super-edge, added to $E_{ss}$.

\textbf{Fast queries.} We can calculate the exact $k$-shortest path using the three edge sets \cite{DBLP:conf/sigmod/TaoSP11}. Consider the following cases: \textbf{(i)} $u,v\in V_{co}$. We simply return the query result on graph $G^*=(V_{co}, E_{cc})$; \textbf{(ii)} $u,v\in V_{sub}$. First we check the super-edges of $u$ in $E_{ss}$. If $u$ can reach $v$ directly, return the weight of super-edge. Otherwise, follow the general cases in \textbf{(iii, iv)}; \textbf{(iii)} $u\in V_{co}$. We need an additional step to add $u$ with its super-edges $(u,\cdot)\in E_{sc}$ into $G^*$; \textbf{(iv)} $v\in V_{co}$. Add $v$ with its super-edges $(\cdot,v)\in E_{cs}$ into $G^*$ and return the query result. Its correctness has been proved. For details, please refer to~\cite{DBLP:conf/sigmod/TaoSP11}. 

In summary, our final output for HMPO graph $G_h(V,E_h)$ is consists of: $V_{co}$, which forms the graph $G^*$ for query together with super-edges $E_{cc}$; $V_{su}$, of which each vertex is added to $G^*$ with edge sets $E_{sc}, E_{cs}, E_{ss}$ temporally for query; $V_{de}$, which is only called at the route planning stage with $MC$. $E_h=E_{cc}\cup E_{cs}\cup E_{sc}\cup E_{ss}$ and $V=V_{co}\cup V_{su}\cup V_{de}$.

\begin{figure}[t!]\centering\vspace{-2ex}
	\centering
	\includegraphics[width=0.7\linewidth]{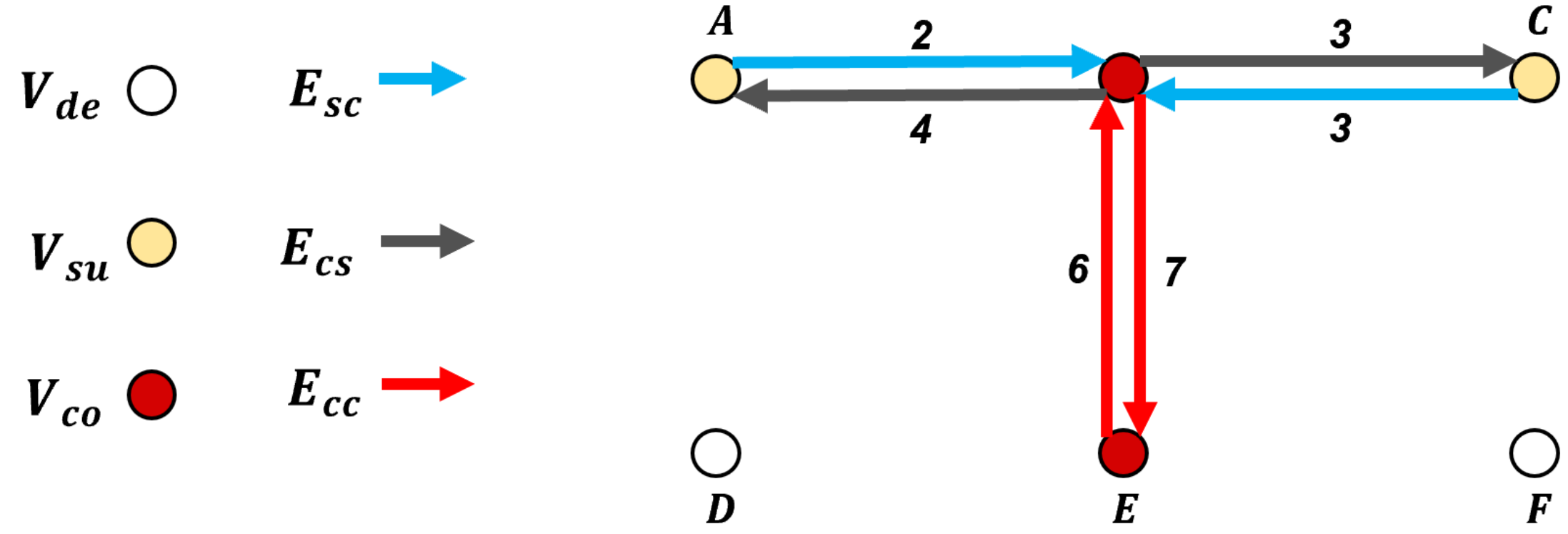}\vspace{-1ex}
	\caption{\small HMPO Graph}\label{fig:final}\vspace{-4ex}
\end{figure}

\begin{example}
	Following the setting in Example \ref{ex:graph}, we construct the $G_h$ in Figure~\ref{fig:final}. $V_{de}=\{D,F\}$ are marked in white and their edges are removed. $V_{co}=\{B,E\}$ are marked in dark red. Sub-level vertices $\{A,C\}$ are marked in light yellow. Two core-to-core edges between them are added to $E_{cc}$ and marked in red. Two edges from $B$ to $A$ and $C$ are added to $E_{cs}$ and marked in grey. Two edges from $A$ and $C$ to $B$ are added to $E_{sc}$ and marked in blue. 
\end{example}

\section{HMPO Graph Based Insertor}
\label{sec:hmpo_framework}
\begin{figure*}[t!]\centering
	\scalebox{0.37}[0.37]{\includegraphics{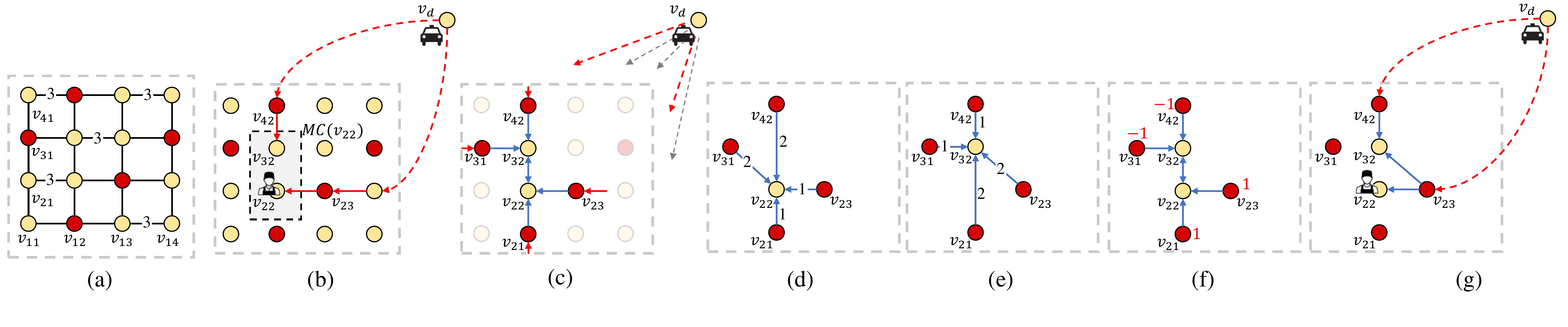}}\vspace{-5ex}
	\caption{\small An Example for the Set Maximum Differrence for Meeting Points}
	\vspace{-4ex}
	\label{fig:smdb}
\end{figure*}
With the HMPO Graph, we introduce a new algorithm to boost the insertion phase for the solution of MORP problem.

Requests are served one-by-one, where the MP candidates of their origins and destinations are inserted into drivers' routes. With limited walking distance, the MP candidates $MC(u)$
for each $u\in V_p$ are close to each other.
One interesting problem is, if we fail to insert a candidate $v\in MC(u)$, do
the rest $MC(u)-\{v\}$ help? To answer it, we define a new distance
correlation, which bounds the time saving of switching to any vertex in
$MC(u)-\{v\}$. If deducting the saving still cannot meet the time limitation,
we can safely prune the whole set. However, to derive it on the traditional
graph, we need all the distances from $\abs{V}$ sources even though vertices in
$MC(u)$ are close to each other.
Fortunately, we find that based on the $k$-skip structure in our hierarchical
graph, an upper bound can be derived effectively. As shown in Figure \ref{fig:k-skip},
$k$-skip cover ``cuts off'' the shortest paths using the core vertices, which can
be used as anchors to compare the difference between the paths from the same
source to a pair of vertices. Most of the related works
only use k-skip cover for faster query and vertex clustering
\cite{cheng2017utility}. To the best of our knowledge, this is the first
paper to explore this property of k-skip cover graph for task assignment.

We demonstrate an example through this section in Figure \ref{fig:smdb}. The network and edge weights are consistent with Figure \ref{fig:k-skip}(a), where red vertices belong to core vertices and yellow ones belong to sub-level vertices. In Figure \ref{fig:smdb}(b), a rider is waiting to be picked up at origin $v_{22}$ and a driver is at $v_d$. The MP candidates for $v_{22}$ are $\left\{v_{22},v_{32}\right\}$. We want to know that if driver cannot reach MP $v_{22}$/$v_{23}$ in time, can we prune the other MP.

In the following subsections, we first extract the important distance correlations in subsection~\ref{subs:MD}, then devise an effective algorithm \emph{SMDB}  for MORP problem in subsection~\ref{subs:SMDB}.

\subsection{Maximum Difference for Reaching Distances}
\label{subs:MD}
Insertion for MPs induces many shortest path queries, which are from the \emph{same} source to closely-located MPs. Here, we bound the distance differences between these queries. We first give related definitions for the bound.
Given a vertex $u$ and its MP candidates $MC(u)$, we have \textbf{(i)} one \emph{checker} vertex $Ch(u)\in MC(u)$, and \textbf{(ii)} the distance bound, named \emph{set maximum difference} ($SMD(u)$). We show that after computing the query of one long path from a vertex $v$ to checker $Ch(u)$, (i.e., $SP_c(v,Ch(u))$), the distance from $v$ to any other candidate vertices is no shorter than $SP_c(v,Ch(u))-SMD(u)$.

In this subsection, we first introduce the above concepts. Then, as any shortest path is split by $k$-skip cover $V^*$, we can quickly derive the bound locally, where only the subgraph formed by surrounding $V^*$ is needed.

First, we define \emph{maximum difference} $MD(\cdot,\cdot)$ between two vertices as a general distance correlation as follows:
\begin{definition}(Maximum Difference)
	Given a graph {\scriptsize$G(V,E)$}, for each pair of vertices {\scriptsize$v_1,v_2\in V$}, the maximum difference for $v_1$ and $v_2$ is: {\small$$MD(v_1,v_2) = \max_{v_3\in V}(SP_c(v_3,v_1)-SP_c(v_3,v_2))$$}\vspace{-3ex}
\end{definition}\vspace{-1ex}

In general, starting from any source $v\in V$, the largest difference between distances to reach two given vertices is denoted as their $MD$.

When we insert a set of MP candidates of source or destination into a route, it is more reasonable to use a set-based relationship for maximum difference rather than pairwise $MD$. To be more specific, for the MP candidates $MC(u)$ of vertex $u$, by defining a checker vertex $Ch(u)\in MC(u)$, we want to know the Maximum Differences between checker $Ch(u)$ and all the vertices in the candidate set. The upper bound of these $MD$s is defined as \emph{set maximum difference} ($SMD$). That is,
{\small$$SMD(u)=\max_{v\in MC(u)}MD(Ch(u),v).$$}\vspace{-1ex}

In the following part, we show that for a modified $SMD$ with narrower source vertices, which are not sub-level vertices that have super-edges to {\small$v\in V_M$}, the $k$-skip structure can help us get $SMD$ in $O(\abs{V})$ time.
Mathematically, we define {\scriptsize$$MD^*(u_1,u_2, V_M) = \max_{v:(v,u)\notin E_{sc}\cup E_{ss}, u\in V_M}(SP_c(v,u_1)-SP_c(v,u_2)),$$}\vspace{-1ex}
and 
\begin{equation}
	\label{SMDeq}
	\scriptsize{SMD(u)=\max_{v\in MC(u)}MD^*(Ch(u),v,MC(u)),}
\end{equation}
where $E_{sc}$ and $E_{ss}$ are super-edge sets in Section~\ref{subs:cons}.
\vspace{2ex}

With a $k$-skip cover $V^*$, any shortest paths of length $>k$ ended at a vertex $u$ must reach its surrounding $V^*$. Our core vertices $V_{co}$, also a $k$-skip cover, further builds super edges between $u$ and these surrounding $V_{co}$. As shown in Figure \ref{fig:smdb}(c), any paths ended at $v_{22}$ and $v_{32}$ need reach $v_{31},v_{42},v_{21}$, and $v_{23}$ first. They are also the only 4 vertices having super edges with $v_{22}$ and $v_{32}$ in the HMPO graph, e.g., Figure \ref{fig:smdb}(d). We check super edges and get the surrounding $V_{co}$ for each vertex in an MP candidate set $MC(u)$ and define their union as {\scriptsize$VC(u)$}. In Figure \ref{fig:smdb}, $VC(v_{22})=\left\{v_{31},v_{42},v_{21}, v_{23}\right\}$. As vertices in {\scriptsize$MC(\cdot)$} are close to each other, {\scriptsize$VC(\cdot)$} is a small set.

Note that the distance from any $vc\in VC(u)$ to any $v\in MC(u)$ is already recorded in HMPO graph, which does not need to compute. We denote its cost as $CC(vc,v)$. By fixing a vertex $w\in MC(u)$ as checker vertex, we can compute the following Local Maximum Difference $LMD$ quickly:
{\small\vspace{1.5ex}$$LMD(w)=\max_{v\in MC(u), vc\in VC(u)}(CC(vc,w)-CC(vc,v)),$$}\vspace{-2ex}

Now, we have the following Lemma:

\begin{lemma}
	\label{SMD}
	If we use $w\in MC(u)$ as the checker vertex, the Local Maximum Difference $LMD(w)$ is a valid Set Maximum Difference $SMD(u)$. Mathematically, if a vertex $l_c\in V$ has no super-edge {\small$(l_c,v)\in E_{ss}\cup E_{sc}$} towards any $v\in MC(u)$, then {\small$SP_c(l_c, w)-SP_c(l_c, v)\leq LMD(w)$}.
\end{lemma}\vspace{-2.5ex}

\begin{proof}\vspace{-1ex}
	We prove it by contradiction. Given a vertex {\scriptsize$u\in V$} and its
	{\scriptsize$MC(u)$} with outputs {\scriptsize$Ch(u)=w$} and
	{\scriptsize$SMD(u)$}, assume that {\scriptsize$\exists l_c\in V$} has no
	super-edges in {\scriptsize$E_{ss}\cup E_{sc}$} to {\scriptsize$MC(u)$},
	{\scriptsize$\exists v\in MC(u)$} which satisfies {\scriptsize$SP_c(l_c,
		w)-SP_c(l_c, v)>SMD(u)$}. We denote the last core vertex in the path
	from $l_c$ to $v$ as $vc$. Here we have
	{\scriptsize\begin{align}
			&LMD(w)<SP_c(l_c, w)-SP_c(l_c, v)\notag\\
			\leq &(SP_c(l_c,vc)+SP_c(vc,w))-(SP_c(l_c, vc)+SP_c(vc,
			v))\notag\\
			=&SP_c(vc,w)-SP_c(vc, v)\notag
	\end{align}}\vspace{-2ex}
	
	As $vc$ is the last core vertex along the path, $vc$ must be a $k$-skip neighbors of $v\in MC(u)$. Thus, we have: \vspace{-1ex}
	{\scriptsize\begin{align}
			&SP_c(vc,w)-SP_c(vc, v)\notag\\
			=&CC(vc,w)-CC(vc,v)\leq LMD(w)\notag
	\end{align}}\vspace{-4ex}
	
	Contradiction. The original lemma is proved.
\end{proof} \vspace{-2ex}

For example, in Figure \ref{fig:smdb}(f), if we choose $v_{32}$ and check vertex, we can calculate {\small$CC(v_{31},v_{32})-CC(v_{31},v_{22})=CC(v_{42},v_{32})-CC(v_{42},v_{22})=-1$} and {\small$CC(v_{21},v_{32})-CC(v_{21},v_{22})=CC(v_{23},v_{32})-CC(v_{23},v_{22})=1$. So $LMD(v_{32})=1$} when $vc=v_{21}$ or $v_{23}$.
After we compute the cost from $v_d$ to $v_{32}$ as $SP_c(v_d,v_{32})$, we want to bound the cost from $v_d$ to $v_{22}$ (i.e., $SP_c(v_d,v_{22})$). In Figure \ref{fig:smdb}(g), its last passed core vertex is $v_{23}\in VC$. we have {\small$SP_c(v_d,v_{22})=SP_c(v_d,v_{23})+SP_c(v_{23},v_{22})\geq SP_c(v_d,v_{23})+\left[SP_c(v_{23},v_{32})-LMD(v_{32})\right]\geq SP_c(v_d,v_{32})-LMD(v_{32})$}.

Now we can efficiently find $SMD$ given a checker. Here, we show the detail in our HMPO Graph-based Maximum Difference Generator (HMDG), which finds the {\scriptsize$Ch(\cdot)$} for each {\scriptsize$MC(\cdot)$} that minimize {\scriptsize$SMD(\cdot)$} when calculating $LMD$.

\noindent\textbf{Algorithm sketch}
Lines 3-8 collect all the cost from vertices {\scriptsize$vc\in VC$} to vertices {\scriptsize$v\in MC(u)$}. As {\scriptsize$VC\subseteq V_{co}$}, if  {\scriptsize$v\in V_{co}$}, we check the its super edges from core to core vertices. The costs are stored in the form of dictionary {\scriptsize$CC[vc][v]$}; if {\scriptsize$v\in V_{su}$}, we check its super edges from sub to core vertices instead.

Secondly, we initialize {\scriptsize$LMD(v)\rightarrow SMD$} as a dictionary, which can be further used to choose the checker with minimal $SMD$.
Instead of fixing checker and calculating corresponding $SMD$,
for each {\scriptsize$vc\in VC$}, we find the vertex $v^-$ which has minimal {\scriptsize$CC[vc][v^-]$}
in line 11, that is, among vertices in {\scriptsize$MC(u)$}, $v^-$ is the closest destination
for $vc$. If we use $v$ as checker, the {\scriptsize$SMD$} is the maximal of
{\scriptsize$CC[\cdot][v]-CC[\cdot][v^-]$}, which should be minimized. Lines 12-14
enumerate $MC$ and update the {\scriptsize$LMD[v]$} if a higher maximum difference is
found, that is, {\scriptsize$LMD[v]<CC[vc][v]-CC[vc][v^-]$}. Finally, {\scriptsize$LMD[\cdot]$}
saves the required {\scriptsize$SMD$} for each checker. We find the minimal value of {\scriptsize$LMD$}
with its key $v$ and return {\scriptsize$Ch(u)=v$} and {\scriptsize$SMD=LMD[v]$}.

{\small
	\begin{algorithm}[t!]
		\DontPrintSemicolon
		\KwIn{HMPO graph $G_h=(V,E_h)$, MP candidate sets $MC$.}
		\KwOut{Checker $Ch$ and Set Max Diff $SMD$ for each MP candidate set}
		
		\ForEach{$u\in V$}{
			Build set $VC= \{vc| (vc,v)\in E_{cs},v\in MC(u)\}$\;
			
			Initialize dictionary $CC$ for costs from $VC$ to $MC(u)$\;
			
			\ForEach{$v\in MC(u)$}{
				\If{$v\in V_{co}$}{
					Check the super edges of $v$. For each edge $\in E_{cc}$ from $vc$, record the costs into $CC[vc][v]$\;
				}
				
				\If{$v\in V_{su}$}{
					Check the super edges of $v$. For each edge $\in E_{cs}$ from $vc$, record the costs into $CC[vc][v]$\;
				}
			}
			
			Initialize $LMD[v]=0$ for all $v\in MC(u)$\;
			
			\ForEach{$vc\in VC$}{
				Find the minimal in $CC[vc][\cdot]$, denote the key as $v^-$\;
				
				\ForEach {$v\in MC(u)$}{
					\If{$LMD[v]< CC[vc][v]-CC[vc][v^-]$}{
						$LMD[v]= CC[vc][v]-CC[vc][v^-]$\;
					}
				}
			}
			
			Find the minimal value of $LMD[v]$, return $Ch(u)=v$ and $SMD(u)=LMD[v]$\;
		}
		
		\Return{$Ch$, $SMD$}
		\caption{\small HMPO graph-based Max Difference Generator}
		\label{algo:HMDG}
	\end{algorithm}
}

So our algorithm finds $Ch(u)$ with $SMD(u)$ such that: (i) for any source $vc\in VC$ from nearby core vertices, $\forall v\in MC(u)$, $SP_c(vc, Ch(u))-SP_c(vc, v)\leq SMD(u)$; (ii) $Ch(u)$ is chosen among $MC(u)$ to minimize $SMD(u)$.

\noindent\textbf{Time Complexity}. The time costs of Lines 2 and 3 are $O(nc_m)$. There are $O(nc_m)$ iterations in lines 4-8 and each iteration costs $O(\sigma^*)$, where $\sigma^*$ is the number of super edges a vertex has on average. Line 9 costs $O(nc_m)$ time. Line 11 costs $nc_m$ to find the minimum value and lines 12-14 form $O(nc_m)$ iterations taking $O(1)$ time in each iteration. For the size of $VC$, we borrow the definition $\bar{\sigma}_k$ from \cite{DBLP:conf/sigmod/TaoSP11}, where $\bar{\sigma}_k$ is the average number of k-hop
neighbors of the vertices in $V$. Thus, there are $O(\abs{VC})=O(\bar{\sigma}_k nc_m)$ iterations in lines 10-14. Their total time complexity is $O(\bar{\sigma}_k nc_m^2)$. Line 15 cost $O(nc_m)$ to find the minimum. So the total time complexity of the big loop in lines 1-15 is $O(\abs{V}(nc_m\sigma^*+\bar{\sigma}_k nc_m^2))$ and grows linearly with $\abs{V}$, where both $\sigma^*$ and $\bar{\sigma}_k$ depend on the structure of the road network instead of the size.

Note that, though we only discuss the distance relationships from the same source to different destinations, similar property still hold for queries from different sources to the same destination. We further emphasize that such a relationship is not only applicable for meeting points, but a property of $k$-skip cover. It is useful for many crowdsourcing scenarios, which have shortest path queries with close sources and destinations.

\subsection{SMD-Boost Algorithm}
\label{subs:SMDB}
We use $SMD$ to boost insertion in this subsection.

Recall that whenever we try to insert a request $r_j$ into a route, the rider
should be picked up no later than $tp_j$. 
Considering that we try to insert one pick-up into a fixed position of route. If we use $Ch(s_j)$ as the MP, it is not insertable if the time to pick up rider at $Ch(s_j)$, denoted as $t_c$ here, is later than $tp_j$. The worker needs to pick up rider at least $t_c-tp_j$ earlier. In the last subsection, we find the
maximum time saving for each checker, $SMD$. If $SMD$ of $Ch(s_j)$ is smaller than $t_c-tp_j$, no matter which MP of $s_j$ is used for insertion, we cannot pick up the request by $tp_j$. In such case, there is no need to try other MPs for insertion with the help of $SMD$.

{\small
	\begin{algorithm}[t!]
		\DontPrintSemicolon
		\KwIn{ a driver $w_i$ with route $S_{w_i}$, request $r_j$, MP candidate set
			$MC$, set maximum difference $SMD$, checker set $Ch$, dead vertices $DV$}
		\KwOut{ a  route $S_w^*$ for the driver $w$ and updated $DV$}
		\If{Driver's location $l_i\in DV$}{
			Return $S_{w_i}$ and $DV$ without insertion\;
		}
		
		Generate arriving time $arv[\cdot]$ for $S_{w_i}$\;
		
		Collect all sub-level
		vertices which have super-edges to vertices in $MC(s_j)$ into set $Ne$\;
		
		The largest index to insert pick-up: $id^*=\abs{S_{w_i}}$\;
		
		\ForEach{$v\in S_{w_i}$}{
			\If{$v\in Ne$}{
				Continue\;
			}
			\If{\scriptsize$arv[v]+SP_h(v,Ch(s_j))-SMD(Ch(s_j))\geq tp_j$}{
				\If{$v$=$l_i$}{
					Add $l_i$ to $DV$. Insertion fails and returns Null\;
					
				}
				Record $id^*=idx(v)-1$\;
				Break\;
			}
		}

		Insert $r_j$ with adapted insertion algorithm where insertion indexes of pick-ups larger than $id^*$ are pruned.\;
		
		\Return{$S_w^*$, $DV$}
		\caption{SMDBoost}
		\label{algo:SMDB}
	\end{algorithm}
}

We illustrate our algorithm
\emph{SMDBoost} for the insertion phase in Algorithm~\ref{algo:SMDB}. Note
that we add one more set for pruning, dead vertices $DV$. It means that no
driver $w_i$ with current location $l_i\in DV$ can serve this rider. We
initialize $DV=\emptyset$ for each new request. Assume that we try to insert
rider $r_j$ into the route of driver $w_i$. First, if $l_i\in DV$, we can prune
driver $w_i$. Otherwise, we derive arriving time $arv[\cdot]$ for each route
vertex according to \cite{tong2018unified}. In line 4, all the sub-level
vertices which have super-edges towards $MC(s_j)$ are collected into a set $Ne$, which covers vertices that cannot be pruned. This is because that
distances between these vertices and an MP can be arbitrarily short
through super-edges in $E_{sc}\cup E_{ss}$ and not follow the definition of $SMD$ in Equation~\ref{SMDeq}.

Pruning strategy in lines 5-12 finds the largest index $id^*$ to insert any pick-up in $MC(s_j)$. We initialize $id^*=\abs{S_{w_i}}$. Then we check each vertex $v\in S_{w_i}$ in order.
As the distances from vertices in $Ne$ to some MPs in $MC(s_j)$ can be
arbitrarily short, if $v\in Ne$, it is possible to insert rider after $v$. In
this case, we continue to check the next insertion position in line 8.
For each vertex $v\notin Ne$, it can be viewed as a source for $MC(u)$ subject
to $SMD(u)$. So if {\small$arv[v]+SP_h(v,Ch(s_j))-SMD(Ch(s_j))\geq tp_j$}, according to
Lemma.\ref{SMD}, inserting any MP after $v$ misses the deadline
$tp_j$. Insertion position can only be smaller than index of $v$ denoted as
$idx(v)$, that is, $id^*=idx(v)-1$. A special case is that if inserting after
the driver's current location $l_i$ cannot catch $t_j$, driver $w_i$ and all
the other drivers at $l_c$ currently can not serve $r_j$. So we add $l_i$ into
$DV$ for future pruning.

After the checking phase, we insert $r_j$ without checking indexes after $id^*$
for the pick-up point. The base algorithm adapts the linear insertion algorithm
\cite{tong2018unified} for MPs.

\noindent\textbf{Time Complexity}. Line 3 is a linear operation according to \cite{tong2018unified}. Line 4 costs $O(nc_m)$. There is an $O(\abs{S_{w_i}})$ loop in lines 5-12. All other lines from 1 to 12 cost $O(1)$. As the algorithm in line 13 is linear, the total time complexity is $O(\abs{S_{w_i}})$.

\begin{lemma}
	\label{proof6}
	Pruning Algorithm \ref{algo:SMDB} has no performance loss.
\end{lemma}

\begin{proof}
	For details, please refer to Appendix 
	\ref{A:proof6}.
\end{proof}

For example, in Figure \ref{fig:smdb}(g), assume that rider $r_i$ {needs} to be picked up in 10 time steps. If the shortest path from $v_d$ to $v_{32}$ is 12 time steps, as $SMD$ is 1, the shortest time to reach any other MP candidates is bounded by $12-1=11>10$. We can stop checking other MP candidate(s) ($v_{22}$). Any other worker at $v_d$ currently cannot serve $r_i$. $v_d$ is marked as a dead vertex and added to $DV$.

\section{Experimental Study}
\label{sec:experiment}

\subsection{Experimental Methodology}
\noindent\textbf{Data set}. We use both real and synthetic data to test our HMPO Graph-Based solution. Specifically, for the real data, we use a public data set
NYC~\cite{NYC}. It is collected from two types of taxis (yellow and green) in New York City, USA.
We use all the request data on December 30th to simulate the ridesharing requests in our experiments. Each taxi request in NYC contains the latitudes/longitudes of its source/destination locations, its starting timestamp, and its capacity. We can generate a ridesharing request and initialize its locations, release time, and capacity correspondingly.

In addition, we derive the distribution of requests
of all the NYC requests in December and generate {5 synthetic
	datasets (Syn) with the size of requests as 100k, 200k, 400k,  800k, and 1000k}.
We extract the road network of NYC from Geofabrik \cite{geofabrik}. It includes the labels of roads for both driving and walking. We clean it into the directed car graph $G_c$ and passenger graph $G_p$ according to road labels.
The road network is widely used in existing ridesharing studies~
\cite{tong2018unified}.

The settings of our dataset are summarized in Table~\ref{table:data}.

\begin{table}[t!]
	\centering
	{\small 
		\caption{\small Setting of Dataset and Model} \vspace{-3ex}\label{table:data}
		\begin{tabular}{c|c}
			\hline
			Parameters &Settings \\
			\hline
			Number of vertices of NYC & 57030\\
			
			Number of edges of NYC & 122337\\
			
			Number of valid requests of NYC & 277410\\
			\hline
		\end{tabular}
	}\vspace{-3ex}
\end{table}

\begin{figure*}[t!]\centering 
	\centering
	\subfigure[\scriptsize Served requests ($\abs{W}$)]{
		\includegraphics[width=0.23\linewidth]{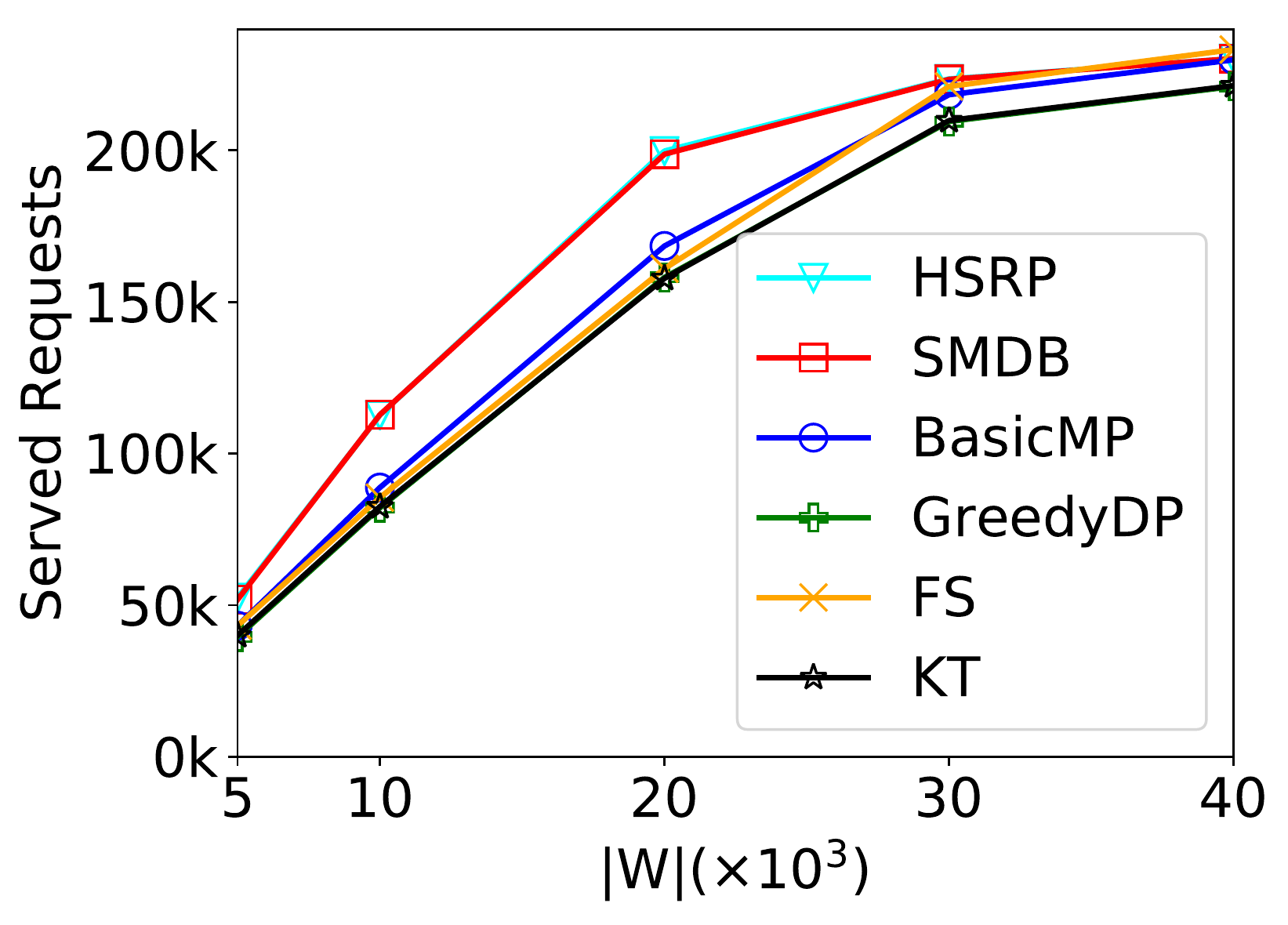}\label{subfig:sr_worker}
	}
	\subfigure[\scriptsize Served requests ($a_w$)]{
		\includegraphics[width=0.23\linewidth]{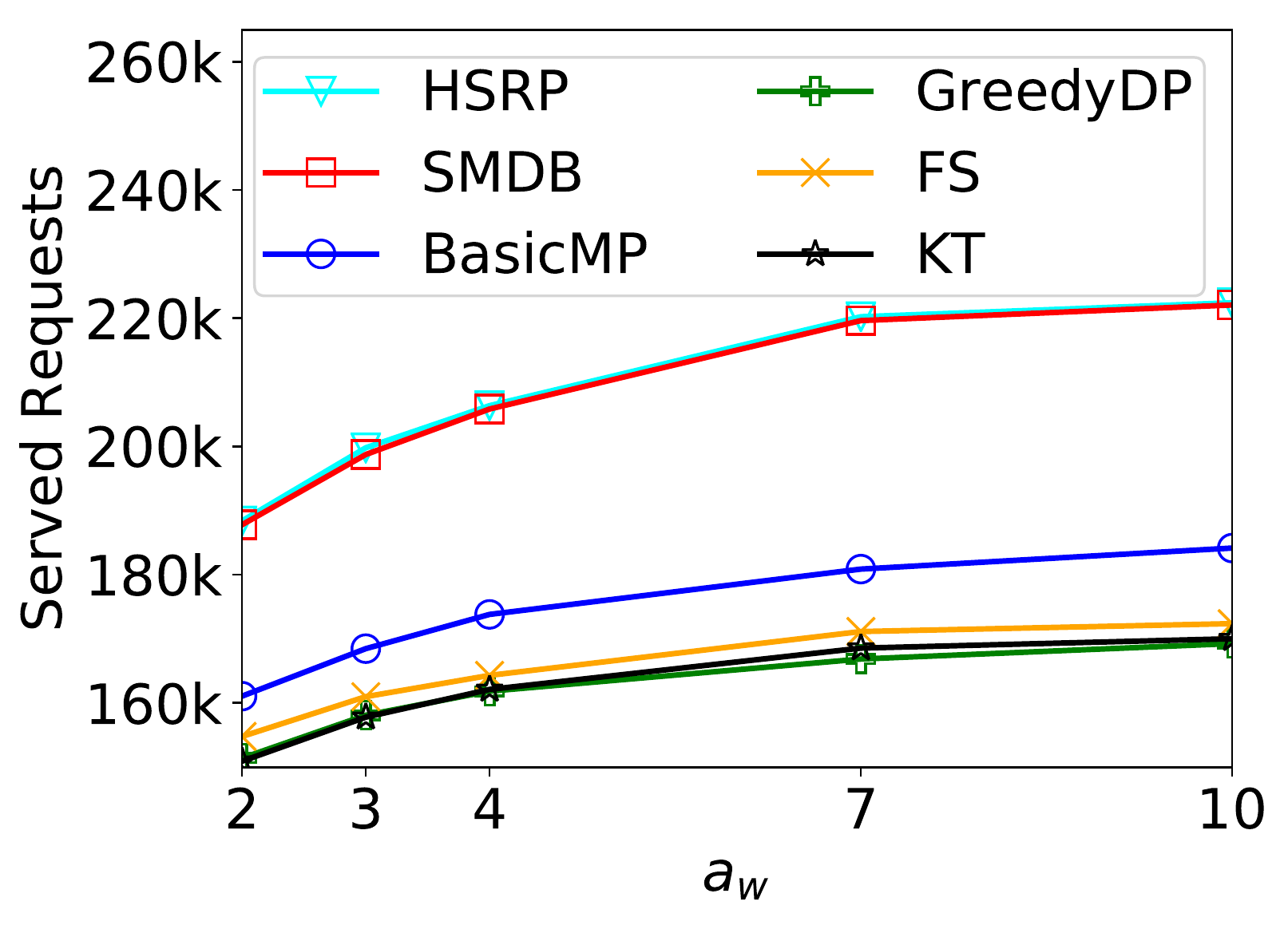}\label{subfig:sr_capacity}
	}
	\subfigure[ \scriptsize Served requests ($e_r$)]{
		\includegraphics[width=0.23\linewidth]{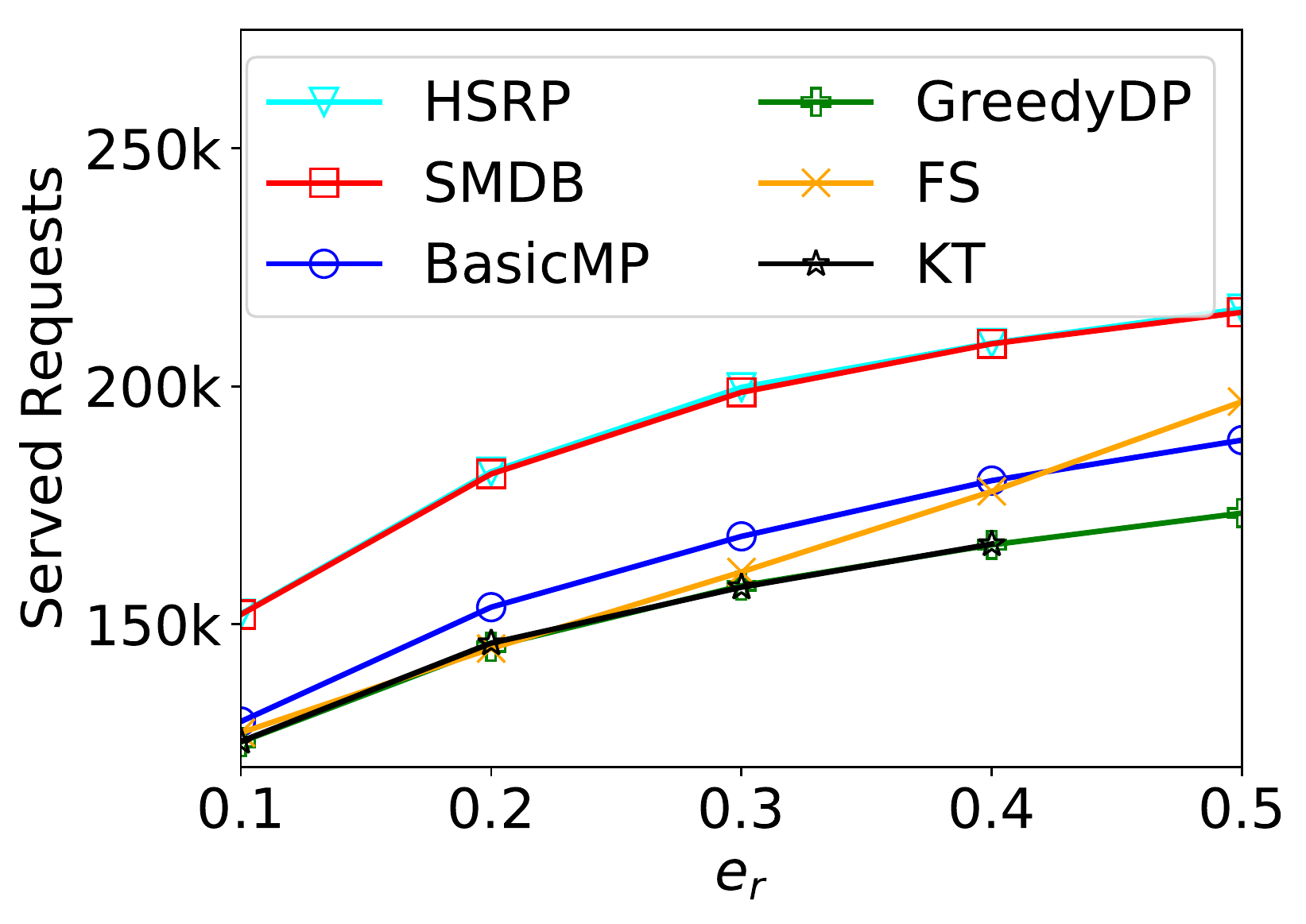}\label{subfig:sr_deadline}
	}
	\subfigure[\scriptsize Served requests ($\abs{R}$)]{
		\includegraphics[width=0.23\linewidth]{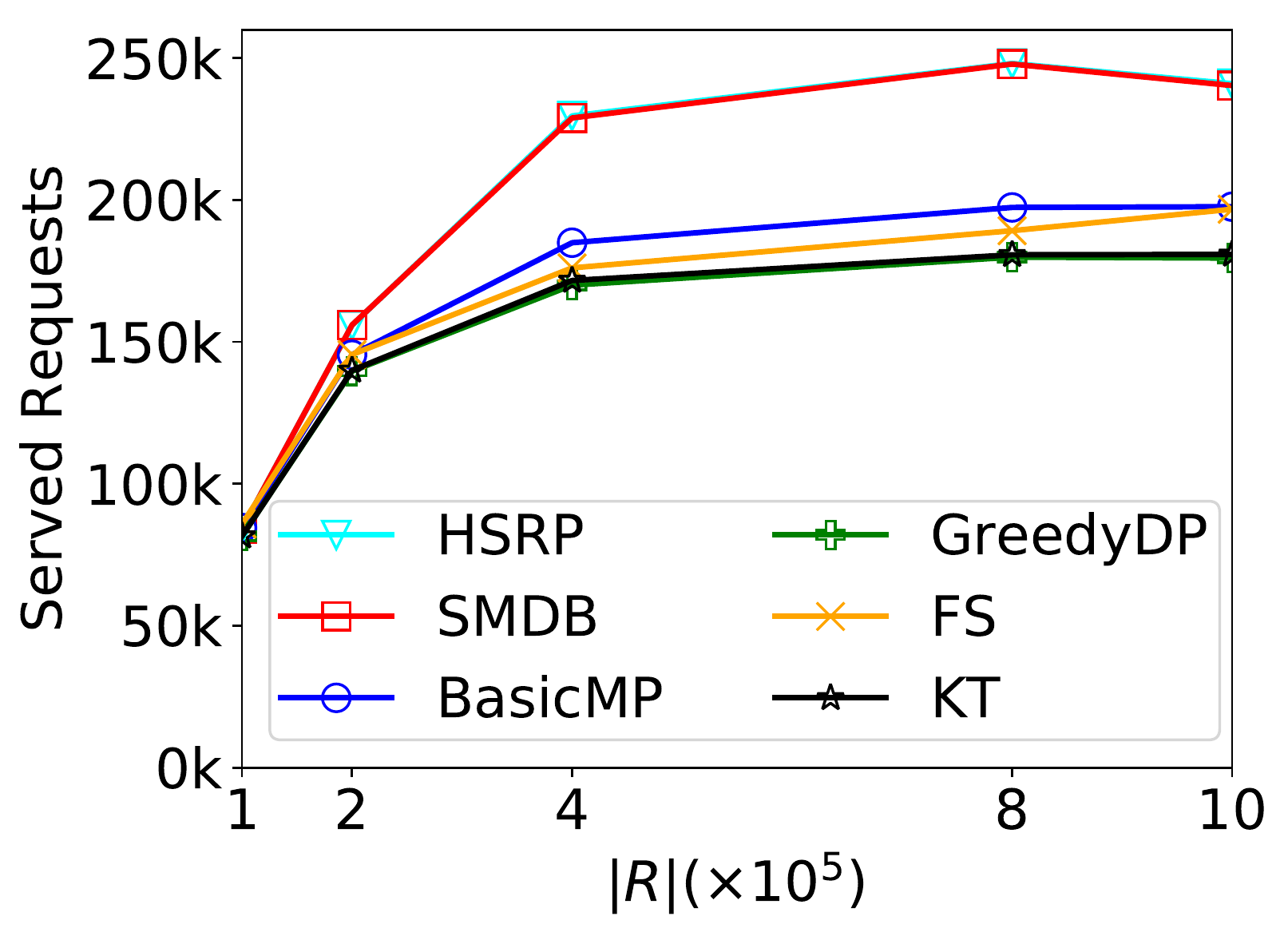}\label{subfig:sr_request}
	}\vspace{-2ex}
	
	\subfigure[ \scriptsize Unified cost ($\abs{W}$)]{
		\includegraphics[width=0.23\linewidth]{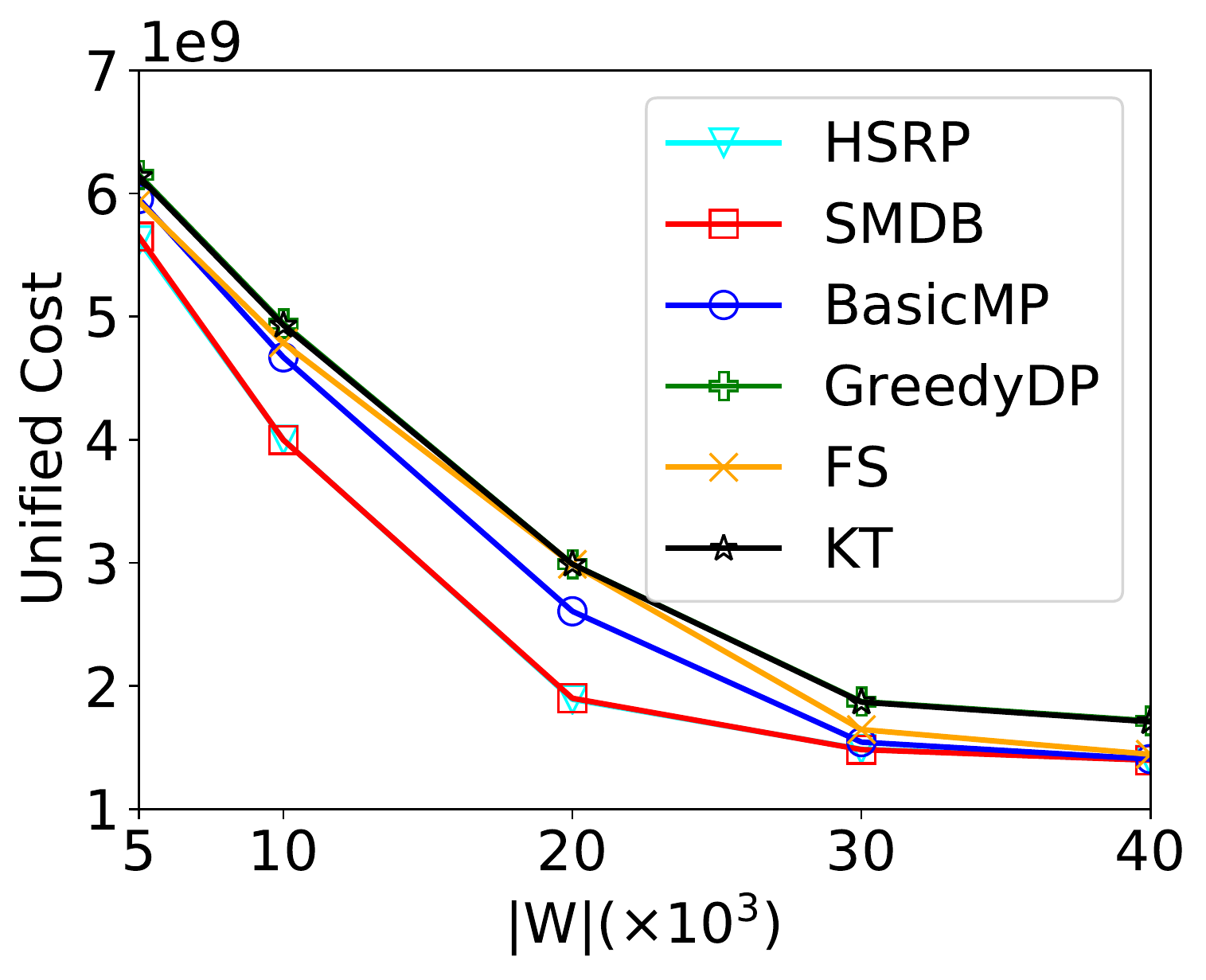}\label{subfig:uc_worker}
	}
	\subfigure[\scriptsize Unified cost ($a_w$)]{
		\includegraphics[width=0.23\linewidth]{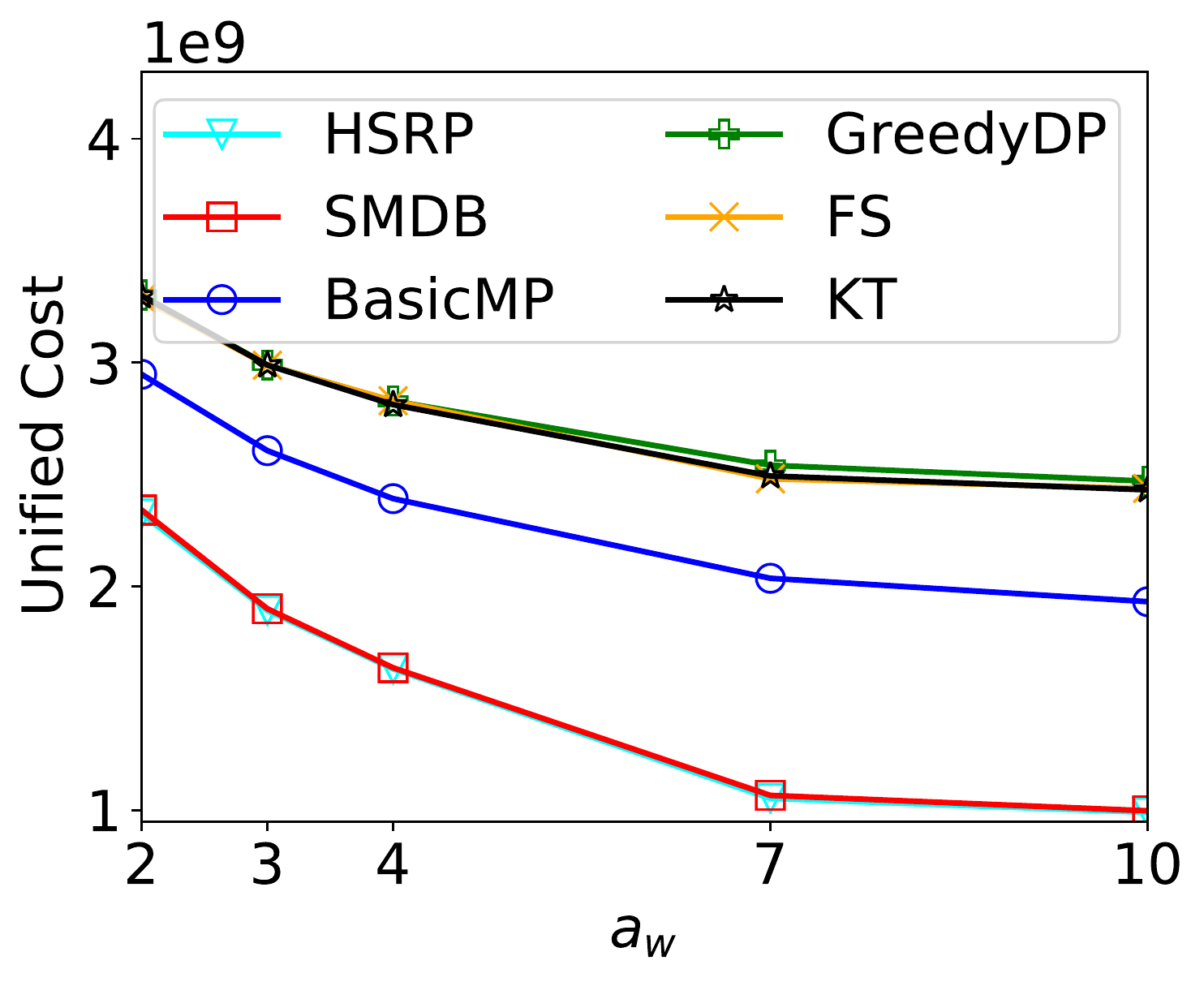}\label{subfig:uc_capacity}
	}
	\subfigure[\scriptsize Unified cost ($e_r$)]{
		\includegraphics[width=0.23\linewidth]{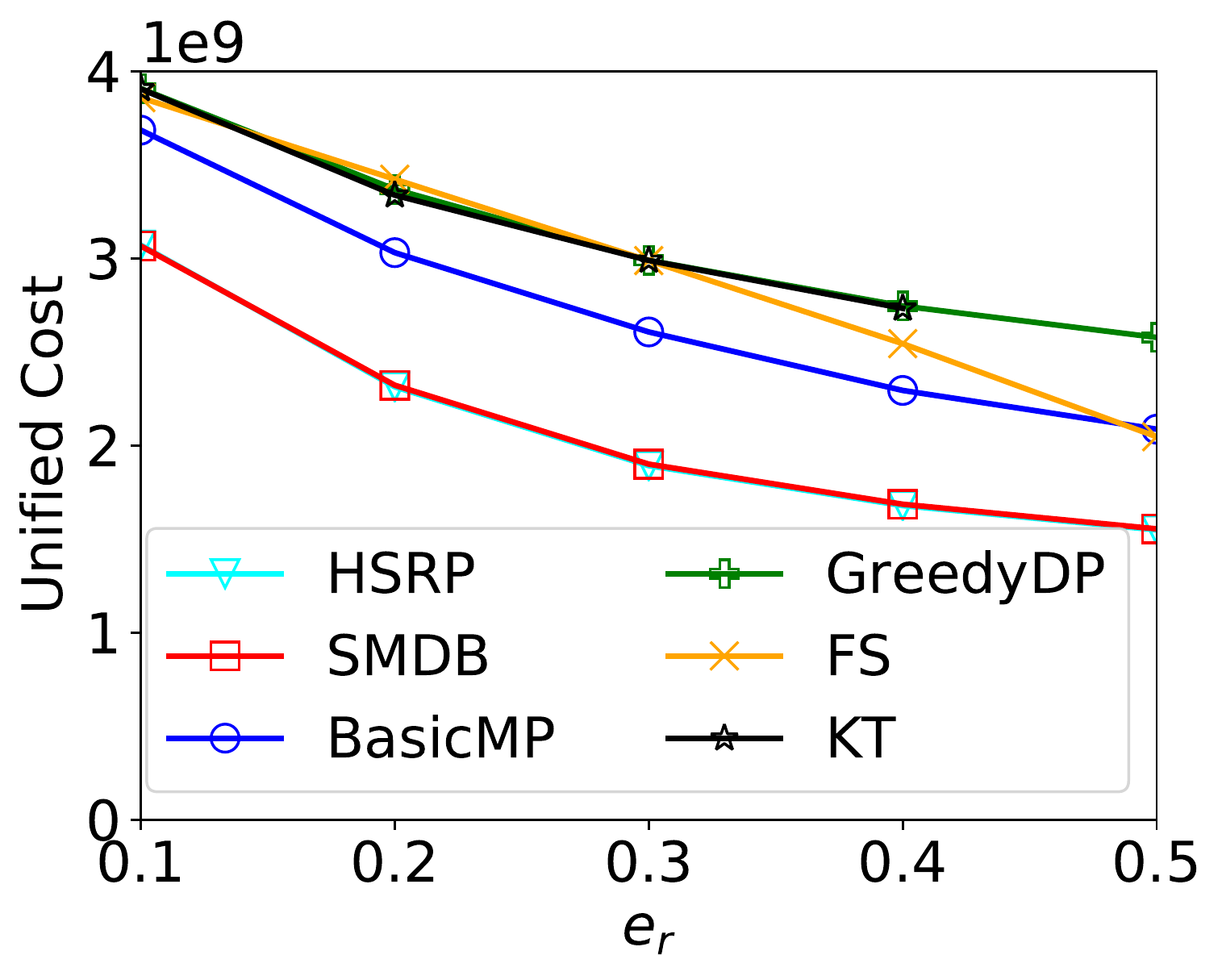}\label{subfig:uc_deadline}
	}
	\subfigure[\scriptsize Unified cost ($\abs{R}$)]{
		\includegraphics[width=0.23\linewidth]{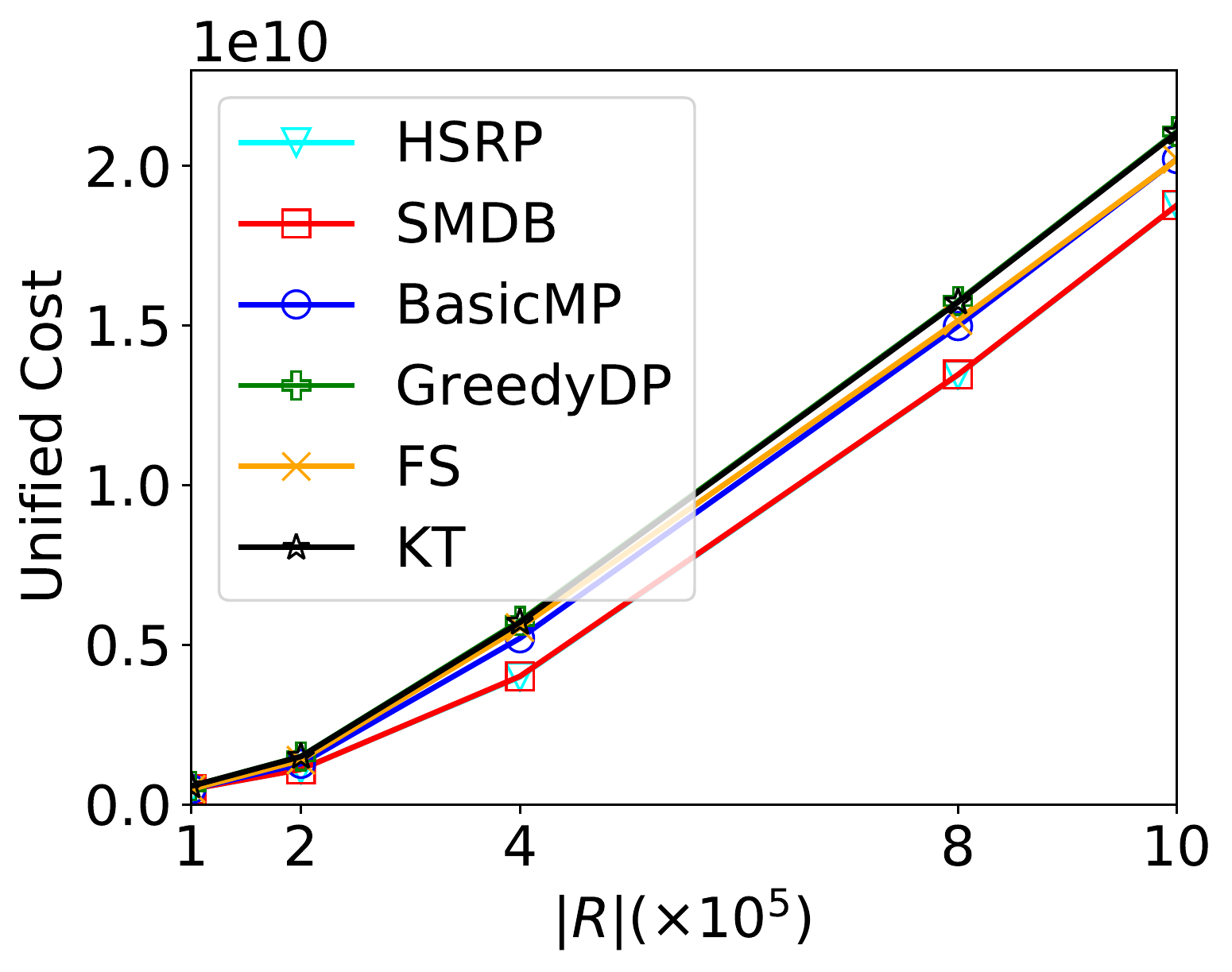}\label{subfig:uc_request}
	}\vspace{-2ex}
	
	\subfigure[\scriptsize Resp. time ($\abs{W}$)]{
		\includegraphics[width=0.23\linewidth]{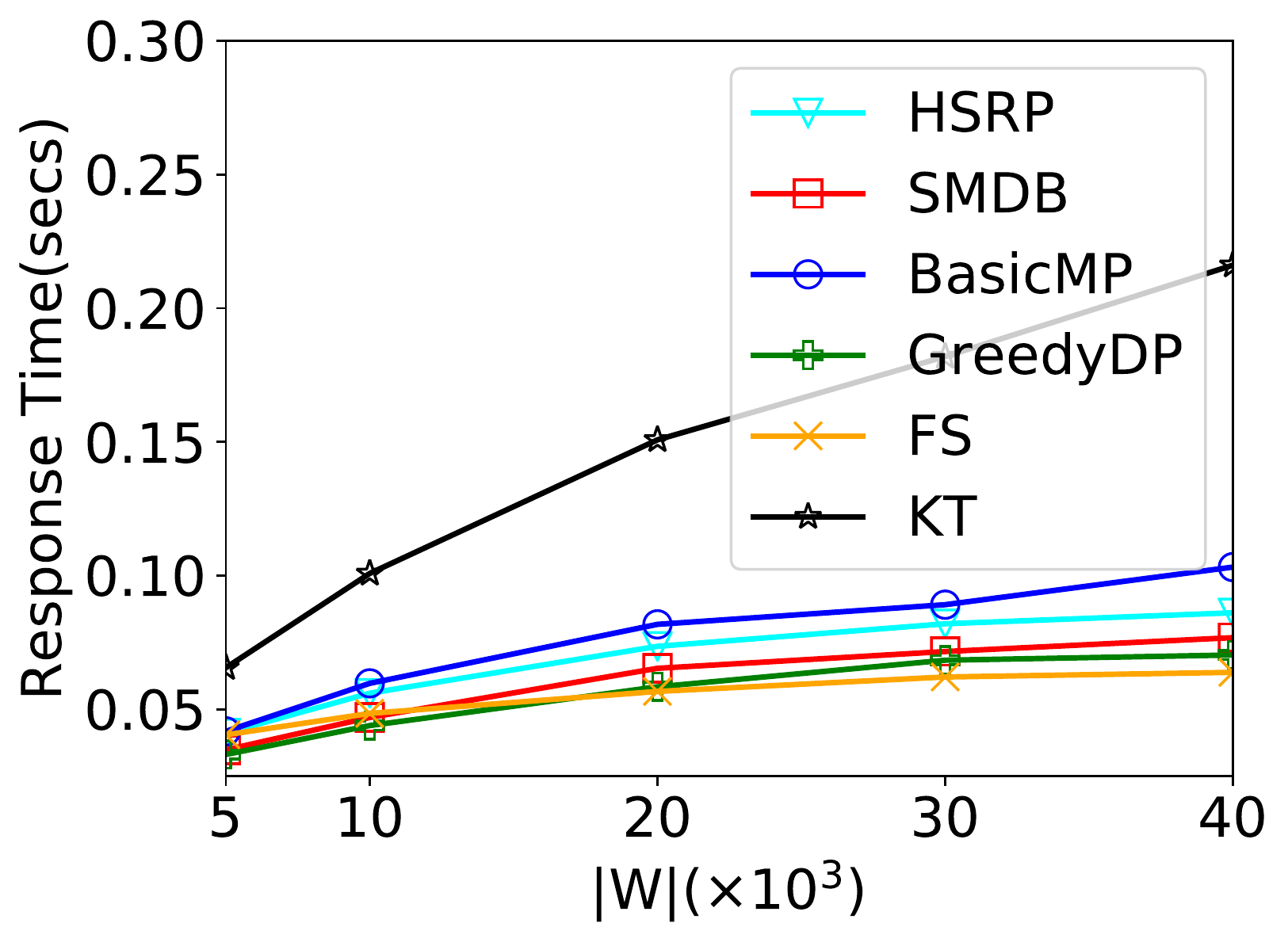}\label{subfig:rt_worker}
	}
	\subfigure[\scriptsize Resp. time ($a_w$)]{
		\includegraphics[width=0.23\linewidth]{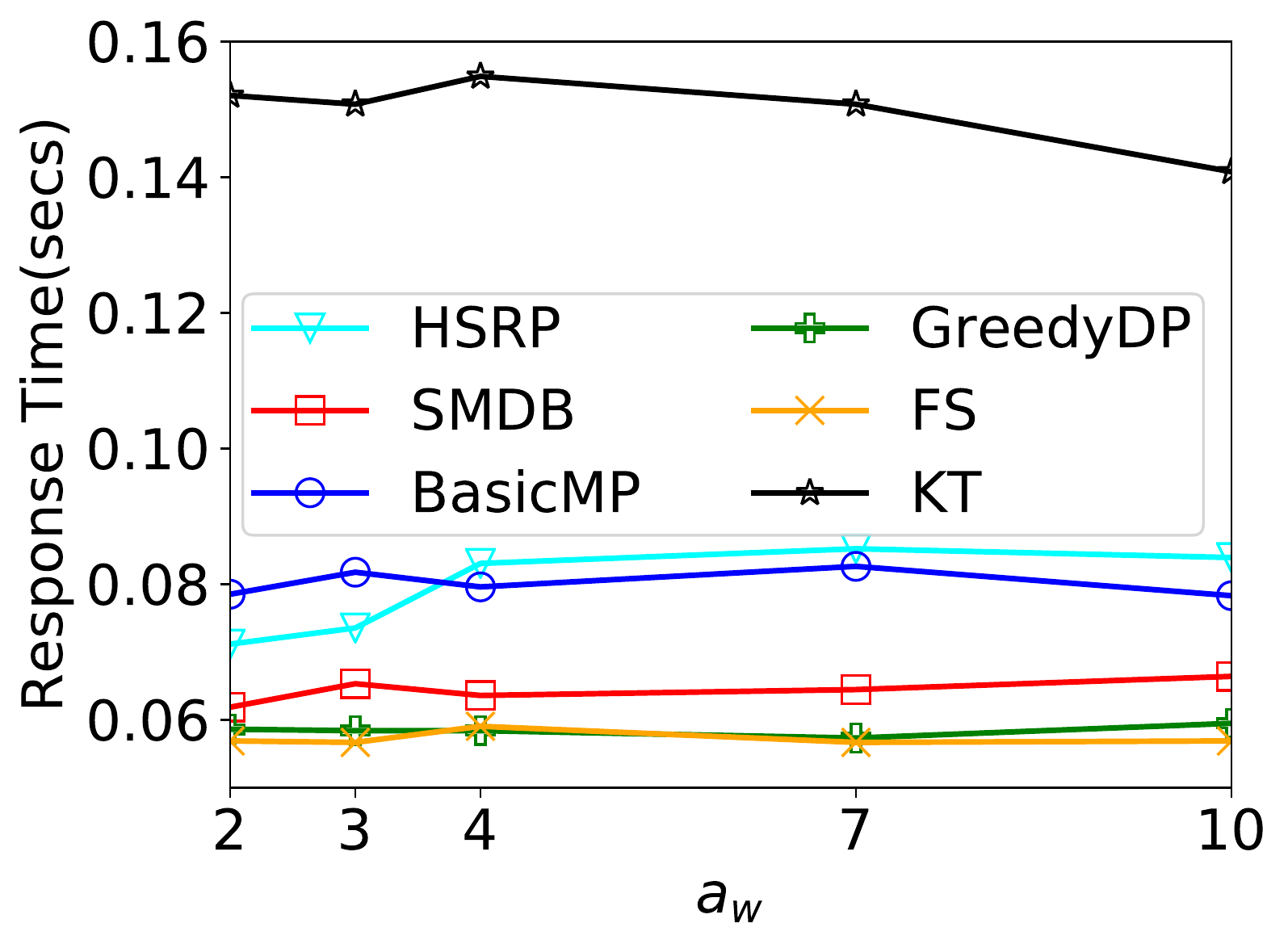}\label{subfig:rt_capacity}
	}
	\subfigure[\scriptsize Resp. time ($e_r$)]{
		\includegraphics[width=0.23\linewidth]{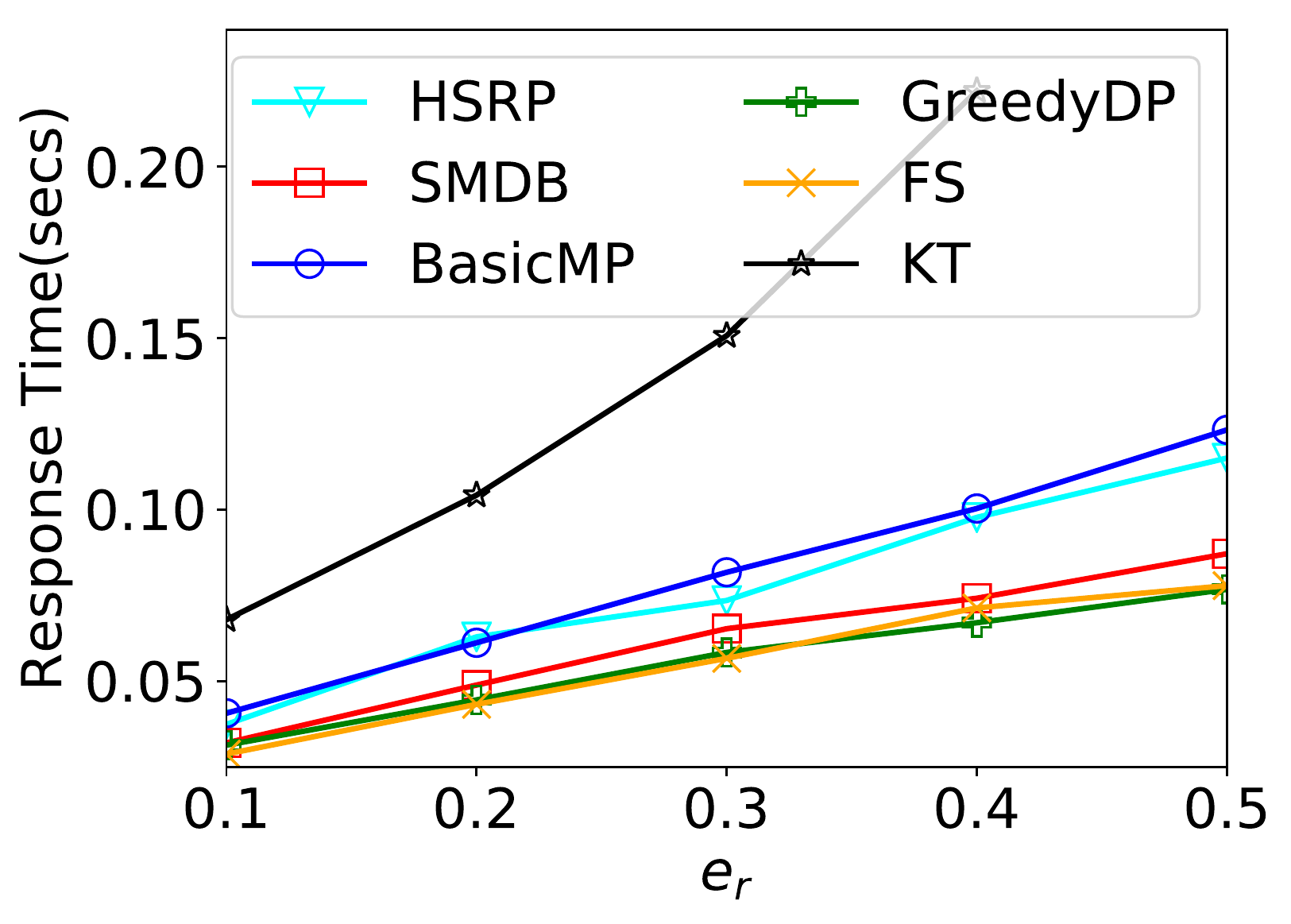}\label{subfig:rt_deadline}
	}
	\subfigure[\scriptsize Resp. time ($\abs{R}$)]{
		\includegraphics[width=0.23\linewidth]{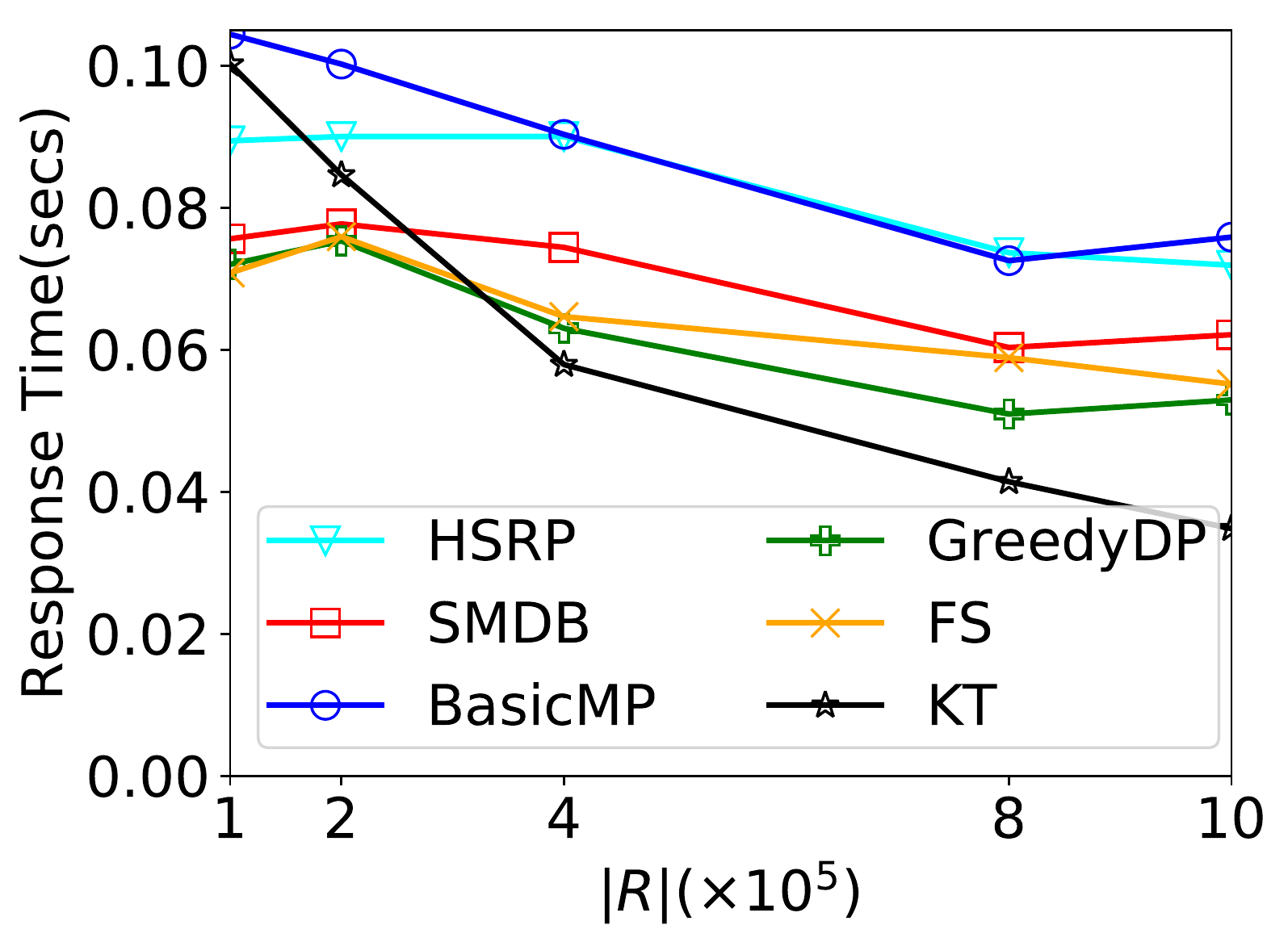}\label{subfig:rt_request}
	}
	\vspace{-2ex}
	\caption{Performance of varying number of drivers $\abs{W}$, capacity $a_w$, deadline coefficient $e_r$, and number of requests $\abs{R}$}\label{fig:exp_result}\vspace{-2ex}
\end{figure*}

\noindent\textbf{Implementation}. We follow the common settings for simulating ridesharing applications in \cite{asghari2016price, huang2014large, tong2018unified}. While building the graph for road network, the weights of edges are set to their time cost (divide the road length on Geofabric by the velocity of its road type). The walking speed is set to 3.6km/h, { \label{rw:r1-4}while the driving speed varies according to the road type in OSM, i.e., 60\% of the maximum legal speed limit in their cities \cite{Speed}}. For each request, we map its source and destination to the closest vertex in the road network. The initial location of each driver is randomly chosen.

\subsection{Setting Default  Values for Parameters}

\begin{table}[t!]
	\centering
	{\small 
		\caption{\small Parameter Settings.}\vspace{-3ex} \label{table1}
		\begin{tabular}{c|c}
			\hline
			Parameters &Settings \\
			\hline
			Deadline Coefficient $e_r$ &0.1, 0.2, \textbf{0.3}, 0.4, 0.5\\
			Capacity $a_w$ &2, \textbf{3}, 4, 7, 10\\
			Driving Distance Weight $\alpha$ &\textbf{1}\\
			Walking Distance Weight $\beta$ &0.5, \textbf{1}, 1.5, 2\\
			Penalty $p_o$  &{3, 5, 10, 15, }\textbf{30}\\
			Number of drivers $\abs{W}$ &5k, 10k, \textbf{20k}, 30k, 40k\\
			Number of requests $\abs{R}$ &100k, 200k, 400k, 800k, 1000K\\
			\hline
		\end{tabular}
	}\vspace{-2ex}
\end{table}

Table~\ref{table1} summarizes the major parameter settings for our experiments, where default values are in bold font. { \label{rw:r1}Other hyperparameters are discussed in Appendix1
	\ref{A:set} with extensive experiments.}

Note that for a mix mode where some riders are not willing to walk, our algorithm \emph{SMDB} can directly handle it by manipulating its input. Instead of generating meeting point candidates according to Section~\ref{sec:MeetingPoints}, we restrict the candidates as source $\{s_j\}$ for pick-up point and destination $\{e_j\}$ for drop-off point, where the walking distance is 0. The candidates are fed into \emph{SMDB}, which still works regardless of the number of candidates.

The experiments are conducted on a server with Intel(R) Xeon(R) E5 2.30GHz
processors with hyper-threading enabled and 128 GB memory. The simulation
implementation is single-threaded, and the total running time
is limited to 14 hours for
NYC. In reality, a real-time solution should stop before its time limit (24 hours for us)~\cite{huang2014large, tong2018unified}.
All the algorithms are implemented in Java 11.
We first preprocess our road network according to Section~\ref{sec:HSGraph}.
The modified vertices and weighted edges can be loaded directly for route planning.
We boost the shortest distance and path queries with an LRU cache according to the previous studies~\cite{huang2014large,tong2018unified}.

\noindent\textbf{Compared Algorithms.} We compare SMDB with the state-of-the-art algorithms for route planning of ride-sharing. 
\begin{itemize}[leftmargin=*]
	\item \textbf{GreedyDP} \cite{tong2018unified}. It uses a greedy strategy
	for route planning without MPs. Each request is assigned to the feasible
	new route with minimum increased cost. 
	{\item \textbf{Kinetic Tree (KT)}. For each driver, it saves all the possible routes for the assigned request using a structure called Kinetic Tree. The new request is inserted by traversing and updating the tree, which is more local optimal but slower compared with GreedyDP.}
	\item \textbf{BasicMP}. It is an extension from GreedyDP by adapting
	MPs to solve the MORP problem.
	{\item \textbf{First Serve (FS)}. A variant of BasicMP, where each request is directly assigned to the first driver who can serve it. }
	\item \textbf{HSRP}. It uses the HMPO Graph to improve the effectiveness of
	BasicMP without pruning.
\end{itemize}

\noindent\textbf{Metrics}. All the algorithms are evaluated in terms of
total unified cost, served requests
{\scriptsize$\abs{\hat{R}}$}
and response time (average waiting time to arrange a request, \emph{resp. time}
for short), which are widely used as metrics in
large-scale online ride-sharing proposals \cite{huang2014large,ma2013t,tong2018unified}.

\subsection{Experimental Results}
In this subsection, we present the experimental results.

\textbf{Impact of Number of Drivers $\abs{W}$}.
The first column of Figure \ref{fig:exp_result} presents the results with
different numbers of drivers in NYC. Compared with GreedyDP, BasicMP
outperforms it in terms of the number of served requests by 6.6\% to 12.7\%
with the help of MPs, while SMDB and HSRP outperforms it by 21.4\%
to 29.9\%. More requests are served with more drivers, results in a decrease of
unified costs and an increase in the served rates of all the algorithms.
BasicMP decreases the cost by 2.7\% to 13.2\% and SMDB decreases it by 4.5\% to
32.7\%. { GreedyDP and FS run the fastest, where GreedyDP does not use MPs and FS assigns each request to the first feasible worker suboptimally.} SMDB runs faster than other { local-optimal} MP-based methods (
HSRP and BasicMP) with a \emph{resp. time} $ <0.08s$.

\textbf{Impact of Capacity of Drivers $a_w$}.
The second column of Figure \ref{fig:exp_result} presents the effect of the
capacities of drivers. BasicMP serves 8.6\% to 12.2\% more requests than
GreedyDP with 5.1\% to 9.2\% less cost. SMDB outperforms other algorithms on
both serving rate, 27.3\% to 33.1\% higher than GreedyDP, and unified cost,
13.0\% to 20.5\% less than GreedyDP. With a larger capacity, all the algorithms
serve more requests with lower costs. However, the improvement with capacity
from 4 to 10 is not as significant as from 2 to 4. The reason is that the extra
spaces are mostly wasted. GreedyDP and FS still runs faster
but serve fewer requests. SMDB costs less time than BasicMP and HSRP. {KT is much slower than other methods but serves the fewest requests}.

\textbf{Impact of Deadline Coefficient $e_r$}.
The third column of Figure \ref{fig:exp_result} shows the results of varying
the deadline coefficient $e_r$. With larger $e_r$, all the algorithms serve
more requests with a lower unified cost. SMDB still serves more requests with
lower cost, which outperforms GreedyDP by serving 26.6\% to 30.3\% more
requests than GreedyDP and decreasing its cost by 9.8\% to 18.7\%. With meeting
points, BasicMP increases the number of served requests by 8.0\% to 12.6\% and
decreases the cost by 3.4\% to 9.2\% compared with GreedyDP. With a larger
deadline coefficient $e_r$, they find more feasible routes and serve more
requests. {KT builds tree to record all the possible routes for each driver, which even fails to finish with $e_r=0.5$.} Time cost increases with larger $e_r$ as each
request can find more feasible candidate routes. GreedyDP and FS still run the fastest
and SMDB is faster than BasicMP and HSRP.

\textbf{Impact of Number of Requests $\abs{R}$}.
The fourth column of Figure \ref{fig:exp_result} displays the results on
different sizes of synthetic requests. The datasets are generated based on the
distribution of all the NYC requests in December. { We increase the number of requests from 10k to 100k to test the extreme case that  drivers are fully busy, i.e., the number of served requests does not increase with 100k arrival requests compared with 80k requests for all the algorithms.}

{When drivers are not fully busy with requests fewer than 80k,} all the algorithms serve more
requests with a lower unified cost as the $\abs{R}$ increases.  Comparing with
each other, SMDB serves 7.3\% to 28.4\% more requests than GreedyDP and
decreasing its cost by 10.6\% to 21.8\%. BasicMP works weaker that increases
the number of served requests by 3.3\% to 10.8\% and decreases the cost by
5.0\% to 14.8\% compared with GreedyDP. { Response time of KT drops quickly with more requests while it rejects more requests. SMDB has a closer response time compared with GreedyDP and FS with increasing $\abs{R}$, where HSRP and BasicMP are the slowest methods.}

\vspace{1ex}
\textbf{Summary of Results}.
\begin{itemize}[leftmargin=*]
	\item Our SMDB algorithm can serve 7.3\% to 33.1\% more requests than the state-of-art algorithm \cite{tong2018unified}. The unified cost is decreased by 4.5\% to 32.7\%. These results validate the effectiveness of our algorithm in large scale datasets. 
	\item With MPs, BasicMP, HSRP, and SMDB outperform
	GreedyDP and KT. { Suboptimal MP-based First Serve method also has competitive performance compared to GreedyDP. } SMDB potentially arranges more requests on the highway and prunes
	candidates, which uses less time to serve more requests than BasicMP and
	HSRP. { \label{rw:r2-1}On the real-world dataset, SMDB can save 11\% to 24\% time compared with HSRP. On the synthetic dataset where a large amount of requests is test, SMDB saves more time (up to 31\% compared with HSRP) given more requests, which is important to handle heavy traffic cases. With a response time lower than 0.2 seconds, SMDB is acceptable to be
		used as a real-time solution for ridesharing tasks.}
\end{itemize}\vspace{-1ex}

\section{Conclusion}
\label{sec:conclusion}
In this paper, we propose the MORP problem, which utilizes meeting points for better ride-sharing route planning. We formulate a modified objective function to cover the cost of walking. 
We prove that the MORP problem is NP-hard and there is no polynomial-time algorithm with a constant competitive ratio
for it. 
To cut the search space of meeting points for fast assignemnts, we devise an algorithm to prepare meeting point candidates for each vertex. 
Besides, 
we construct an HMPO graph with hierarchical order on vertices for fast route planning, which takes the advantage of flexibility from meeting points to improve efficacy.
Based on it, we propose SMDB algorithm to solve MORP problem effectively and efficiently.
Extensive experiments on real and synthetic datasets show that
our proposed solution outperforms baseline and the state-of-the-art algorithms for traditional ridesharing in
effectiveness greatly without losing too much efficiency. Our paper
is provided as a comprehensive theoretical reference for optimizing route planning with meeting points in ridesharing.

\begin{acks}
	Peng Cheng’s work is partially supported by the National Natural Science Foundation of China under Grant No. 62102149, Shanghai Pujiang Program 19PJ1403300 and Open Foundation of Key Laboratory of Transport Industry of Big Data Application Technologies for Comprehensive Transport.
	Libin Zheng is supported by the National Natural Science Foundation of China No. 62102463, the Natural Science Foundation of Guangdong Province of China No. 2022A1515011135 and the Fundamental Research Funds for the Central Universities, Sun Yat-sen University No. 22qntd1101.
	Wenjie Zhang's work is partially supported by the ARC Future Fellowship FT210100303.
	Corresponding author: Peng Cheng.
\end{acks}

\bibliographystyle{ACM-Reference-Format}

\bibliography{meetingpoint.bib}

\appendix
\section{Proof of Lemma \ref{proof1}}
\label{A:proof1}
\noindent\textbf{Lemma \uppercase\expandafter{\romannumeral3}.1.} \emph{The 
	MORP problem 
	is NP-hard}.

\begin{proof}
	The basic route planning for ridesharing problems, which only takes the driving 
	cost and rejection cost into account, is NP-hard \cite{tong2018unified}. A 
	reduction from it to the MORP problem can be established by setting 
	$\beta=\infty$ to ban walking. So MORP problem is NP-hard.
\end{proof}\vspace{-2ex}

\section{Proof of Lemma \ref{proof2}}
\label{A:proof2}
\noindent\textbf{Lemma \uppercase\expandafter{\romannumeral3}.2.} \emph{There 
	is no randomized or deterministic algorithm guaranteeing constant CP for the 
	MORP problem.}
\begin{proof}
	By proving no deterministic algorithm can generate constant expected value 
	(e.g., $\infty$) with a distribution of the input including the destinations of 
	the requests, the previous work \cite{tong2018unified} guarantees that no 
	randomized algorithm has a constant CP using Yao's Principle 
	\cite{yao1977probabilistic}, which is also applicable for our problem. Our 
	problem is a variant of the basic problem as stated in the proof of Lemma 
	\ref{proof1}. Thus, no randomized or deterministic algorithm guarantees 
	constant CP for MORP.
\end{proof}\vspace{-2ex}

\section{Local-Flexibility-Filter Algorithm}
\label{A:select}
\begin{algorithm}[h!]
	\DontPrintSemicolon
	\KwIn{Graph $G_c$ for cars and $G_p$ for passengers, the number of reference vertices $n_r$, maximum walking distance $d_m$, threshold $thr_{CS}$ for pruning candidates, maximum number of candidates $nc_m$}
	\KwOut{MP candidates $MC$ for each vertex}
	
	Initialize integer lists $ECI$, $ECO$ with sizes of $\abs{V_c}$.
	
	Initialize set lists $n_i$ and $n_o$ with sizes of $\abs{V_c}$.
	
	Initialize dictionary list $MC$ with size of $\abs{V_p}$.
	
	\ForEach{vertex $u\in V_c$ of $G_c$}{
		Find the $n_r$ nearest vertices for $u$ on $G_c$.\;

		Derive $ECO(u)$.\;

		Find the $n_r$ reverse nearest vertices for $u$ on $G_c$.\;

		Derive $ECI(u)$.\;

	}
	
	\ForEach{vertex $u\in V_p$}{
		Find the accessible vertices, $A_u$, for $u$ on $G_p$ within $d_m$.\;
		
		\ForEach{vertex $v_i\in$ $A_u$}{
			If $SCS(u,v_i)\leq SCS(u,u)+thr_{CS}$, add it to u’s
			MP candidates and only keep top $nc_m$ points.
		}
		
		Add these candidate vertices into $MC(u)$.
	}

	\Return{$MC$}
	\caption{Local-Flexibility-Filter Algorithm}
	\label{algo:LFF}
\end{algorithm}

\textbf{Algorithm sketch}. Algo. \ref{algo:LFF} Local-Flexibility-Filter Algorithm (LFF Algorithm for short) shows the detail of our candidate selection solution. As vertices in road network have bounded number of adjacent edges (usually no more than 5), the degree of a vertex can be regarded as $O(1)$. We will stop after finding $n_r$ nearest vertices in line 5 and line 8, so there are $n_r$ rounds to pop vertex and relax edge. The time complexity of lines 5-8 is $O(n_r\log(n_r))$. The total time complexity from lines 4 to 8 is thus $O(\abs{V}n_r\log(n_r))$. We assume that given any vertex $u$, the maximum number of vertices that can reach $u$ within $d_m$ on $G_p$ is $n_w$, line 10 would have a time complexity $O(n_w\log(n_w))$. Line 12 costs $O(\log(nc_m))$ by using a min-heap to keep top $nc_m$ points. The lines 11-12 thus costs $O(n_w\log(nc_m))$. Line 13 costs $O(n_w)$. The total cost from line 9 to 13 is $O(\abs{V}n_w(\log(n_w)+\log(nc_m)))$. Because maximal number of candidates is always smaller than the number of reachable vertices, that is, $nc_m<n_w$, the time complexity of LFF algorithm is $O(\abs{V}(n_w\log(n_w)+n_r\log(n_r)))$.

\section{Defective Vertices Selection Algorithm}
\label{A:dvs}
\begin{algorithm}[h!]
	\DontPrintSemicolon
	\KwIn{Graph $G_c=\langle V_c,E_c\rangle$ for cars, list $ECO$ and $ECI$ 
		from Algo.\ref{algo:LFF}, MP candidates $MC$}
	\KwOut{The set of defective vertices $V_{de}$}
	
	Sort the vertices in $V_c$ in decreasing order of $ECO(v_i)+ ECI(v_i)$\;
	
	Initialize defective vertices $V_{de}=\emptyset$\;
	
	Initialize reserved vertices $V_{re}=\emptyset$\;
	
	\ForEach{Vertex $u\in V_c$}{
		\If{$u\in V_{re}$ or $MC(u)\subseteq V_{de}$}{
			Continue\;
		}
		Put vertices that can directly reach to $u$ into $inV$\;
		Put vertices that can be directly reached from $u$ into $outV$\;
		$t_m=\max_{v\in inV} t_c(v,u)+\max_{v\in outV} t_c(u,v)$\;
		
		Remove $u$ and its related edges from $G_c$ \;
		
		\ForEach{Vertex $v\in inV$}{
			Put every vertex $v'$ with distance $SP_c(v,v')\leq t_m$ into $V_{tm}$\;
			\If{$outV \nsubseteq V_{tm}$}{
				Put $u$ with its related edges back to $G_c$\;
				Check next vertex from Line 5\;
			}
			
			\While{Vertex  $v' \in V_{tm}$}{
				\If{{\scriptsize$v'\in outV$} and 
					{\scriptsize$SP_c(v,v')>t_c(v,u)+t_c(u,v')$}}{
					Put $u$ with its related edges back to $G_c$\;
					Check next vertex from Line 5\;
				}
			}
		}
		Add $u$ to $V_{de}$ and add all vertices in $MC(u)$ to $V_{re}$\;
	}
	
	\Return{$V_{de}$}
	\caption{Defective Vertices Selection Algorithm}
	\label{algo:DVS}
\end{algorithm}

\textbf{Algorithm sketch.} The pseudo code of the DVS approach is shown in 
Algorithm \ref{algo:DVS}. The main idea of the DVS approach is to ensure that: 
(a) every selected defective vertex will have at least one MP 
candidate existing in the rest vertices (i.e., core vertices or sub-level 
vertices); (b)  the travel cost of any two vertices in $G_c$ will not increase 
after removing the defective vertices and their related edges. 

The implementation of the DVS approach (Algorithm \ref{algo:DVS})  can be 
summarized in 3 steps:

\textbf{(i)} Sort  $v_i\in V$ in descending order of {\scriptsize$ECO(v_i)+ 
	ECI(v_i)$};

\textbf{(ii)} Initialize each vertex  as unmarked, which means that 
it is not required by any defective vertices as MP;

\textbf{(iii)} pop vertices one-by-one. For each vertex $u$, check whether it 
is marked or its MP candidates $MC(u)$ are all in $V_{de}$. If so, 
pop the next vertex. Otherwise, remove it from $G_c$ and record each vertex $v$ 
which has an edge with $u$. If the edge starts from $v$ (an edge comes to $u$), 
we add $v$ into a set $inV$ with the weight $t_c(v,u)$; if it starts from $u$, 
we add $v$ into a set $outV$ with the weight $t_c(u,v)$. Save the sum of the 
largest in and out weights as $t_m$. Then we remove $u$ with its adjacent edges 
from $G_c$. For each vertex $v\in inV$, we run the Dijstra algorithm from 
source $v$ until the new reached vertex costs more than $t_m$. If any $v'\in 
outV$ is accessed with cost larger than $t_c(v,u)+t_c(u,v')$, we refuse to add 
$u$ into $V_{de}$ and check next one. If all vertices in $outV$ have been 
accessed, add $u$ to $V_{de}$ and mark vertices in $MC(u)$. $u$ and its edges 
are recovered if we failed to add it to $V_{de}$. 

A vertex $u$ with a large $ECO$/$ECI$ is hard to reach and leave, thus it has a 
high possibility to be isolated from the city center and the main street. 
Removing the defective vertices will have little influence on the connectivity 
of the road network. We sort the vertices according to their $ECO$ and $ECI$ at 
the beginning of DVS, which help us remove the unimportant vertices first as 
they affect fewer vertices, such that a more concise HMPO graph 
can be achieved.

\textbf{Time Complexity}. We assume that each vertex has an $O(1)$ degree, which is common in road networks. Line 1 runs in $O(N\log N)$ time, where $N=\abs{V}$. Lines 7 and 8 collect adjacent vertices in $O(1)$ time as degree is in $O(1)$ scale. Assume that (i) the longest length-2 path in $G_c$ costs $t_{near}$; (ii) given any vertex $u\in G_c$, the number of vertices that $u$ can reach within cost $t_{near}$ is no larger than $\sigma_{near}$, then the number of accessed vertices is bounded by $\sigma_{near}$. Thus, line 12 is a $O(\sigma_{near}\log\sigma_{near})$-time-complexity loop. All checking phases (lines 5, 13, 17) are implemented in hash index and cost $O(1)$. Loop from lines 11 to 19 and loop from lines 16 to 19 enumerate $O(1)$ vertices. Loop from line 4 to 20 enumerates $O(N)$ vertices. Thus, the overall time complexity is $O(N\cdot(\log N+\sigma_{near}\log\sigma_{near}))$.

\section{Example for Algorithm \ref{A:dvs}}
\label{A:exp}

\begin{example}
	Let us continue the setting in Example \ref{ex:graph}. Now we want to select the defective vertices among them. We first sort the vertices and initial their labels as unmarked. Then we check all the vertices one-by-one and show their status at each round with a row in Table~\ref{table:de}. Round 0 initializes flags as unmarked ($u$). Round 1 pops $F$ and there is no node in $inV$. We add it to $V_{de}$ and mark its candidates $MC(F)$ as $m$. In the table we use $de$ to suggest that $F$ is added to $V_{de}$. Round 2 adds $D$ to $V_{de}$ and updates its unmarked candidate(s) ($A$). Rounds 3~5 all pop marked vertices and fail to add a new vertex to $V_{de}$. Round 6 checks the $B$ by first removing it from the graph with edges and recording its $inV= outV=\{A,C,E\} $ with $tm=6+7=13$. Running Dijkstra algorithm from $A$ results in $\infty$ cost and poping $B$ is rejected. We mark the failed ones as $f$. 
\end{example} 

\begin{table}[h!] \vspace{1ex}
	\centering
	{\small \scriptsize
		\caption{\small Status of each vertex at each round}
		\label{table:de}
		\begin{tabular}{c|cccccc}
			\hline
			{Round}&$F$&$D$&$E$&$C$&$A$&$B$\\
			\hline
			$0$ &$u$&$u$&$u$&$u$&$u$&$u$ \\
			
			$1$ &$de$&$u$&$m$&$m$&$u$&$u$ \\
			
			$2$ &$de$&$de$&$m$&$m$&$m$&$u$ \\
			
			$3$ &$de$&$de$&$f$&$m$&$m$&$u$ \\
			
			$4$ &$de$&$de$&$f$&$f$&$m$&$u$ \\
			
			$5$ &$de$&$de$&$f$&$f$&$f$&$u$ \\
			
			$6$ &$de$&$de$&$f$&$f$&$f$&$f$ \\
			\hline
		\end{tabular}
	}\vspace{-2ex}
\end{table}

\section{Proof of Lemma \ref{proof3}}
\label{A:proof3}
\noindent\textbf{Lemma \uppercase\expandafter{\romannumeral5}.1.} 
\emph{Removing all vertices selected by the DVS algorithm from $G_c$ with 
	their edges leads to no detour cost.}
\begin{proof}
	Recall that if a shortest path query $SP_c(v_1,v_2)$ finds a route contains 
	a vertex $u$ in the middle, we need to find an alternative path without $u$ 
	after removing $u$ from graph $G_c$. If the new route has higher cost, 
	removing $u$ results in a detour cost. We define the new car graph without 
	$V_{de}$ as $G_{c'}=G_c-V_{de}$. Hereby the proof is equivalent to proof 
	$\forall v_1,v_2\notin V_{de}, SP_c(v_1,v_2)=SP_{c'}(v_1,v_2)$, where 
	$SP_{c'}(v_1,v_2)$ is the shortest distance query on graph $G_{c'}$. 
	We prove it by construction, that is, given any shortest path on $G_c$ from 
	$v_1$ to $v_2$, where $v_1,v_2\notin V_{de}$, we show that there is a path 
	from $v_1$ to $v_2$ on $G_{c'}$ with the same cost.
	
	We use $\{u_1,u_2,\cdots,u_{\abs{V_{de}}}\}$ to denote the removed 
	defective vertices in order. When we remove a $u_k$, any length-2 shortest 
	path $(v_x, u_k, v_y)$ must have a same cost substitution 
	$SCS_{x,y}=$$(v_x, v_{s_1},v_{s_2}\\\cdots, v_{s_p}, v_{y})$, where 
	$v_{s_i}\notin \{u_q|0<q<=k\}$, $i=1, 2,\cdots,p$. It is satisfied 
	according to the phase \textbf{(iii)} of Algorithm~\ref{algo:DVS}.
	
	Then, for any path on $G_c$ from $v_1$ to $v_2$, we can iteratively find 
	each $u_k$ with the lowest index. Denote its previous and latter vertices as 
	$v_x, v_y$, we substitute $(v_x, u_k, v_y)$ with $SCS_{x,y}$. In each 
	round, the lowest index of $u_k$ in the new path is increasing. After at 
	most $\abs{V_{de}}$ rounds, the path has no vertex in $V_{de}$. As each 
	substitution does not increase the cost, the final substitution is a valid 
	path on $G_{c'}$ with cost equal to $SP_c(v_1,v_2)$.
\end{proof}\vspace{-2ex}

\section{Proof of Lemma \ref{proof4}}
\label{A:proof4}
\noindent\textbf{Lemma \uppercase\expandafter{\romannumeral5}.2.} 
\emph{$\forall u\in V$ is accessible after removing vertices selected by the 
	DVS	algorithm from $G_c$ with the help of meeting points.}
\begin{proof}
	After popping each vertex $u$ in phase 3 of the DVS algorithm, we ensure 
	that $\exists v\in MC(u)$ is in $V_{de}$ to guarantee that any vertex 
	without available meeting points is still in graph $G_c$. Once a request 
	starts from or aims at $u$, we can serve it by $u$ itself. On the other 
	hand, if a vertex is added to $V_{de}$, we mark its meeting point 
	candidates so that we would not further remove these vertices from $G_c$. 
	Request with $u$ as origin or destination can be served via its meeting 
	point candidates in $G_c$. 
\end{proof}\vspace{-2ex}

\section{Core Vertices Selection Algorithm}
\label{A:CVS}
{\small
	\begin{algorithm}[t!]
		\DontPrintSemicolon
		\KwIn{All the vertices $V$ and defective vertices $V_{de}$, list $ECO$ and 
			$ECI$ from Algo.\ref{algo:LFF}, MP candidates $MC$}
		\KwOut{The set of core vertices $V_{co}$}
		
		Initialize candidate serving set for each $u\in V-V_{de}$ as $\emptyset$.
		
		\ForEach{$v\in V_p$}{
			\ForEach{$u\in MC(v)$}{
				Add $v$ to $MS(u)$;
			}
		}
		
		Solve the partial set cover problem using the algorithm in  
		\cite{DBLP:journals/jal/GandhiKS04}, where the weights of sets are 
		substituded with $ECO+ECI$ in its sorting step. Finally, record its cost 
		$cost_u$ of choosing each $MS(u)$ as the highest cost set. Denote its 
		output as $V'_{co}$.
		
		Initialize the core vertices $V_{co}=V-V_{de}$.
		
		Sort the vertices $u\in V-V'_{co}-V_{de}$ in decreasing order of $cost_u$. Check them one-by-one and remove a vertex from $V_{co}$ if the $V_{co}$ is still a k-skip cover.
		
		Sort and check the vertices $u\in V'_{co}$ in decreasing order of $cost_u$. 
		Remove a vertex $v$ from $V_{co}$ if $V_{co}-\{v\}$ is still a k-skip cover 
		and the $MS$s of $V_{co}-\{v\}$ cover $\epsilon\cdot\abs{V_p}$ vertices.
		
		\Return{$V_{co}$}
		\caption{Core Vertices Selection Algorithm}
		\label{algo:CVS}
	\end{algorithm}
}

\textbf{Algorithm sketch.} We construct it with three steps:

(1) In lines 1-4, we first construct candidate serving sets $MS(\cdot)$. Then we obtain the partial set cover solution $V_{co}'\in V-V_{de}$ satisfying attribute (i) using \cite{DBLP:journals/jal/GandhiKS04}. As each set (a vertex with its $MS$ in our setting) has the same cost 1, we sort each vertex in increasing order of their $ECO+ECI$ at the sorting step of their algorithm. During the process, we record the cost $cost_u$ of choosing each candidate serving set $MS(u)$ to be the highest cost set. Return their output $V'_{co}$.

(2) Initialize the core vertex set $V_{co}=V-V_{de}$. Sort the vertices $\forall u\in V-V_{de}-V_{co}'$ according to $cost_u$ in decreasing order.  We check each vertex $u\in V-V_{de}-V_{co}'$ in order and remove it from $V_{co}$ if the constraint of $k$-skip shortest path is not violated.

(3) Sort the vertices $\forall u\in V_{co}'$ according to $cost_u$ in decreasing order. Pop each vertex $u\in V_{co}'$ in order and check whether the $MS$s of remaining vertices cover $\epsilon$ vertices and satisfy a $k$-skip cover over $G_c-V_{de}$. If so, remove it.

\noindent\textbf{Time Complexity}. Lines 1-4 construct $MS(\cdot)$ in $O(\abs{V}\cdot nc_m)$ time. Line 5 costs time as same as the solution in \cite{DBLP:journals/jal/GandhiKS04} costs, which is $O(\abs{V}^2)$. In lines 6 to 8, checking k-skip cover constraint costs $O(\bar{\sigma}_{k-1}\log\bar{\sigma}_{k-1})$, where $\bar{\sigma}_{k-1}$ is the average number of $(k-1)$-hop neighbors of the vertices in $V$, which grows linearly with $\abs{V}$ \cite{DBLP:journals/pvldb/FunkeNS14}. Using a hashset to record a copy of $MC$ can help us check the servable vertices in $O(1)$ time. With an $O(\abs{V})$ loop, lines 6 to 8 cost $O(\abs{V}^2\log{\abs{V}})$ time. So the overall time complexity is $O(\abs{V}^2\log{\abs{V}})$.

\section{Proof of Lemma \ref{proof5}}
\label{A:proof5}
\noindent\textbf{Lemma \uppercase\expandafter{\romannumeral5}.3.} 
\emph{Assume that we have $N$ vertices in total, with $M$ set as optimal 
	solution for the attribute \emph{(ii)}, the upper bound of the size of core 
	vertex set is $\sigma(k)=\max(\frac{N}{k}log\frac{N}{k}, nc_m\cdot M)$.}
\begin{proof}\vspace{-1.5ex}
	We prove it by dividing all the cases into two types: 
	
	(1) $\bm{nc_m\cdot M\geq\frac{N}{k}log\frac{N}{k}}$. Step (1) returns 
	$V'_{co}$ with size $\abs{V'_{co}}\leq nc_m\cdot M$. 
	So after step (2), if the size of remaining vertices is larger than 
	$nc_m\cdot M$, there must be some vertices $u\in V-V_{de}-V_{co}'$ left.
	
	As \cite{DBLP:conf/sigmod/TaoSP11} shows that any subset of $V$ with size 
	at least $\frac{N}{k}log\frac{N}{k}$ is a k-skip cover of $V$, 
	step (2) can still prune some vertices $u\in V-V_{de}-V_{co}'$ without 
	violating attribute \emph{(i)}, which is a contradiction. Thus, at most 
	$nc_m\cdot M$ vertices are left after step (2).
	As the remaining vertices including all the $V'_{co}$, which is a valid 
	cover, the size of the final output is no larger than $nc_m\cdot M$.
	
	(2) $\bm{nc_m\cdot M<\frac{N}{k}log\frac{N}{k}}$. 
	Step (1) returns $V'_{co}$ with size $\abs{V'_{co}}\leq nc_m\cdot M< 
	\frac{N}{k}log\frac{N}{k}$. 
	So after step (2), if the size of the remaining vertices is larger than 
	$\frac{N}{k}log\frac{N}{k}$, there must be some vertices $u\in 
	V-V_{de}-V_{co}'$ left. This also implies that step (2) can still prune 
	some vertices $u\in V-V_{de}-V_{co}'$ without violating the $k$-skip cover, 
	which results in a contradiction. So at most $\frac{N}{k}log\frac{N}{k}$ 
	vertices are left after step (2).
	
	The step (3) will maintain the valid cover and reduce the size. So the 
	final output size is no larger than $\frac{N}{k}log\frac{N}{k}$.
\end{proof}\vspace{-2ex}

\section{Proof of Lemma \ref{proof6}}
\label{A:proof6}
\noindent\textbf{Lemma \uppercase\expandafter{\romannumeral6}.2.} 
\emph{Pruning Algorithm \ref{algo:SMDB} has no performance loss..}

\begin{proof}
	Line 9 in Algorithm \ref{algo:SMDB} guarantees that starting from $v$ at 
	time $arv[v]$ to pick up $r_j$ 
	directly misses the deadline $tp_j$. So inserting pick-up in latter indexes 
	can be viewed as, starting from $v$ to pick up $r_j$ with some detour, 
	which costs more. So the pruned insertion positions latter than $v$ are not 
	insertable in the original algorithm. Thus, our algorithm has no 
	performance loss.
\end{proof}\vspace{-2ex}

\section{Parameter Settings}
\label{A:set}

In real-application, the penalty of a rejection can be treated mainly as the 
money loss in proportion to the length of the tour (i.e. $p_j = p_o \times 
SP_c(s_j,e_j)$). {\label{rw:r1-1}The penalty weight $p_o$ is usually greatly larger than the 
	weight for travel cost $\alpha$. This guarantees no request will be rejected if 
	a feasible new route can serve it. Different $p_o$ neither changes the 
	assignment result nor affects the serving rate. }{ To make our evaluation solid, we decrease the penalty until it affects the serving results, where $p_o$ is down to $3$. The served requests and unified costs under varied $p_o$ are shown in Figure~\ref{fig:exp_pen}. The serving rate is nearly constant and the cost is linear to $p_o$.}

\begin{figure}[h!]\centering \vspace{-4ex}
	\centering
	\subfigure[\scriptsize Served requests ($p_o$)]{
		\includegraphics[width=0.47\linewidth]{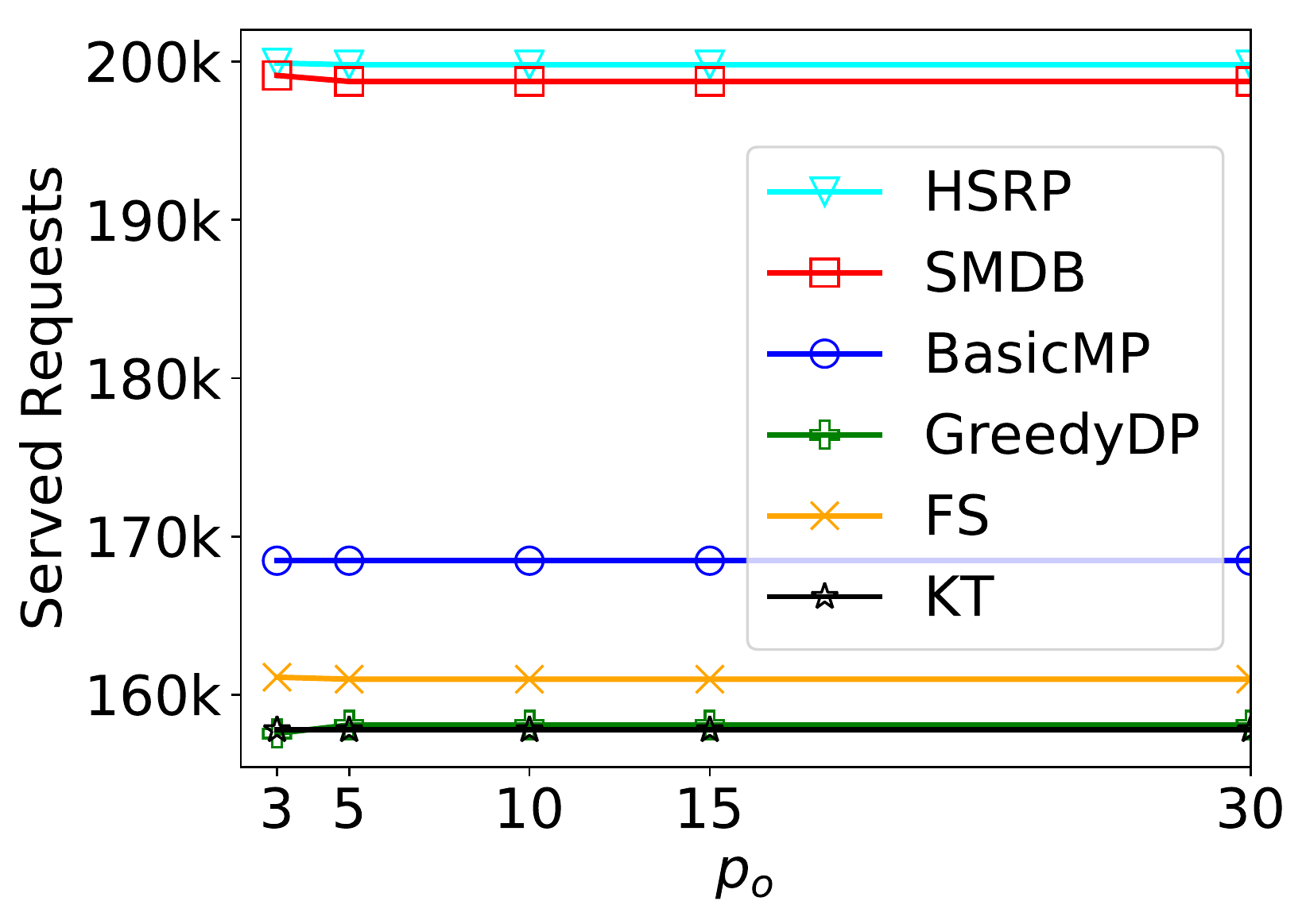}\label{subfig:sr_pen}
	}
	\subfigure[\scriptsize Unified cost ($p_o$))]{
		\includegraphics[width=0.47\linewidth]{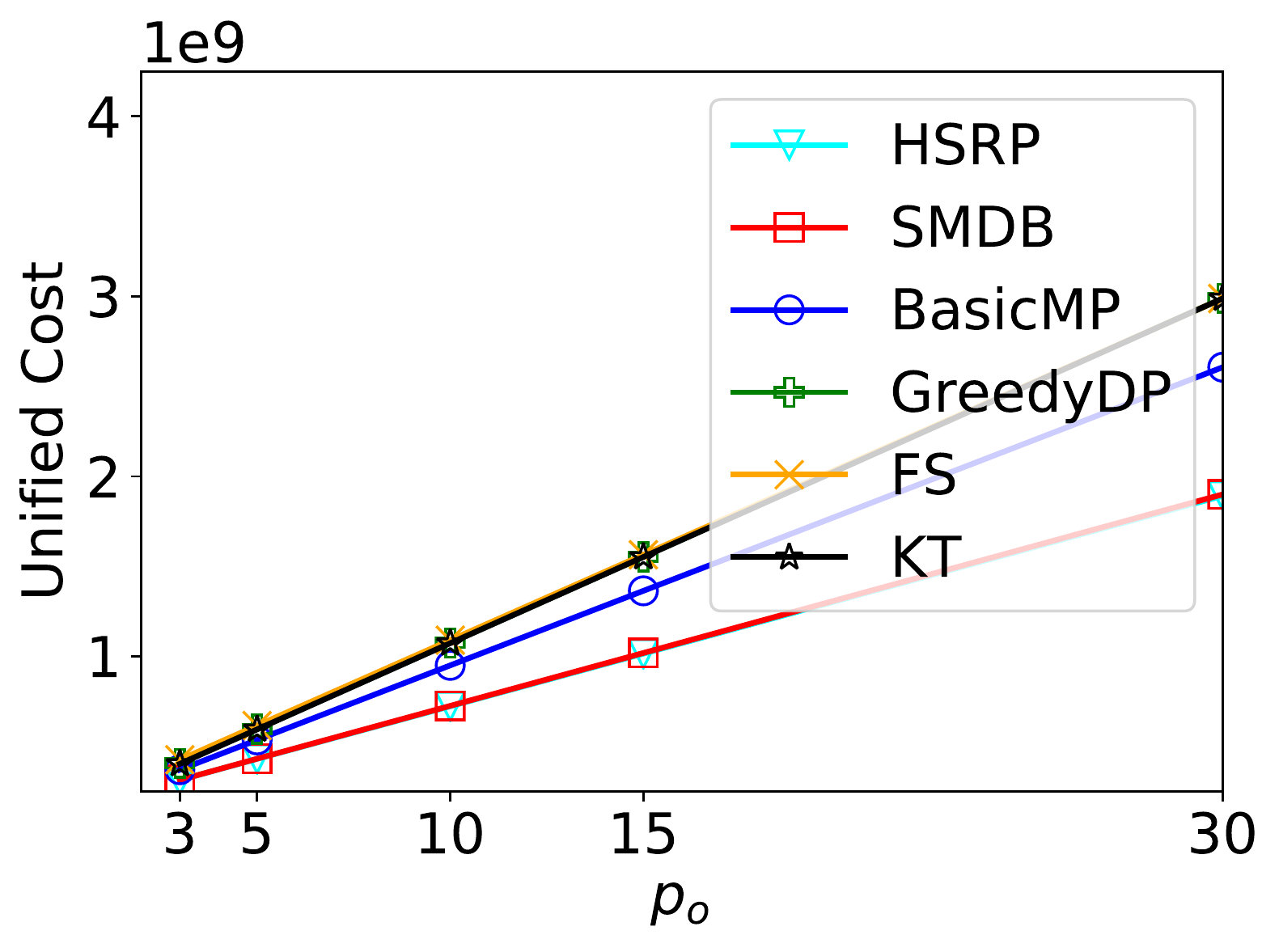}\label{subfig:uc_pen}
	}
	\caption{Performance of varying $p_o$}\label{fig:exp_pen}
\end{figure}\vspace{1ex}

\begin{figure}[h!]\centering \vspace{-4ex}
	\centering
	\subfigure[\scriptsize Served requests ($\beta$)]{
		\includegraphics[width=0.47\linewidth]{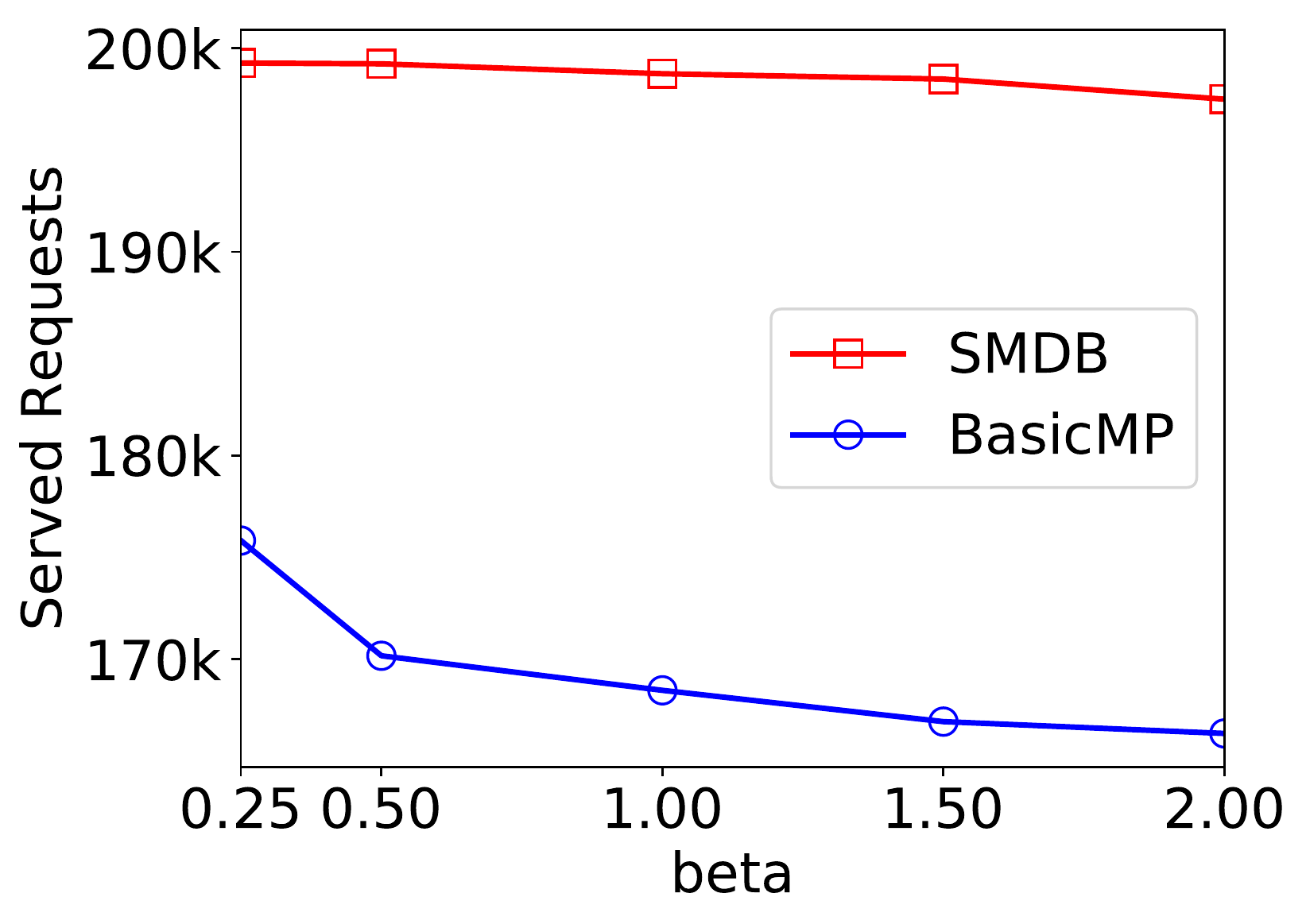}\label{subfig:sr_beta}
	}
	\subfigure[\scriptsize Unified cost ($\beta$)]{
		\includegraphics[width=0.47\linewidth]{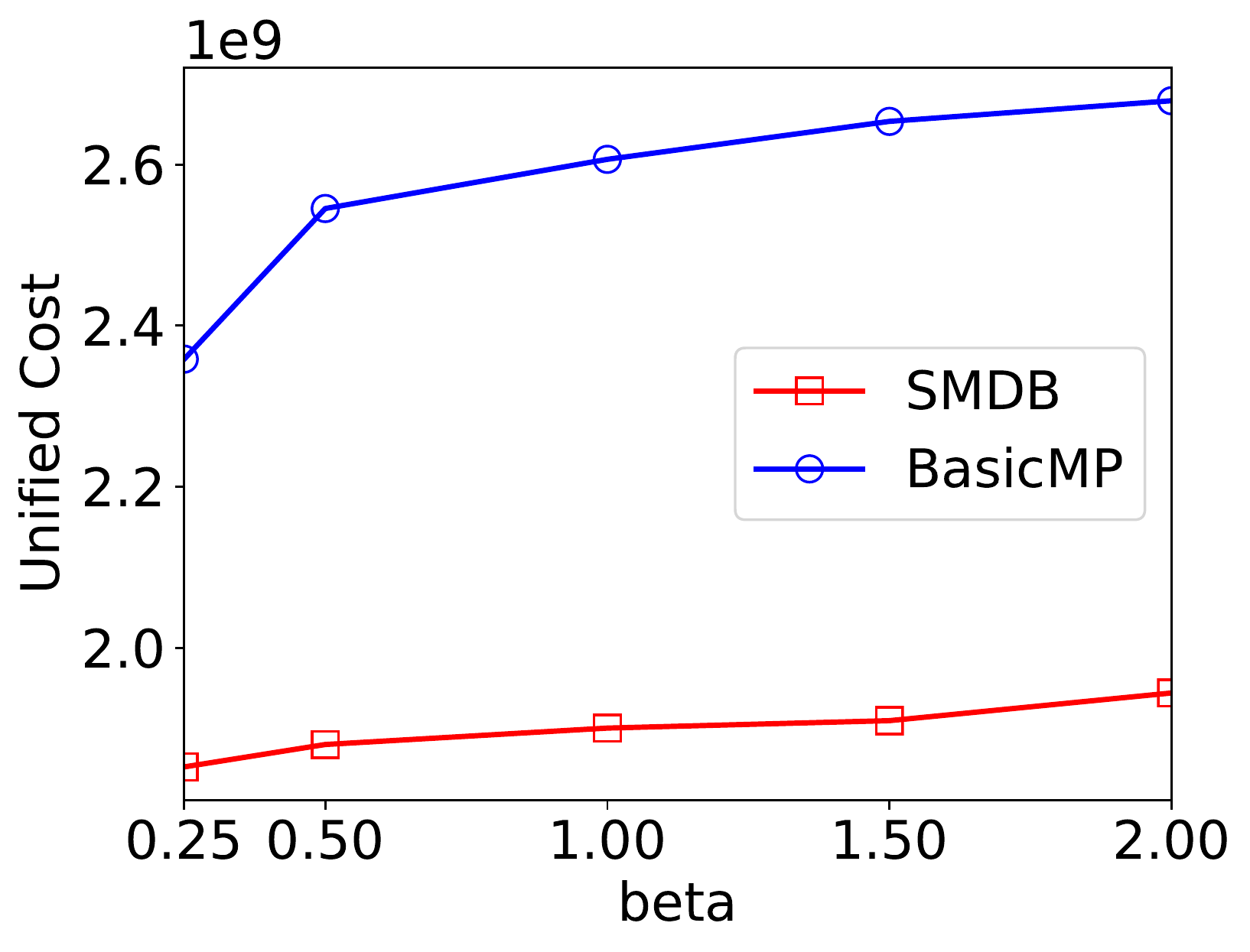}\label{subfig:uc_beta}
	}
	\caption{Performance of varying $\beta$}\label{fig:exp_beta}
\end{figure}\vspace{1ex}

\begin{figure*}[t!]\centering \vspace{-4pt}
	\centering
	\subfigure[\scriptsize Served requests ($n_r$)]{
		\includegraphics[width=0.18\linewidth]{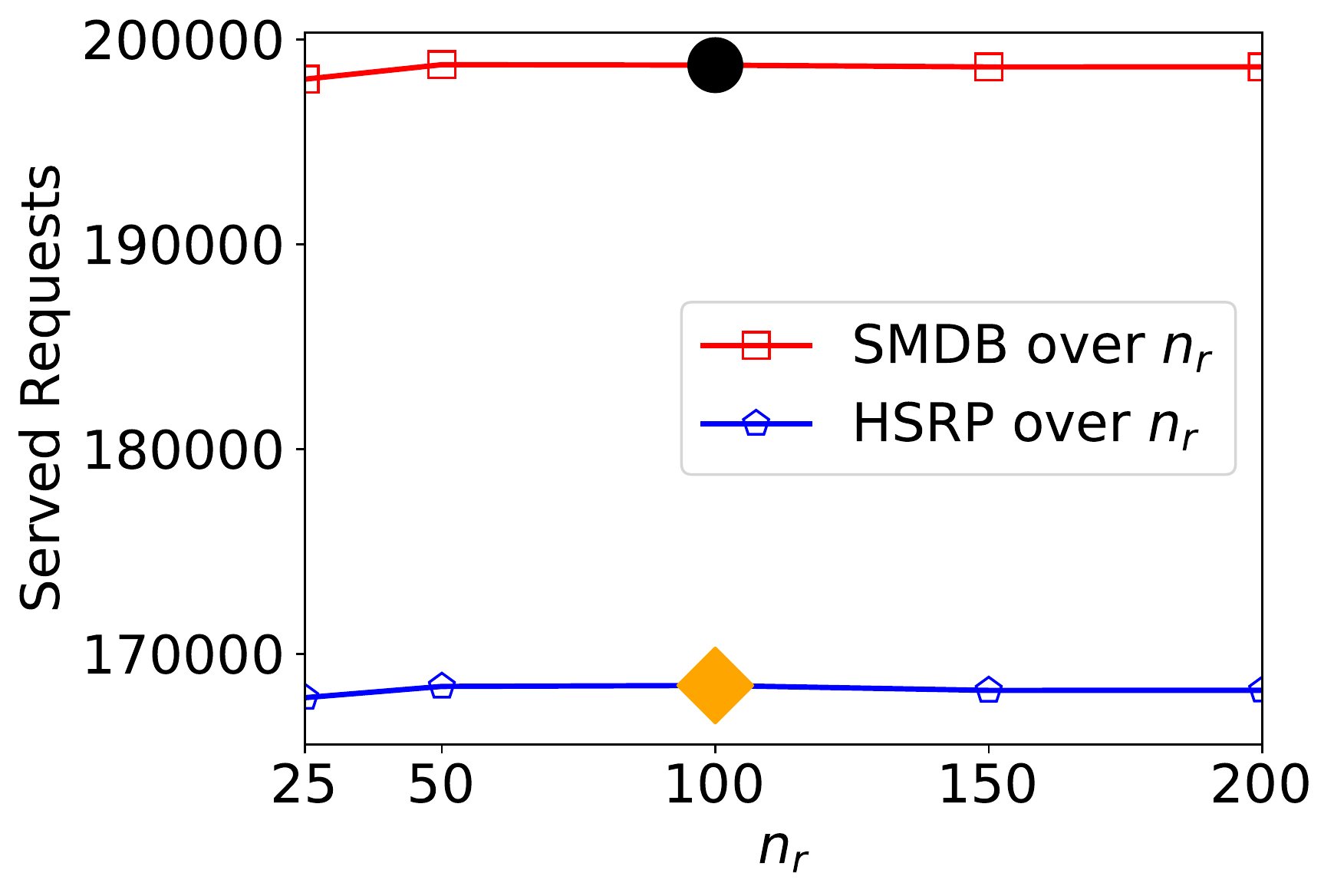}\label{subfig:serve_nr}
	}
	\subfigure[ \scriptsize Served requests ($d_m$)]{
		\includegraphics[width=0.18\linewidth]{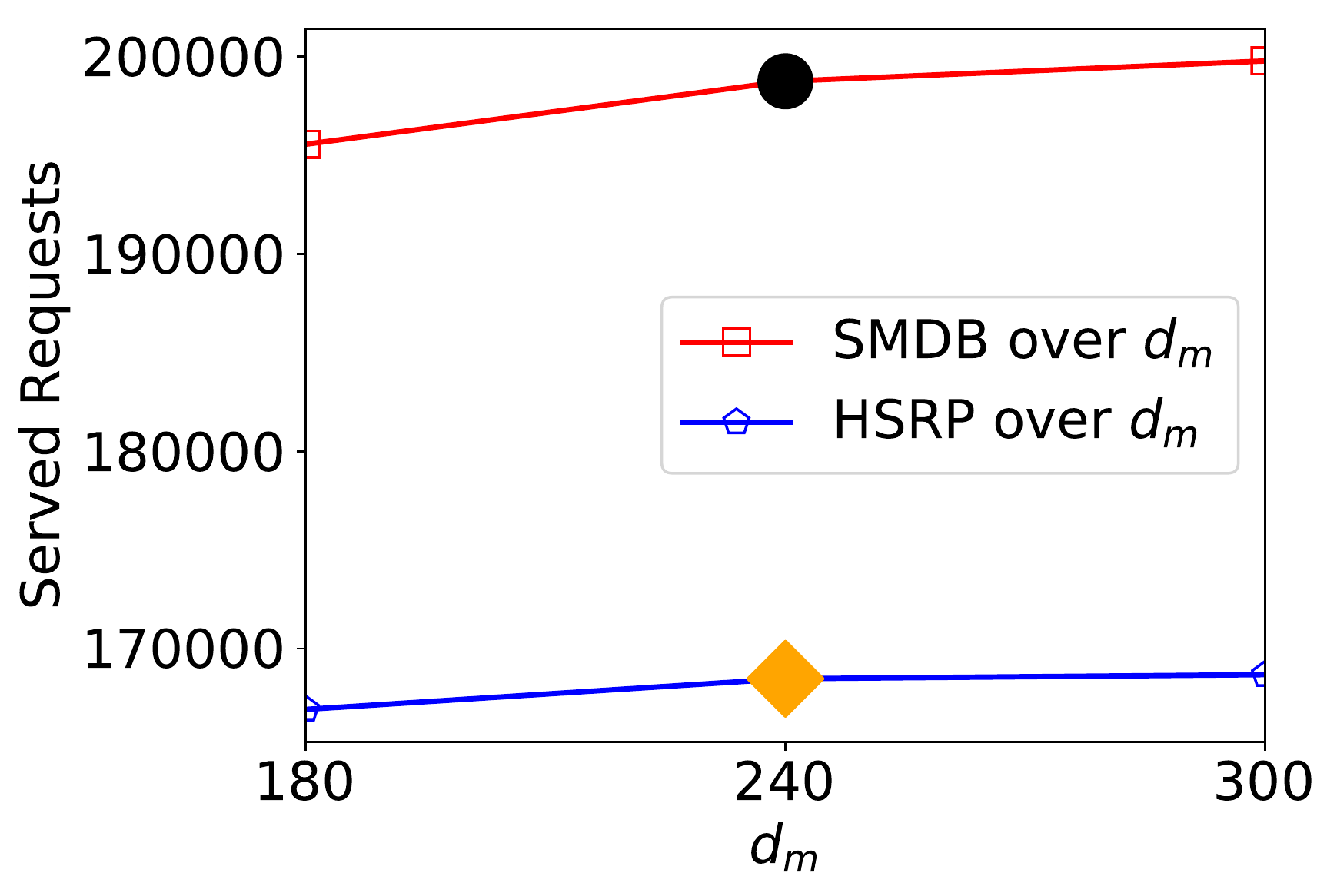}\label{subfig:serve_dm}
	}
	\subfigure[\scriptsize Served requests ($nc_m$)]{
		\includegraphics[width=0.18\linewidth]{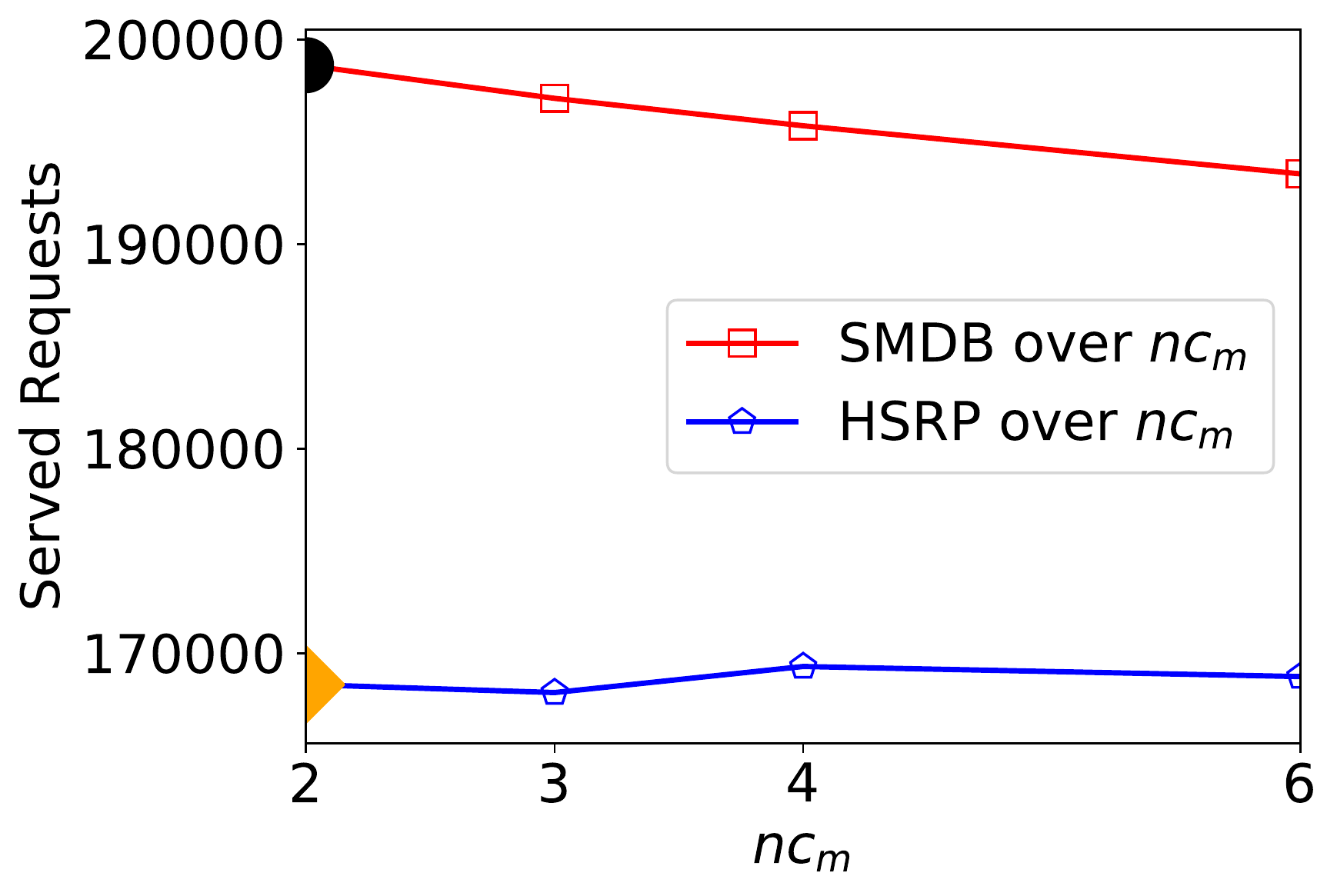}\label{subfig:serve_ncm}
	}
	\subfigure[ \scriptsize Served requests ($thr_{CS}$)]{
		\includegraphics[width=0.18\linewidth]{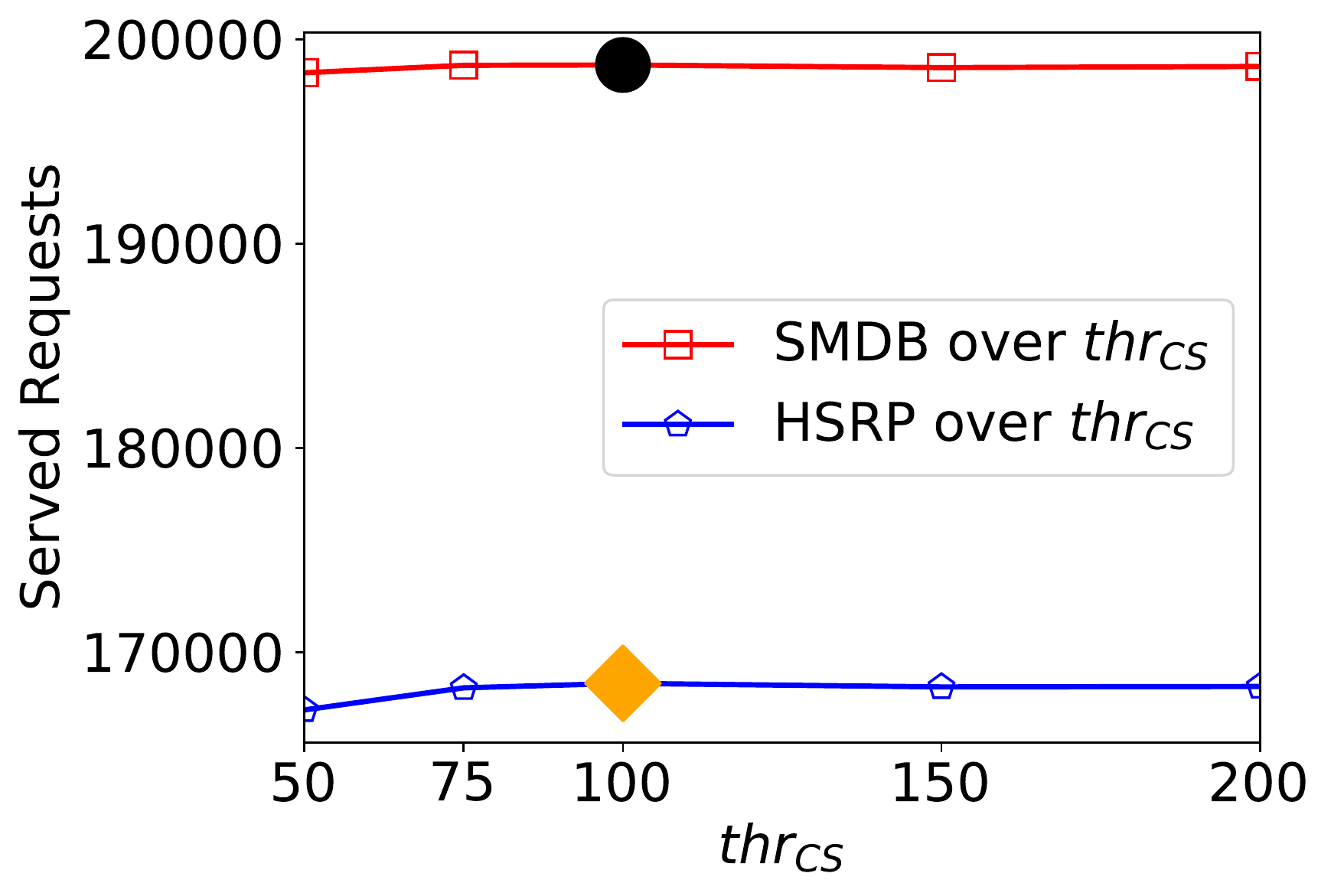}\label{subfig:serve_thr}
	}
	\subfigure[\scriptsize Response time ($\epsilon$)]{
		\includegraphics[width=0.18\linewidth]{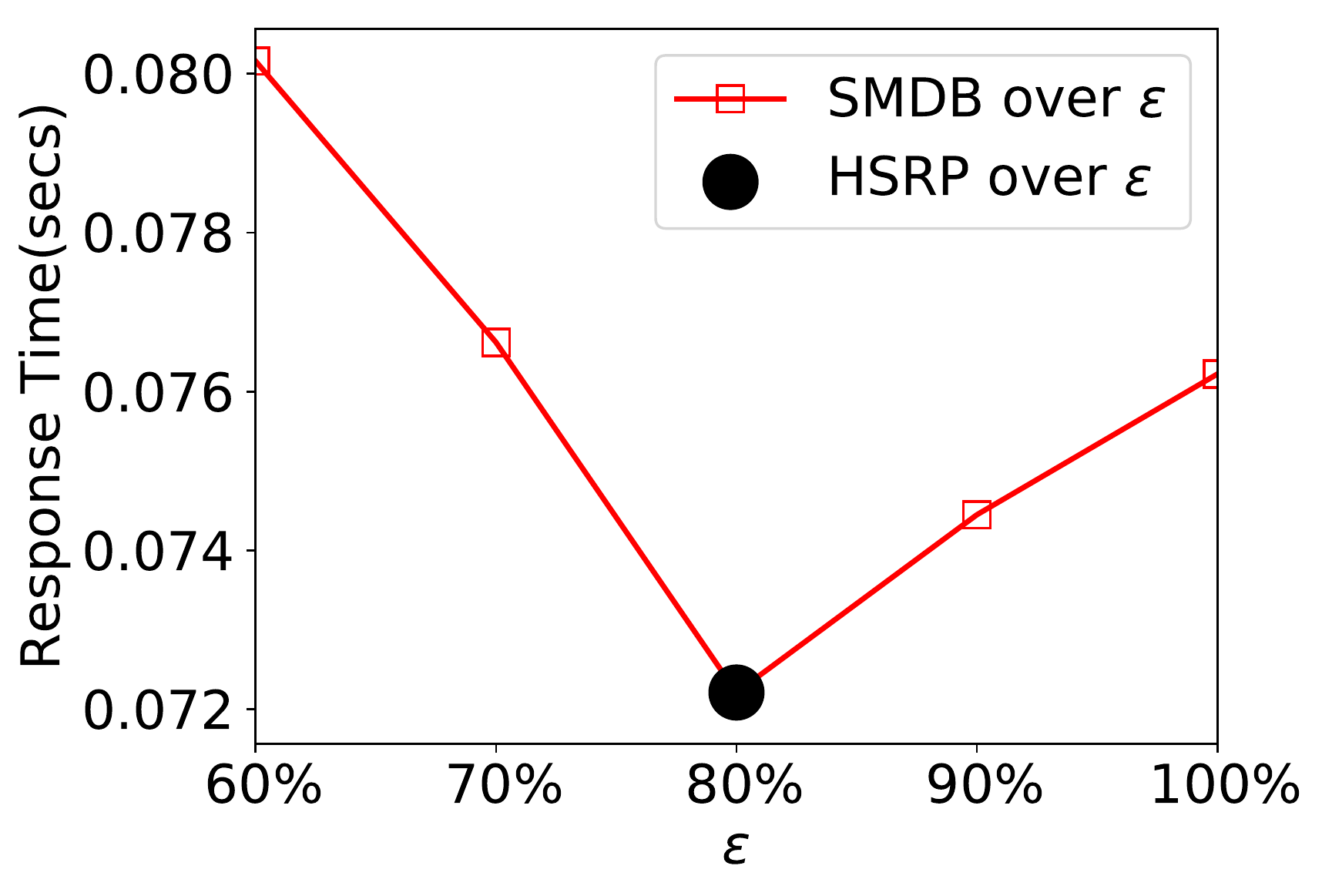}\label{subfig:time_eps}
	}
	\vspace{-2ex}
	
	\subfigure[\scriptsize Unified cost($n_r$)]{
		\includegraphics[width=0.18\linewidth]{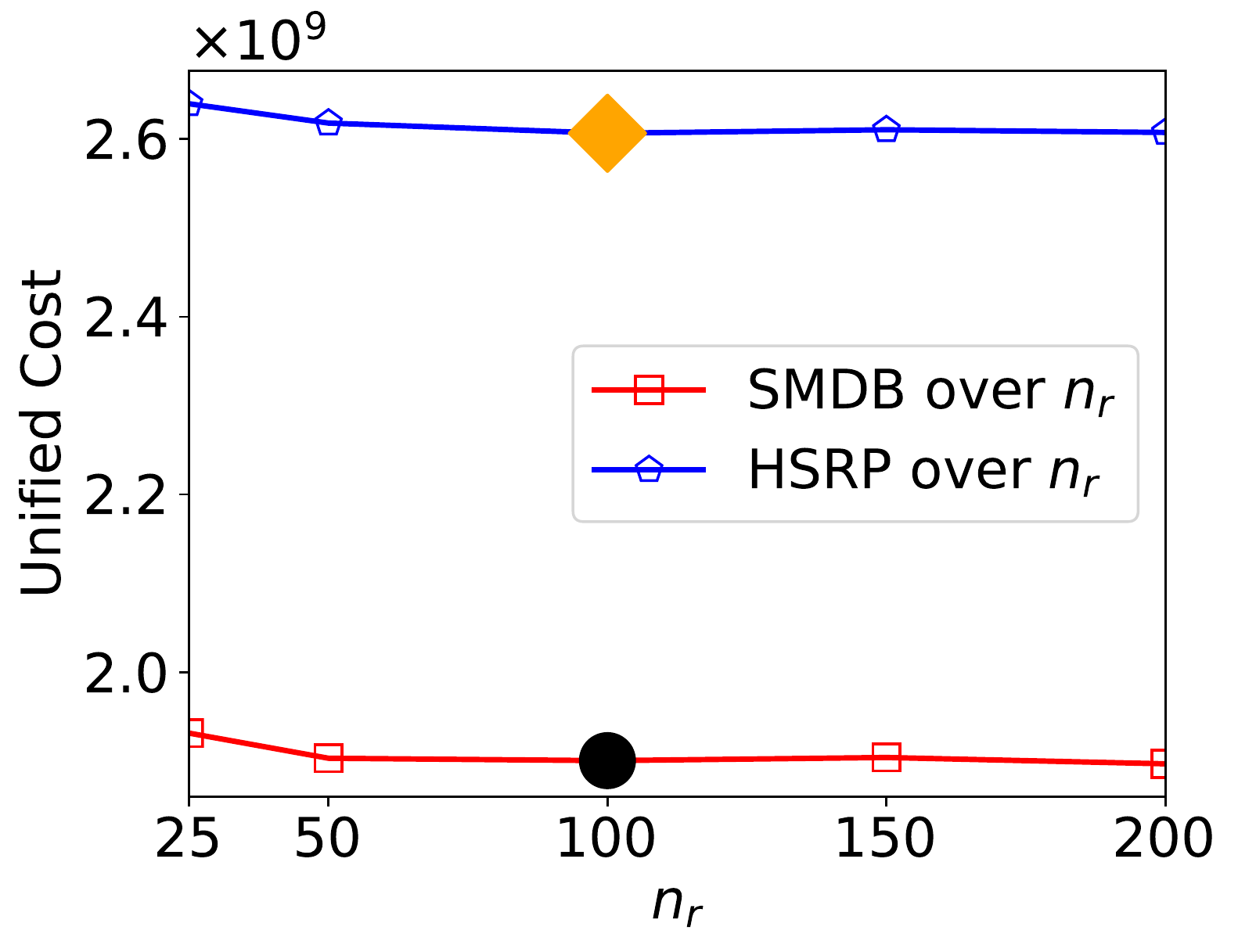}\label{subfig:cost_nr}
	}
	\subfigure[ \scriptsize Unified cost ($d_m$)]{
		\includegraphics[width=0.18\linewidth]{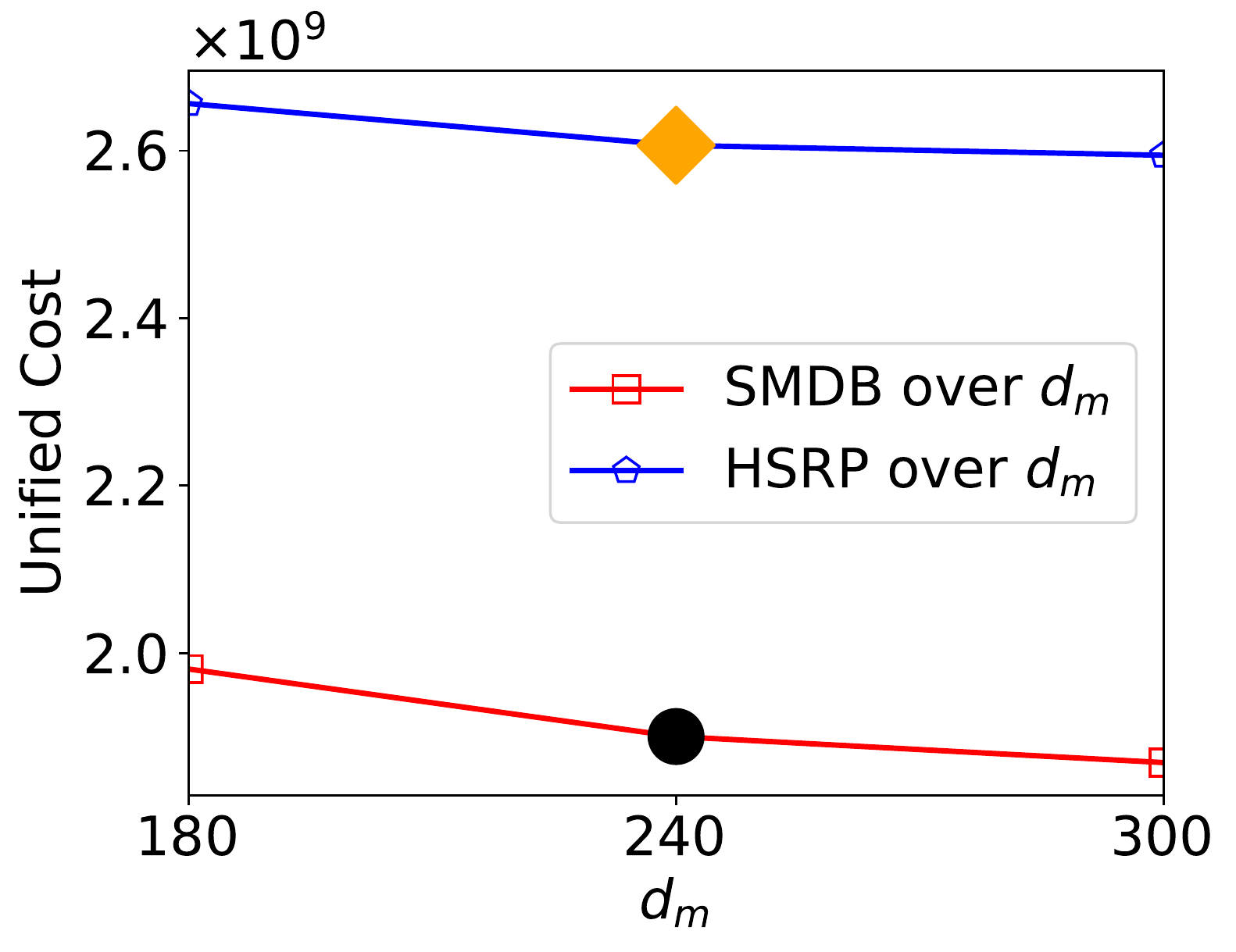}\label{subfig:cost_dm}
	}
	\subfigure[\scriptsize Unified cost ($nc_m$)]{
		\includegraphics[width=0.18\linewidth]{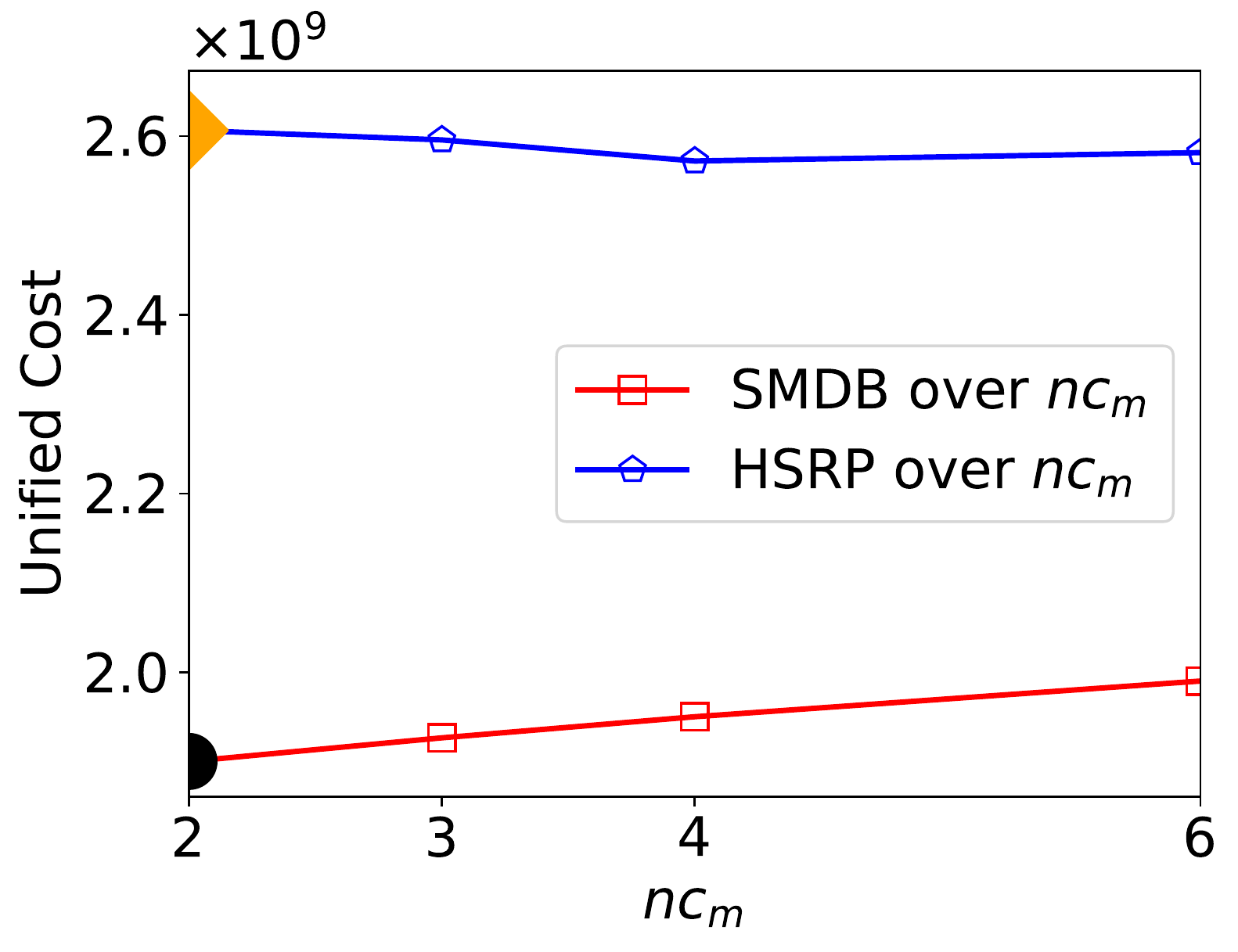}\label{subfig:cost_ncm}
	}
	\subfigure[ \scriptsize Unified cost ($thr_{CS}$)]{
		\includegraphics[width=0.18\linewidth]{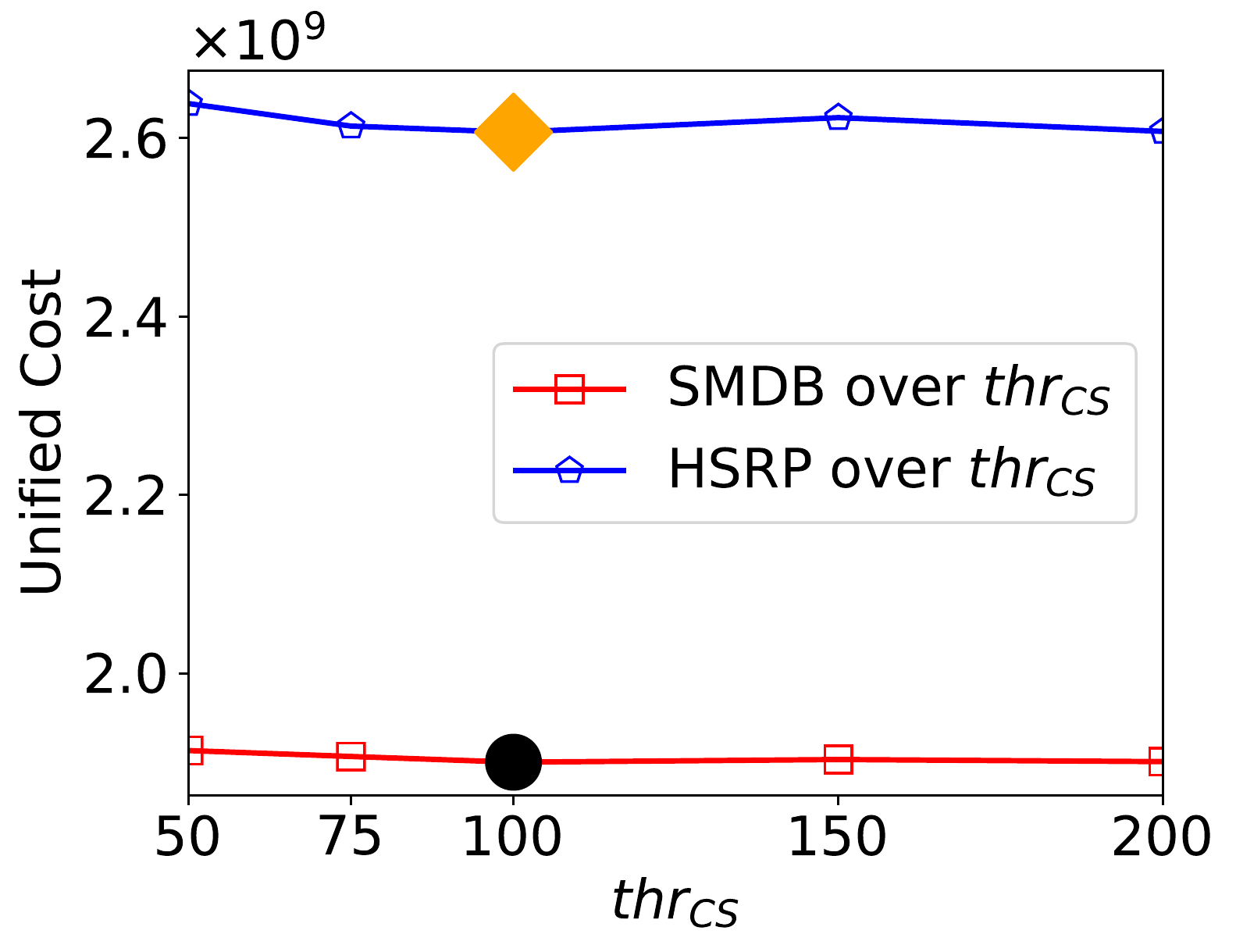}\label{subfig:cost_thr}
	}
	\subfigure[ \scriptsize Response time ($k$)]{
		\includegraphics[width=0.18\linewidth]{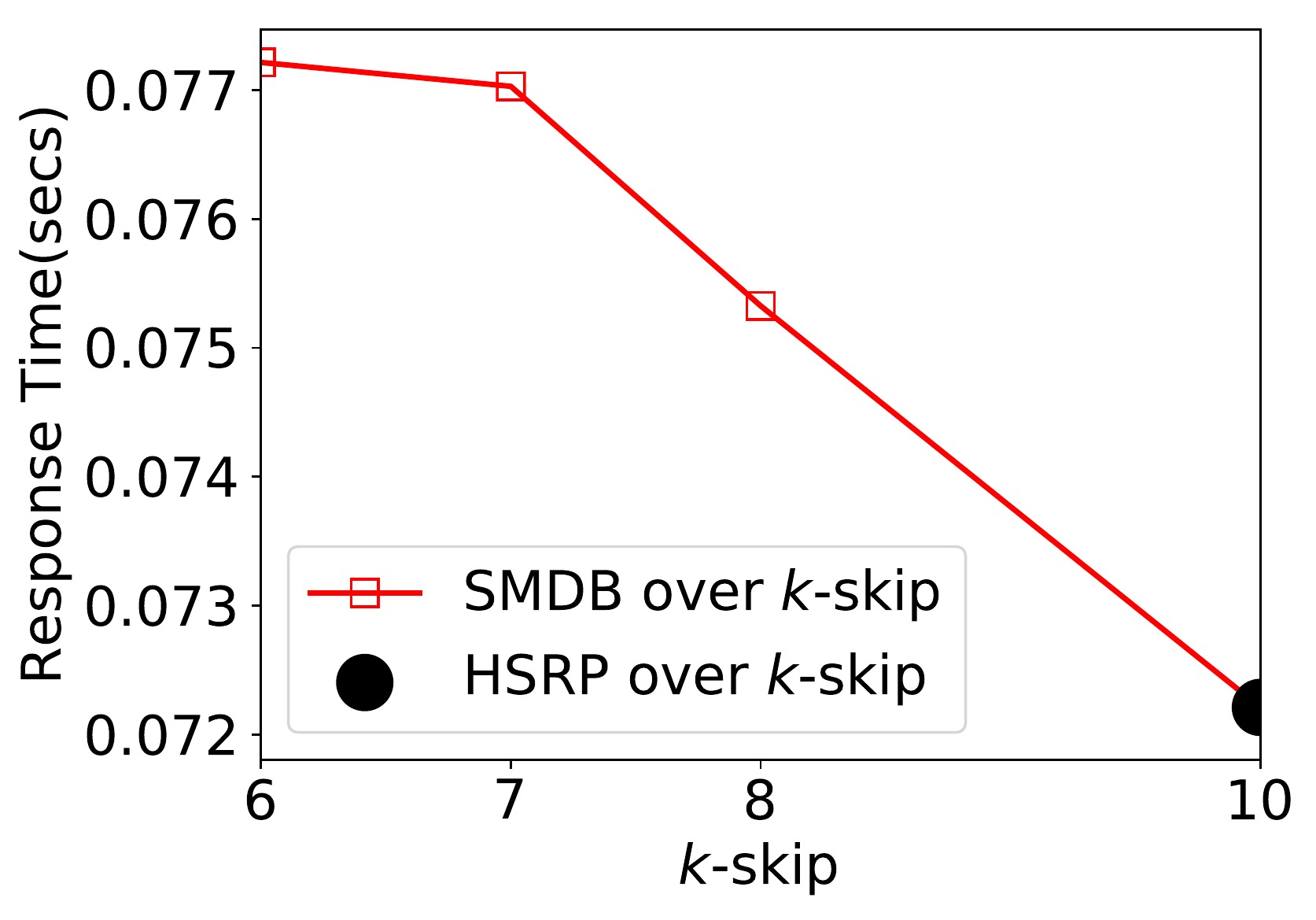}\label{subfig:time_k}
	}
	\vspace{-1ex}
	
	\caption{\small Performance of varying maximum walking distance $d_m$, 
		number of reference vertices $n_r$, maximum number of candidates $nc_m$, 
		the threshold $thr_{CS}$, $\epsilon$, and $k$ for the $k$-skip 
		cover.\label{fig:tune_result}}
\end{figure*}
{\color{red}
	\begin{table}[h!]
		\centering
		{\small 
			\caption{\small Parameter Settings.}\vspace{-3ex} \label{table_hyper}
			\begin{tabular}{c|c}
				\hline
				Parameters &Settings \\
				\hline
				number of reference vertices  $n_r$ &25, 50, \textbf{100}, 150, 200\\
				Maximum walking distance of riders $d_m$ &180, \textbf{240}, 300\\
				Number of MP candidates $nc_m$ &\textbf{2}, 3, 4, 6\\
				Threshold for MP candidate prunning $thr_{CS}$ &50, 75, \textbf{100}, 150, 200\\
				Percent of vertices to cover by core vertices $\epsilon$  &60\%, 70\%, \textbf{80\%}, 90\%, 100\%\\
				$k$ of the constructed $k$-skip cover $k$ &6, 7, 8, \textbf{10}\\
				\hline
			\end{tabular}
		}\vspace{-2ex}
	\end{table}
}
We set the delivery deadline of each request as the sum of its release time and 
the shortest time from source to destination extended by a Deadline Coefficient 
$e_r$. For example, the default deadline for a
request with release time $tr_j$ is 
$tr_j +(1+e_r)\cdot SP_c(s_j,e_j)$.
We set the deadline for pick-up as the latest time, that is, $tr_j$ minus the 
shortest time required to finish it after pick-up.
$a_j$ is varied from 2 to 10. { \label{rw:r1-2}The platform pays a driver for its travel time. 
	We set the unit travel fee as the unit cost. (i.e. $\alpha=1$). Walking 
	time results in a discount for riders. The unit walking time corresponds to 
	$\beta$ unit cost, which varies from 0.5 to 2. Its default value is set to 1, 
	which ensures that using meeting points should not increase the travel time of 
	the rider.}

Figure~\ref{fig:exp_beta} displays the effect of $\beta$. We propose 
$BasicMP$, which adapts the state-of-art traditional ridesharing solution to 
fit meeting points mode, as a baseline to compare with $SMDB$. We use the 
default setting in Table~\ref{table1} to test them. Larger $\beta$ leads to a 
larger cost of walking, thus decreases the flexibility. As $\beta$ increases 
from 0.25 to 2.0, both of the two algorithms serve fewer requests and cost 
higher. However, $BasicMP$ works worse and serves 12.2\% fewer requests while 
$SMDB$ serves 7.8\% fewer. 
The reason is that $SMDB$ exploit the benefit of flexibility by arranging 
routes with convenient vertices. $BasicMP$ only cares about the temporal cost 
and loses a lot.

\begin{table}
	\centering 
	{\small \scriptsize
		\caption{\small Symbols and Descriptions.}
		\label{table:notation}
		\begin{tabular}{c|c}
			\hline
			Notation &Description \\
			\hline
			\multicolumn{2}{c}{Problem Definition}\\
			\hline
			$G_c=\langle V_c, E_c\rangle$& road network for car with vertices and edges\\
			$G_p=\langle V_p, E_p\rangle$& road network for passengers with vertices and edges\\
			$G=\langle V, E\rangle$& union of road network for car and passengers with vertices and edges\\
			${SP_c}(u,v)$ & shortest travel cost on graph $G_c$ for car from vertex $u$ to $v$\\
			${SP_p}(u,v)$ &shortest travel cost on graph $G_p$ for passenger from vertex $u$ to $v$\\
			$W$ &a set od drivers\\
			$w_i$ & a driver of index $i$\\
			$l_{i}$ &current location of driver $w_i$\\
			$R$ &a set of requests \\
			$r_j$ & a request of index $j$\\
			$s_j$ &source location of request $r_j$\\
			$e_j$ &destination of request $r_j$\\
			$p_{j}$ &penalty of unserved request $r_j$\\
			$tr_j$ &release time of request $r_j$\\
			$tp_j$ &deadline to pick up a request $r_j$\\
			$td_j$ &deadline for a request to reach its destination $r_j$\\
			$pi_j$ &pick-up point of request $r_j$\\
			$de_j$ &drop-off point to drop request $r_j$\\
			$wp_j$ &time for request $r_j$ to walk from its source to pick-up point \\
			$wd_j$ &time for request $r_j$ to walk from its drop-off point to destination\\
			$\hat{R}, \bar{R}$ & the set of served and rejected requests\\
			$S_{w_i}, D(S_{w_i})$ &schedule of driver $w_i$ and its distance\\
			$a_i$ or $a_j$ &capacity of driver $w_i$ or request $r_j$ \\
			$\alpha$ &weight parameter for unit driving distance in $S_w$\\
			$\beta$ &weight parameter for cost from walking distance\\
			\hline
			\multicolumn{2}{c}{Meeting Point Selection}\\
			\hline
			$MC(u)$ &meeting point candidates of vertex $u$\\
			\emph{ECI}$(\cdot)$ or \emph{ECO}$(\cdot)$ & equivalent in or out cost for a vertex\\
			$n_r$ & number of reference vertices that are closest to a given vertex\\
			$n_o(u)$& reference vertices of vertex $u$\\
			$d_m$ &maximum walking distance\\
			$SCS(u, v)$ &serving cost score to serve $u$ with $v$ as MP\\
			$thr_{CS}$ &threshold to prune MP according to $SCS$\\
			\hline
			\multicolumn{2}{c}{HMPO Graph Construction}\\
			\hline
			$V_{co}, V_{su}, V_{de}$& Core, sub-level, defective vertices\\
			$\epsilon$&the proportion of vertices that should be servable by core vertices as MPs\\
			$MS(u)$ & candidate serving set of $u$ where $u$ is the candidate MP of all vertices in the set\\
			$G_h=\langle V, E_h\rangle$& road network of our hierarchical graph with vertices and edges\\
			$E_{cc}$&super edges from core to core vertex\\
			$E_{cs}$&super edges from core to sub-level vertex\\
			$E_{sc}$&super edges from sub-level to core vertex\\			
			$E_{ss}$&super edges from sub-level to sub-level vertex\\
			\hline
			\multicolumn{2}{c}{HMPO Graph Based Insertor}\\
			\hline
			$MD(u, v)$&the Maximum Difference for vertices $u$ and $v$\\
			$Ch(u)$&checker vertex for the MP candidate set of vertex $u$\\
			$SMD(u)$& the set Maximum Difference for vertex $u$ with regard to its MP candidate set\\
			$VC(u)$&the core vertices that has super edges with vertices in $MC(u)$\\
			$LMD(v)$& the Local Maximum Difference given a MP candidate set $MP(u)$\\
			$DV$& the pruned vertices based on SMDB\\
			$arv[u]$&arriving time for a route vertex $u$\\
			
			\hline
		\end{tabular}
	}\vspace{-2ex}
\end{table}

Besides, we conduct extensive experiments to compare the impact of different 
parameters for meeting point candidate selection and HMPO graph construction. 
Note that the construction is offline and its time cost would not affect the 
online assignment. Here we display the results of different factors for meeting 
point candidate selection: maximum walking distance $d_m=[180, 240, 300]$; 
number of reference vertices $n_r=[25, 50, 100, 150]$; maximum number of 
candidates $nc_m=[2, 3, 4, 6]$; and the threshold $thr_{CS}=[0, 50, 100, 200]$. 
To construct the HMPO Graph, we compare $\epsilon=[40\%, 60\%, 80\%, 100\%]$ 
and $k=[5, 8, 10, 15]$ for the $k$-skip cover. The generated candidates and 
HMPO graphs are applied to our $SMDB$ algorithm with the default setting in 
Table\ref{table1}. {The notations and settings are displayed in Table~\ref{table_hyper}.} We show their performances in Figure~\ref{fig:tune_result}. 

According to the results, we choose the best setting: $n_r=100$, $nc_m=2$, 
$thr_{CS}=100$, $\epsilon=80\%$, and $10$-skip cover. All the chosen settings 
results in the lowest unified cost and highest serving rate. As for 
the maximum walking distance, a larger $d_m$ always has better performance 
(more flexible choices for meeting points) but impairs the users’ experience 
(walking farther). To increase $d_m$ from $240$ to $300$ (25\%), the number of 
served requests only increase $0.5\%$. Here we choose $d_m=240$ as a trade-off.

{\label{rw:r2-2}To avoid tedious fine-tuning, we give suggestions for default settings. $n_r$ is supposed to be the number of nodes that is accessible within the average time cost of historical requests. Walking distance $d_m$ is user-specified according to their tolerance for walking, where threshold $thr_{cs}$ could be set to about half of $d_m$ as the average walking distance. As vertices in city networks usually have edges no more than 4, we can set the number of candidate $nc_m$ to $2$ or $3$. $\epsilon$ is chosen by selecting vertices according to their \emph{ECI/ECO} and collecting the vertices they can serve (i.e., $MS$) into a set $SS$. The process stops when the size of $SS$ changes slowly while selecting new vertices. We set $\epsilon=|SS|/|V_p|$.  $k$ should be the average length of shortest paths from historical requests.}

\begin{figure}[h]\centering 
	\includegraphics[width=0.47\linewidth]{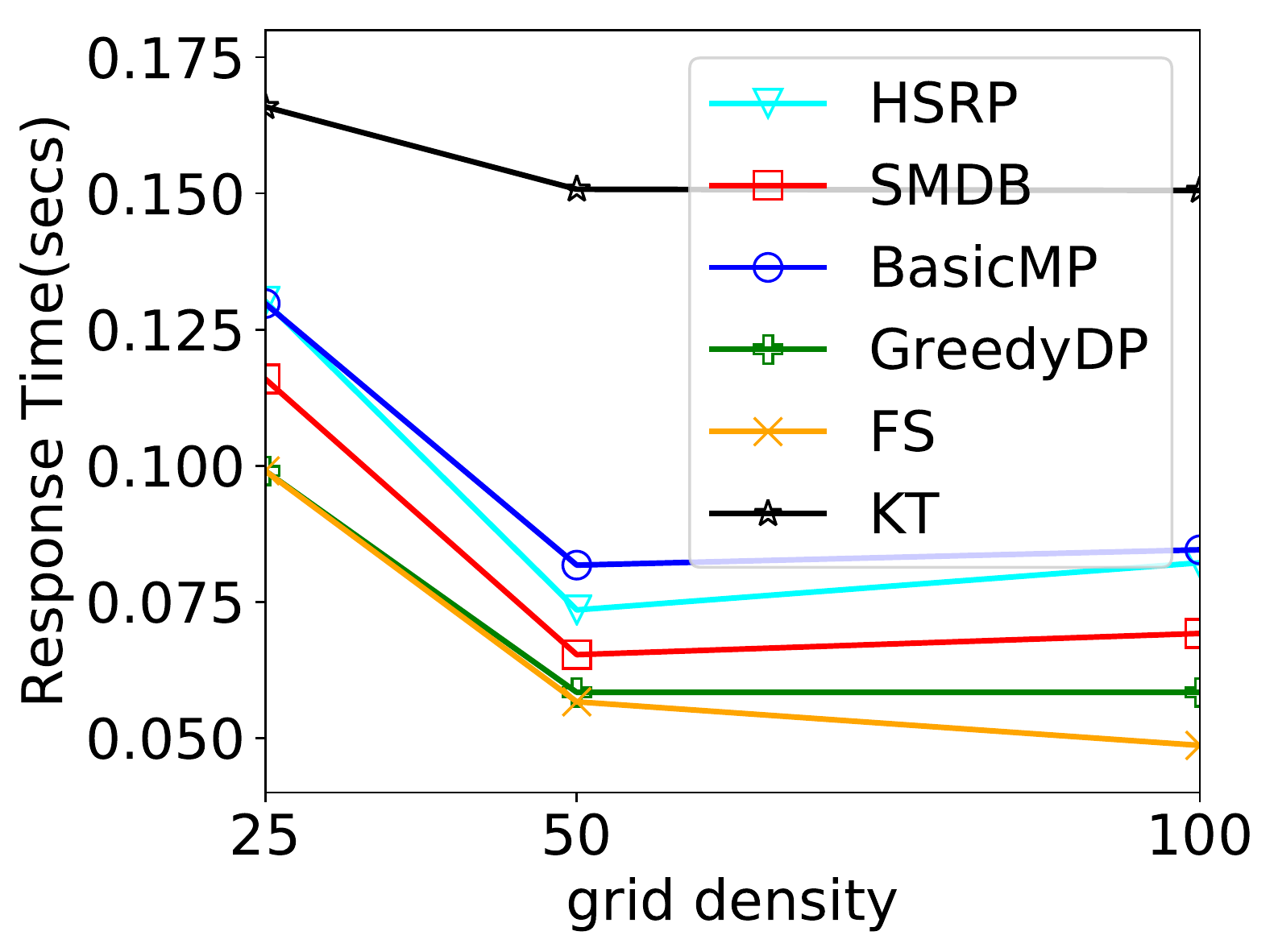}
	\caption{\small Time cost of varying grid density}
	\label{fig:exp_grid}
\end{figure}\vspace{1ex}

{ We use grids to prune workers during assignment. The grids are cut according to lantitude and longitude. We set the number of grid per degree (grid density) to 50 by default, e.g., an area of $1^\circ\times 1^\circ$ has $50\times50$ grids. We vary the grid density from 25 to 100, which only affect the efficiency without changing the effectiveness. The result is shown in Figure~\ref{fig:exp_grid}.}

{
	\section{Notation Summary}
	\label{A:notation}

	In Table~\ref{table:notation}, we summarize all the major notations.}

\end{document}